\documentclass[11pt,a4paper]{article}
\pdfoutput=1
\usepackage{jheppub}
\usepackage[dvipsnames]{xcolor}
\usepackage{caption,subcaption,comment}
\usepackage{float,tabularx,booktabs,array}
\usepackage{graphicx,amsmath,amsthm,dcolumn,hyperref,verbatim,amsfonts,tikz,amssymb,mathtools,tikz-cd,upgreek,pgfplots,mathrsfs,textcomp}
\usepackage[utf8]{inputenc}
\usepackage[capitalize]{cleveref}
\pgfplotsset{compat=1.5}
\newtheorem{theorem}{Theorem}[section]
\newtheorem{lemma}[theorem]{Lemma}
\newtheorem{corollary}[theorem]{Corollary}
\newtheorem{proposition}[theorem]{Proposition}
\newtheorem{remark}{Remark}[section]
\usepackage[most]{tcolorbox}
\usepackage{enumitem} 
\title{Generalised Entanglement Entropies from Unit-Invariant Singular Value Decomposition}
\author[a,b,c]{Pawe{\l} Caputa,}
\author[c]{Abhigyan Saha,}
\author[c]{Piotr Su{\l}kowski}

\newcommand{\be}{\begin{equation}}
\newcommand{\ee}{\end{equation}}

\newcommand{\Tr}{\text{Tr}}
\newcommand{\N}{\mathcal{N}}

\newcommand{\op}{\mathcal{O}}

\newcommand{\ket}[1]{\left|#1\right\rangle}
\newcommand{\bra}[1]{\left\langle#1\right|}

\newcommand{\ba}{\begin{eqnarray}}
\newcommand{\ea}{\end{eqnarray}}

\newcommand{\diag}{\mathrm{diag}}

\newcommand{\Prob}{\mathbb{P}}
\newcommand{\E}{\mathbb{E}}
\newcommand{\Var}{\mathrm{Var}}

\preprint{\makebox[\textwidth][s]{
\hfill YITP-25-193}}
\affiliation[a]{The Oskar Klein Centre and Department of Physics, Stockholm University, AlbaNova, 106 91 Stockholm, Sweden}
\affiliation[b]{Yukawa Institute for Theoretical Physics, Kyoto University, Kitashirakawa Oiwakecho, Sakyo-ku, Kyoto 606-8502, Japan}
\affiliation[c]{Faculty of Physics, University of Warsaw, Pasteura 5, 02-093 Warsaw, Poland}
\abstract{We introduce generalisations of von Neumann entanglement entropy that are invariant with respect to certain scale transformations. These constructions are based on the Unit-Invariant Singular Value Decomposition (UISVD) in its left-, right-, and bi-invariant incarnations, which are variations of the standard Singular Value Decomposition (SVD) that remain invariant under the corresponding class of diagonal transformations. These measures are naturally defined for non-Hermitian or rectangular operators and remain useful when the input and output spaces possess different dimensions or metric weights. We apply the UISVD entropy and discuss its advantages in the physically interesting framework of Biorthogonal Quantum Mechanics, whose important aspect is indeed the behaviour under scale transformations. Further, we illustrate features of UISVD-based entropies in other well-known setups, from simple quantum mechanical bipartite states to random matrices relevant to quantum chaos and holography, and in the context of Chern-Simons theory. In all cases, the UISVD yields stable, physically meaningful entropic spectra that are invariant under rescalings and normalisations.}
\begin{document}
\maketitle
\newpage
\section{Introduction}\label{sec_introduction}
Quantum entanglement has become one of the most powerful and unifying concepts in modern theoretical physics. It underlies our understanding of correlations and information in quantum many-body systems \cite{Amico:2007ag,PasqualeCalabrese_2009}, quantum field theories \cite{Srednicki:1993im,Holzhey:1994we,Calabrese:2004eu}, and even quantum gravity \cite{Ryu:2006bv,Takayanagi:2017knl,VanRaamsdonk:2010pw} through holography \cite{Maldacena:1997re}. Entanglement entropy and its generalisations have emerged as indispensable probes of quantum structure, diagnosing topological order, characterising quantum phases, and encoding the emergence of spacetime geometry \cite{Takayanagi:2025ula}. Still, despite its success, quantifying entanglement in general situations, particularly for mixed states, or for open as well as non-Hermitian systems, remains conceptually and technically challenging.

Traditional measures of entanglement rely on the Schmidt or Singular Value Decomposition (SVD) of a pure state or operator in a bipartite Hilbert space $\mathcal{H}=\mathcal{H}_A\otimes \mathcal{H}_B$ (see \cite{Hosur:2015ylk,Dubail:2016xht,Nie:2018dfe} for more on operator entanglement). For pure states, the Schmidt coefficients fully determine the entanglement spectrum and yield Rényi or von Neumann entropies with clear operational meaning. However, when the system evolves non-unitarily, interacts with an environment, or the setup is described by transition rather than density matrices, the conventional SVD-based framework becomes ambiguous. In such cases, even defining an ``entropy" that reflects physical, quantum correlations rather than normalisation or scaling artefacts can be subtle.

Related to these challenges, several generalisations of entanglement have recently been proposed, including pseudo-entropy \cite{Nakata:2020luh,Mollabashi:2021xsd}, SVD entropy \cite{Parzygnat:2023avh}, and entropies in time \cite{Doi:2022iyj,Doi:2023zaf,Milekhin:2025ycm,Das:2025fcd} among them. These approaches use the singular values of transition matrices or non-Hermitian operators to construct Rényi-type quantities that extend entanglement to post-selected, time-dependent, or general open quantum systems. Such ideas and tools have already found applications across diverse areas, from dS/CFT correspondence \cite{Doi:2022iyj,Hikida:2022ltr,Hikida:2021ese}, renormalisation \cite{Grieninger:2023knz}, topological theories \cite{Nishioka:2021cxe,Caputa:2024qkk} to spectral form factors used to diagnose quantum chaos \cite{Caputa:2024gve}.

However, a fundamental issue that persists across all these constructions is that the singular values themselves depend on the choice of units and scaling conventions in the subspaces involved. To see the problem concretely, consider the simple quantum mechanical example of a two-level system where one rescales the basis states in one subsystem by an arbitrary factor (a trivial change of normalisation or physical units). Although this transformation leaves the physical content unchanged, it modifies the singular values. As a result, the entropy constructed from those singular values varies under a mere redefinition of scale; i.e., conventional SVD, though invariant under unitaries, is not invariant under such diagonal rescalings. Consequently, SVD-based entropic measures may conflate genuine quantum correlations with artefacts of normalisation, a subtle but crucial distinction, especially in field-theoretic or non-Hermitian settings where operator norms (choices of inner-products) or couplings carry non-trivial scaling behaviour (see e.g., \cite{Edvardsson:2022day}).

A mathematically rigorous resolution to this issue for SVD was recently provided by Uhlmann \cite{uhlmann-main2}. Following earlier works on the construction of singular values invariant under more specific rescalings, he introduced the Unit-Invariant Singular Value Decomposition (UISVD) as a variation of the standard SVD that remains invariant under diagonal (unit) transformations acting on a given operator from both sides. The UISVD produces unit-invariant singular values that eliminate spurious dependencies on scaling conventions. In effect, UISVD isolates the physically meaningful, scale-independent structure of a linear map.

This property makes the UISVD particularly appealing for the above-mentioned physical setups, where operators often connect Hilbert spaces with different metric or normalisation structures present. Indeed, in these contexts, one may start from entanglement and information measures that are invariant under simple changes of physical units and hope that classifying them will be easier and more reliable, similarly to von Neumann entropies for pure states.

In this work, we introduce the UISVD framework into quantum information and field-theoretic contexts, proposing it as a new foundation for unit-invariant measures of entanglement. We construct UISVD-based analogues of von Neumann entropy and demonstrate the variation from familiar SVD-based quantities while removing their dependence on arbitrary scale choices. These measures are naturally defined for non-Hermitian or rectangular operators and remain useful when the input and output spaces possess different dimensions or metric weights. As an initial demonstration, we apply UISVD-based entropies to several well-known setups, from simple quantum mechanical bipartite mixed states to random matrices relevant to quantum chaos and holography. In all cases, the UISVD yields stable, physically meaningful entropic spectra that are invariant under rescalings and normalisations.

Beyond their practical advantages, our results suggest a deeper conceptual point: UISVD restores a natural physical covariance to operator-based entropies, ensuring that what we quantify as ``entanglement" or ``correlation" is independent of arbitrary normalisation or choice of units. In this sense, we hope that UISVD may contribute to providing a unifying, invariant language for quantifying quantum correlations in systems where conventional assumptions of unitarity or normalisation fail.

Let us briefly summarise how we construct scale-invariant entropy measures. First, following \cite{uhlmann-main2}, for a matrix $A$ we introduce its left unit-invariant, right unit-invariant, and (both left and right) bi-unit-invariant singular values, denoted respectively $\sigma_k^\mathrm{L}(A), \sigma_k^\mathrm{R}(A)$ and $\sigma_k^\mathrm{B}(A)$. Scale (or unit) invariance means that they satisfy the relations
\begin{equation}
\sigma_k^\mathrm{L}(A) = \sigma_k^\mathrm{L}(DAU), \qquad \sigma_k^\mathrm{R}(A) = \sigma_k^\mathrm{R}(UAD), \qquad \sigma_k^\mathrm{B}(A) = \sigma_k^\mathrm{B}(DAD'), \label{rescale}
\end{equation}
where $D$ and $D'$ are complex non-singular diagonal matrices that represent scaling transformations acting on $A$ from left or right, while $U$ is a unitary matrix (so that left- and right-singular values are additionaly invariant under unitary transformation on the opposite side). We then consider quantum systems with bipartite Hilbert spaces $\mathcal{H} = \mathcal{H}_\mathbb{A} \otimes \mathcal{H}_\mathbb{B}$ characterised by a reduced density matrix or a reduced transition matrix $\rho_\mathbb{A}=\textrm{Tr}_\mathbb{B}(\rho)$, obtained by taking the partial trace of a density or transition matrix $\rho$. The systems of our interest involve some natural scaling operations, which amount to the rescaling of $\rho_\mathbb{A}$ by diagonal matrices $D$ and $D'$ from the left or right, or both sides. We then identify the corresponding unit-invariant singular values, which are scale-invariant analogously to \cref{rescale}, and introduce associated unit-invariant entropy measures. Depending on the system under consideration, such entropies generalise other familiar versions of entanglement entropy, such as ordinary von Neumann entropy, or pseudo entropy and SVD entropy (when $\rho$ is a transition matrix with respect to some pre- or post-selected state). In the rest of the paper we provide details of such constructions, discuss properties of these entropy measures, and illustrate them in several examples.

This paper is organised as follows. In \cref{sec-scale}, we review the construction of left, right, and bi unit-invariant singular values, and provide their geometric interpretation. To build intuition for these invariants, we compute them explicitly for $2\times2$ matrices. We also analyse their statistical properties for ensembles of random matrices and show that they obey quarter-circle laws. In \cref{sec:entanglement-measures}, we use these unit-invariant singular values to define new unit-invariant entanglement entropies for reduced density matrices derived from pure states. \cref{UISVDprepost} extends this framework to reduced transition matrices, introducing a variation of the pseudo-entropy and SVD-based entropies introduced in the literature. In \cref{sec:Applications}, we present several applications of these measures, ranging from random pure states and link-complement states in Chern-Simons theory to Biorthogonal Quantum Mechanics (BQM), where we evaluate, test, and compare them for both reduced density matrices and transition matrices. Finally, we conclude in \cref{Conclusions} and defer more technical material to two appendices.
\section{Unit-Invariant Singular Values} \label{sec-scale}
In this section we review the construction of various scaling invariants of matrices discussed in \cite{uhlmann-main2}. We call them left unit-invariant, right unit-invariant, and bi-unit-invariant singular values, and denote respectively LUI-, RUI-, and BUI-singular values. These invariants will play the key role in the construction of entropy measures introduced in the next sections. In the following subsections we collect definitions of unit-invariant singular values, provide their geometric interpretation, present explicit results for $2\times2$ matrices, and derive their distributional behaviour for Gaussian ensembles.
\subsection{Definitions of singular values}
To start with, we quickly recall that the Singular Value Decomposition (SVD) of a matrix \(A \in \mathbb{C}^{m\times n}\) is
\begin{equation}
A = U\Sigma V^{\dagger},
\end{equation}
where \(U\in\mathbb{C}^{m\times m}\) and \(V\in\mathbb{C}^{n\times n}\) are unitary, and \(\Sigma\in\mathbb{R}^{m\times n}\) is (rectangular) diagonal with nonnegative entries \(\sigma_1\ge \sigma_2\ge \cdots \ge 0\). Here $\sigma_k\equiv \sigma_k(A)$ are singular values of $A$ and can be directly obtained using \(\sigma_k=\sqrt{\lambda_k}\), where \(\lambda_k\) are the eigenvalues of \(A^{\dagger}A\). They satisfy the property $\sigma_k(A)=\sigma_k(U^\prime AV^\prime)$, i.e., they are invariant under multiplication on the left or right by arbitrary unitaries $U'$ and $V'$.

We now introduce variations of this framework following \cite{uhlmann-main2}.
\paragraph{Left Unit-Invariant (LUI) singular values.}
Consider a matrix \( A \in \mathbb{C}^{m \times n} \). For each row $i$ of \( A \), its Euclidean norm is given by
\begin{equation}
\|A(i,:)\|_2 = \sqrt{\sum_{j=1}^{n} |A_{ij}|^2}.
\end{equation}
The left diagonal scale function matrix \( D^\mathrm{L}(A) \in \mathbb{R}^{m\times m}_+\) is defined as
\begin{equation}\label{eq:LDSF}
D^\mathrm{L}_{ii} =
\begin{cases}
\frac{1}{\|A(i, :)\|_2}, & \text{if} \ \|A(i, :)\|_2 > 0 \\
1, & \text{otherwise}\,.
\end{cases}
\end{equation}
First, we define the \textit{balanced matrix} \( A^\mathrm{L} \) as 
\begin{equation}
A^\mathrm{L} \coloneq D^\mathrm{L} A.\label{balancedA}
\end{equation}
Then, the singular values of \( A^\mathrm{L} \), determined by its singular value decomposition and denoted \( \sigma(A^\mathrm{L}) \), are called the \textit{left unit-invariant singular values} \( \sigma^\mathrm{L}(A) \) of the original matrix $A$,
\begin{equation}
\sigma^\mathrm{L}(A) \coloneq \sigma(A^\mathrm{L}).\label{LSVA}
\end{equation}
So we essentially decompose $A$ as
\begin{equation}
A = (D^\mathrm{L})^{-1} A^\mathrm{L} = (D^\mathrm{L})^{-1} U_0 \Sigma_0 V^\dagger_0,
\end{equation}
and call it the \textit{left unit-invariant singular value decomposition (LUISVD)}. The invariants $\sigma_k^\mathrm{L}(A)$ lie on the diagonal of the diagonal matrix $\Sigma_0$ and are non-negative real numbers. These values possess the invariance property
\begin{equation}
\sigma^\mathrm{L}(A) = \sigma^\mathrm{L}(DAU),
\end{equation}
for any \( D \in \mathcal{D} \), where $\mathcal{D}$ is the space of complex non-singular diagonal matrices, and any \( U \in \mathcal{U} \), where $\mathcal{U}$ is the space of unitary matrices. For an explicit proof of this claim, see Appendix~\ref{sec:2x2matrix-derivation}.
\paragraph{Right Unit-Invariant (RUI) singular values.} Their definition is analogous to the LUI singular values, and simply involves column norms instead of row norms in the right diagonal scale function matrix \( D^\mathrm{R}(A) \in \mathbb{R}^{n\times n}_+\) which is defined as
\begin{equation}\label{eq:RDSF}
D^\mathrm{R}_{jj} =
\begin{cases}
\frac{1}{\|A(:, j)\|_2}, & \text{if} \ \|A(:, j)\|_2 > 0 \\
1, & \text{otherwise}\,.
\end{cases}
\end{equation}
The \emph{balanced matrix} \( A^\mathrm{R} \) is then defined by
\begin{equation}
A^\mathrm{R} \coloneq AD^\mathrm{R}.\label{balancedAR}
\end{equation}
The singular values of \( A^\mathrm{R} \), denoted \( \sigma(A^\mathrm{R}) \), are called the \textit{right unit-invariant singular values} \( \sigma^\mathrm{R}(A) \) of the original matrix $A$,
\begin{equation}
\sigma^\mathrm{R}(A) \coloneq \sigma(A^\mathrm{R}).\label{RSVA}
\end{equation}
In this case the \textit{right unit-invariant singular value decomposition (RUISVD)} of $A$ reads
\begin{equation}
A = A^\mathrm{R}(D^\mathrm{R})^{-1} = U_1 \Sigma_1 V^\dagger_1 (D^\mathrm{R})^{-1}.
\end{equation} 
The RUI singular values are invariant under multiplication by \( D \in \mathcal{D} \) and \( U \in \mathcal{U} \) from the right and left, respectively
\begin{equation}
\sigma^\mathrm{R}(A) = \sigma^\mathrm{R}(UAD).
\end{equation}
\paragraph{Bi Unit-Invariant (BUI) singular values.} It is a bit harder to implement invariance under simultaneous multiplication by diagonal matrices on the left and right. Nevertheless, such a construction was provided in \cite{uhlmann-main2}. It amounts to determining a pair of \emph{positive diagonal} matrices from $A$, denoted ${D}^{\mathrm{B_L}}(A) \in \mathbb{R}^{m\times m}_+$ and ${D}^{\mathrm{B_R}}(A) \in \mathbb{R}^{n\times n}_+$, and then forming the \emph{balanced matrix}
\begin{equation}\label{eq:ui-bal-def}
A^\mathrm{B} 
\;=\; 
{D}^{\mathrm{B_L}}\,A\,{D}^{\mathrm{B_R}}.
\end{equation}
A simple example of this construction for matrices $A$ of size $2\times2$ with complex non-zero elements is given in Appendix~\ref{sec:2x2matrix-derivation}, along with a code to compute these values for any general matrix. In general, the diagonals ${D}^{\mathrm{B_L}}$ and ${D}^{\mathrm{B_R}}$ are constructed so that both the \emph{row} and \emph{column} \emph{geometric means} of $|A^\mathrm{B}|$ are equal to $1$, or equivalently the products are 1, with the geometric mean taken over the \emph{non-zero} entries. Concretely, this means that for every row $i$ and column $j$ which are not identically zero,
\begin{equation}\label{eq:UI-balance}
\prod_{j:\,A^\mathrm{B}_{ij}\neq 0} |A^\mathrm{B}_{ij}|=1,\qquad(1\le i\le m),
\qquad
\prod_{i:\,A^\mathrm{B}_{ij}\neq 0}\!\!\! |A^\mathrm{B}_{ij}|=1,\qquad(1\le j\le n).
\end{equation}
When zero entries are present, we adopt precisely the scaling in \cite[Program II]{ROTHBLUM1992159} as mentioned in \cite{uhlmann-main2} applied to $|A|$, which yields positive diagonal $D^{\mathrm{B_L}}, D^{\mathrm{B_R}}$ and a balanced matrix satisfying \cref{eq:UI-balance}. If $A$ contains an all-zero row or column, these are removed prior to scaling and the corresponding diagonal entries are set to $1$ upon re-embedding; the subsequent BUI results then hold verbatim.\footnote{See \cite[Thm.~4.2, Lem.~4.3, Thm.~4.4, and Appendix C]{uhlmann-main2}}

Now, the singular values of \( A^\mathrm{B} \), denoted \( \sigma(A^\mathrm{B}) \), are called the \textit{bi unit-invariant singular values} \( \sigma^\mathrm{B}(A) \) of the original matrix $A$,
\begin{equation}
\sigma^\mathrm{B}(A) \coloneq \sigma(A^\mathrm{B}). \label{sigmaU}
\end{equation}
Equivalently, these values arise on the diagonal of the matrix $\Sigma_2$ upon the \textit{bi unit-invariant singular value decomposition (BUISVD)}
\begin{equation}
A = (D^\mathrm{B_L})^{-1}A^\mathrm{B}(D^\mathrm{B_R})^{-1} = (D^\mathrm{B_L})^{-1} U_2 \Sigma_2 V^\dagger_2 (D^\mathrm{B_R})^{-1}.
\end{equation}
These BUI singular values possess the invariance property
\begin{equation}
\sigma^\mathrm{B}(A) = \sigma^\mathrm{B}(DAD'),
\end{equation}
for any \( D,D^\prime \in \mathcal{D} \), where again $\mathcal{D}$ is the space of complex non-singular diagonal matrices.

Although we do not use it in this paper, for the sake of completeness, we note that when the matrix $A\in \mathbb C^{n\times n}$ is square, there is one more invariant that can be extracted from this construction, namely the \textit{scale-invariant eigenvalues}\footnote{We reserve the notation $\lambda^\mathrm B$ for a different construction used later.} of a matrix $A$
\be
\lambda^\mathrm S (A) \coloneq \lambda(A^\mathrm B).
\ee
This then satisfies the more restricted invariance property 
\be\label{uievd}
\lambda^\mathrm S (A) = \lambda^\mathrm S (D^\mathrm S A {D^\mathrm S}^\prime) \qquad \forall D^\mathrm S,{D^\mathrm S}^\prime \in \mathcal{D} \quad \text{s.t. } (D^\mathrm S {D^\mathrm S}^\prime)\in \mathbb R^{n\times n}_{+}.
\ee
Several kinds of transformations fall under this category, for example invariance under $DA\bar D$ or $DAD^{-1}$ where $D \in \mathcal D$ is any complex non-singular diagonal matrix. Of course the case of \textit{positive scalings}, i.e.\ $D^\mathrm S,{D^\mathrm S}^\prime$ both being positive real diagonal, automatically satisfies \cref{uievd}.
\paragraph{Normalisation.}
Recall that the \emph{sum of squares of the singular values} of any matrix equals its \emph{Frobenius (Hilbert-Schmidt) norm} squared. Specifically, for $A^{\mathrm L}$,
\begin{equation}\label{eq:frob-norm-lui}
\|A^{\mathrm L}\|_\mathbb{F}^2
=
\sum_{i=1}^{m}\sum_{j=1}^{n} |A^{\mathrm L}_{ij}|^2
=
\sum_{k=1}^{\min(m,n)} \sigma_k(A^{\mathrm L})^2
=
\sum_{k=1}^{\min(m,n)} \bigl(\sigma_k^{\mathrm L}(A)\bigr)^2.
\end{equation}
Let $d_{\mathbb A}, d_{\mathbb B}$ be respectively the numbers of non-zero rows and columns of $A$ (hence also for $A^\mathrm L$ and $A^\mathrm R$). If $A$ has no identically-zero rows/columns, then $d_{\mathbb A}=m$ and $d_{\mathbb B}=n$. By construction, every non-zero row of $A^{\mathrm L}$ has Euclidean norm $1$, hence
\begin{equation}
\|A^{\mathrm L}\|_\mathbb{F}^2
=
\sum_{i=1}^{m} \|A^{\mathrm L}(i,:)\|_2^2
=
\sum_{i:\ \|A(i,:)\|_2>0} 1
=
d_{\mathbb A}.
\end{equation}
Analogously, $\|A^{\mathrm R}\|_\mathbb{F}^2=d_{\mathbb B}$. Therefore
\begin{equation}
\sum_{k=1}^{\min(m,n)} \bigl(\sigma_k^{\mathrm L}(A)\bigr)^2 = d_{\mathbb A},
\qquad
\sum_{k=1}^{\min(m,n)} \bigl(\sigma_k^{\mathrm R}(A)\bigr)^2 = d_{\mathbb B}.
\end{equation}
This allows us to introduce the normalised left- and right-singular values as
\be\label{eq-normalised-luirui-values}
\hat{\sigma}_k^\mathrm{L}\coloneq \frac{\sigma_k^\mathrm{L}}{\sqrt{d_{\mathbb{A}}}},\qquad \hat{\sigma}_k^\mathrm{R}\coloneq \frac{\sigma_k^\mathrm{R}}{\sqrt{d_{\mathbb{B}}}}.
\ee
Unfortunately, we were not able to derive a simple expression (in terms of $d_\mathbb{A}$ and $d_\mathbb{B}$) for the Frobenius norm of $A^\mathrm{B}$, so we just write formally the normalised singular values of $A^\mathrm{B}$ as
\be
\hat{\sigma}_k^\mathrm{B}\coloneq \frac{\sigma_k^\mathrm{B}}{||A^\mathrm{B}||_\mathbb{F}}=\frac{\sigma_k^\mathrm{B}}{\sqrt{\sum_l(\sigma_l^\mathrm{B})^2}}.\label{NormsingvalAU}
\ee
These normalised values will appear directly in our definitions of entropies based on the UISVD decompositions.
\subsection{Geometric interpretation}
\paragraph{Geometric interpretation of UISVD.} It is worth invoking a geometric interpretation of the quantities introduced above. SVD has a well-known representation in Euclidean space, where it takes the unit sphere in the input space and turns it into an ellipsoid in the output space. The right singular vectors are the input directions that land on the ellipsoid's principal axes, and the left singular vectors point along those axes in the output space, and the singular values are the axis lengths.

LUI-SVD and RUI-SVD keep this same sphere to ellipsoid story, but they first ``re-unit" one side by a diagonal stretch of the coordinate axes chosen from the Euclidean row or column norms. In LUI-SVD one stretches or shrinks the output coordinate axes so the non-zero rows have equal size under that norm, which puts the output channels on a common scale before looking for the ellipsoid's principal directions. In RUI-SVD one does the analogous axis rescaling on the input side so the non-zero columns have equal size, which makes the input coordinates comparable before looking for principal directions. After this one-sided axis standardisation, an ordinary SVD gives the ellipsoid geometry in the standardised coordinates, and the diagonal (inverse) factor just converts that geometric description back to the original units. 

BUI-SVD is the fully symmetric version. It diagonally stretches or shrinks both input and output axes until the matrix is balanced in a multiplicative sense, so the products of the non-zero entry magnitudes are equal across non-zero rows and columns, which one can think of as balancing typical entry magnitudes across rows and columns, or as centering log-magnitudes when no zeros are present. Once that balancing has fixed the arbitrary choice of units on both sides, SVD again reads off the same sphere to ellipsoid geometry, now in the balanced coordinate system, and the outer diagonal factors simply map the picture back to the original coordinates.
\paragraph{Geometry of the local group actions.} We use the term ``local'' since, when $A$ is the coefficient matrix of a bipartite state, left/right multiplication is interpreted as restricted local transformations on the two subsystems in \cref{sec:entanglement-measures,UISVDprepost} (used as an equivalence/rescaling). Besides the familiar compact Lie group of local unitary basis changes \(\mathcal U(d)\) (real dimension \(d^{2}\)), the new ingredient here is the (generally non-unitary) diagonal
rescaling group \(\mathcal D(d)=\{\mathrm{diag}(z_1,\dots,z_d):z_i\in\mathbb C^*\}\cong(\mathbb C^*)^{d}\) (real dimension \(2d\)). Writing \(z_i=r_i e^{i\theta_i}\) gives the product decomposition (as a real Lie group)
\be
\mathcal D(d)\cong(\mathbb R_{+})^{d}\times (U(1))^{d},\qquad
r_i=e^{x_i}\ \Rightarrow\ (\mathbb R_{+})^{d}\cong\mathbb R^{d},
\ee
so \(\mathcal D(d)\) is ``flat in logarithmic coordinates"; the non-compact directions are the log-magnitudes (units/normalisations), while the phases form a compact torus
(\((U(1))^{d}\subset\mathcal U(d)\)). So for any object matrix of interest $A$ with dimensions $m\times n$, the four local action groups relevant for the constructions above are
\be
\mathcal U(m)\times \mathcal U(n),\quad
\mathcal D(m)\times \mathcal U(n),\quad
\mathcal U(m)\times \mathcal D(n),\quad
\mathcal D(m)\times \mathcal D(n),
\ee
which are direct products of these building blocks and therefore mix compact and non-compact geometry. In particular, \(\mathcal D(m)\times \mathcal D(n)\) deformation retracts onto the $(m+n)$-torus \((U(1))^{m+n}\) but contains non-compact ``runaway" directions in the \(\mathbb R^{m+n}\) factors. This action always has at least the $\mathbb C^*$ stabiliser $(cI_m,c^{-1}I_n)$; additional stabilisers arise if $A$ has zero rows/columns or a decomposable support pattern. Modulo this trivial scaling, the remaining non-compact directions act generically freely for generic full-support matrices.

Correspondingly, the BUISVD scaling factors \(D^\mathrm{B_L}\) and \(D^\mathrm{B_R}\) are determined only up to the one-parameter rescaling \(D^\mathrm{B_L}\mapsto cD^\mathrm{B_L}\), \(D^\mathrm{B_R}\mapsto c^{-1}D^\mathrm{B_R}\), which leaves \(A^{\mathrm B}=D^\mathrm{B_L}AD^\mathrm{B_R}\) unchanged; one can fix a canonical representative by a symmetric scale choice (for square matrices with non-zero entries, see Proposition~\ref{prop:gauge-fully}). Once diagonal rescalings are admitted, the orbit geometry acquires non-compact ``units" directions, and a normalisation/balancing choice is precisely a way to pick a canonical slice through those directions before extracting singular value data.
\subsection{Unit-Invariants for two-dimensional spaces} 
As the simplest illustration, in \cref{tab:invariants} we provide the explicit form of the above invariants (for the sake of completeness, also including eigenvalues and singular values) for an arbitrary complex matrix of size 2 with non-zero entries
\begin{equation}\label{eq:2x2matrix}
A = \begin{pmatrix}
a_1 & a_2 \\
a_3 & a_4
\end{pmatrix}\,.
\end{equation}
These invariants and their properties are derived in Appendix~\ref{sec:2x2matrix-derivation}.

\begin{table}[h!]
\centering
\setlength{\tabcolsep}{10pt}
\renewcommand{\arraystretch}{1.45}
\setlength{\aboverulesep}{0.6ex}
\setlength{\belowrulesep}{0.6ex}
\newcolumntype{L}{>{\raggedright\arraybackslash}p{0.30\textwidth}}
\newcolumntype{R}{>{\raggedright\arraybackslash}X}
\begin{tabularx}{\textwidth}{@{} L R @{}}
\toprule
\textbf{Invariant} & \textbf{Formula} \\
\midrule

Eigenvalues \(\lambda_{\pm}\) &
$\displaystyle
 {\frac12\Big(\,a_1 + a_4 \,\pm\, \bigl((a_1 - a_4)^2 + 4\,a_2 a_3\bigr)^{\frac12}\Big)}
$ \\[2pt]

Singular values \(\sigma_{\pm}
\) &
$\displaystyle
\biggl(
\frac{1}{2}\Bigl(
\sum_{i=1}^{4}\!|a_i|^2
\;\pm\;
\Bigl(\Bigl(\sum_{i=1}^{4}\!|a_i|^2\Bigr)^{\!2} - 4\,\tilde{a}\Bigr)^{\frac12}
\Bigr)
\biggr)^{\!\frac12}
$ \\[2pt]

BUI-singular values \(\sigma^\mathrm{B}_{\pm}\) &
$\displaystyle
 \left(
\hat{a} + \hat{a}^{-1}
\ \pm\
\left[
\left(\hat{a} + \hat{a}^{-1}\right)^{\!2}
-
\frac{\tilde{a}}{\left\lvert a_1 a_2 a_3 a_4 \right\rvert}
\right]^{\frac12}
\right)^{\frac12}
$ \\[2pt]

LUI-singular values \(\sigma^\mathrm{L}_{\pm}\) &
$\displaystyle
 \biggl(
\,1 \pm \Bigl(\,1
- \frac{\tilde{a}}
{(|a_1|^{2}+|a_2|^{2})\, (|a_3|^{2}+|a_4|^{2})}\,\Bigr)^{\frac12}
\biggr)^{\!\frac12}
$ \\[2pt]

RUI-singular values \(\sigma^\mathrm{R}_{\pm}
\) &
$\displaystyle
 \biggl(
\,1 \pm \Bigl(\,1
- \frac{\tilde{a}}
{(|a_1|^{2}+|a_3|^{2})\, (|a_2|^{2}+|a_4|^{2})}\,\Bigr)^{\frac12}
\biggr)^{\!\frac12}
$ \\
\bottomrule
\end{tabularx}
\caption{Closed form expressions for eigenvalues, singular values, and unit-invariant singular values of matrices of size 2 as given in \cref{eq:2x2matrix}. Here \(\hat{a} \coloneq \big|({a_1 a_4})/({a_2 a_3})\big|^{\frac12}\) and \(\tilde{a} \coloneq |a_1 a_4 - a_2 a_3|^{2}\).}
\label{tab:invariants}
\end{table}

\subsection{Unit-Invariant singular values for random matrices}\label{sec:RMT}
To gain more insight into scale invariance, in this section we investigate distributions of invariants for random matrices drawn from several standard ensembles (Ginibre and Wigner). Our main result is the statement that distributions of these invariants take the form of a quarter-circle law. To start with, for a fixed dimension \(n\), and \(m\) independent random matrices, we compute eigenvalues, singular values, as well as left-, right-, and bi-unit-invariant singular values. We denote them by $\sigma^{(j)}_i$, where $i=1,\ldots,n$ labels invariants for a given matrix and $j=1,\ldots,m$ labels matrices in a given random set. The empirical ``spectral" measure is
\begin{equation}
\rho(x)=\frac{1}{mn}\sum_{j=1}^{m}\sum_{i=1}^{n}\delta\!\bigl(x-\sigma^{(j)}_{i}\bigr),
\end{equation}
estimated by pooling all \(m\times n\) samples into a normalised histogram.

\begin{table}[h!]
  \centering
  \footnotesize
  \begin{tabularx}{\linewidth}{@{} l c @{\hspace{15em}} X @{}}
    \toprule
    Ensemble & $\beta$ & Construction \\
    \midrule
    Real Ginibre & 1 &
    $A=\dfrac{1}{\sqrt n}G,\ \ G_{ij}\overset{\text{i.i.d.}}{\sim}\mathcal N(0,1)$ \\
    Complex Ginibre & 2 &
    $A=\dfrac{1}{\sqrt n}G,\ \ G_{ij}\overset{\text{i.i.d.}}{\sim}\mathcal{CN}(0,1)$ \\
    \bottomrule
  \end{tabularx}
  \caption{Ginibre $\beta$-ensembles.}
  \label{tab:ginibre}
\end{table}

In \cref{tab:ginibre} we present the Ginibre \(\beta\)-ensembles which are constructed using Gaussian entries. In \cref{thm:lui-rui-support,thm:uisvd-support} we prove the distribution limits more generally for independent and identically distributed (i.i.d.)\footnote{Random variables that are independent of one another and all have the same probability distribution.} sub-Gaussian entries, hence the limiting behaviours for the real and complex Ginibre cases (\(\beta=1,2\)) are direct corollaries. 

Let $X_n=(X_{ij})_{1\le i,j\le n}$ have i.i.d. entries with mean $0$, variance $1$, sub-Gaussian tails, $\Prob(X_{ij}=0)=0$, and assume a small-ball bound near $0$.\footnote{There exist $\alpha>0$ and $C<\infty$ such that \( \Prob(|X_{ij}|\le t)\le C\,t^{\alpha} \ \text{for all }t\in(0,1). \)} Consider 
\begin{equation}\label{eq:def-An}
A_n\coloneq \frac{1}{\sqrt n}\,X_n.
\end{equation}
We compute the empirical singular-value measures of $A^\mathrm{L}_n,A^\mathrm{R}_n$, and $A^\mathrm{B}_n$ which by definition correspond to the empirical measures associated to $\sigma^{\mathrm L}(A_n), \sigma^{\mathrm R}(A_n)$, and $\sigma^{\mathrm B}(A_n)$ respectively.\footnote{To construct $A_n^\mathrm B$ we choose the symmetric scale as per Proposition~\ref{prop:gauge-fully} to fix a unique choice.} 
\begin{theorem}[Quarter-circle law for LUI- and RUI-SVD]\label{thm:lui-rui-support}
The empirical singular-value measures of $A^\mathrm{L}_n$ and $A^\mathrm{R}_n$ converge almost surely to the quarter-circle law on $[0,2]$ with density $f(s)=\frac1\pi\sqrt{4-s^2}\,\mathbf{1}_{[0,2]}(s)$.
\end{theorem}
\begin{proof}
See Appendix~\ref{proof-of-thm1}.
\end{proof}
\begin{theorem}[Stretched quarter-circle law for BUI-SVD]\label{thm:uisvd-support}
Set $c_\star\coloneq \mathbb{E}\log|X_{11}|\in\mathbb{R}$. Then, on the unnormalised scale, $\|A^{\mathrm{B}}_n\|_{\op}=2e^{-c_\star}\sqrt n\,\big(1+o(1)\big)$ almost surely. The empirical singular-value measures of $T_n\coloneq n^{-1/2}A^{\mathrm{B}}_n$
converge to the quarter-circle law on $[0,\,2e^{-c_\star}]$ with density \(f_\star(s) = \frac{1}{\pi\,e^{-2c_\star}} \sqrt{(2e^{-c_\star})^2-s^2}\,\mathbf{1}_{[0,\,2e^{-c_\star}]}(s)\).
\end{theorem}
\begin{proof}
See Appendix~\ref{proof-of-thm2}.
\end{proof}
For the special case of Ginibre \(\beta\)-ensembles as in \cref{tab:ginibre} we take $X_{ij}=G_{ij}$ with the choice of Gaussian entries with $c_\star = c_\beta$ for the appropriate entry law.\footnote{The numerical values of \(c_\beta\) are given in Lemma~\ref{lem:log-constants}.} 

\begin{table}[h!]
  \centering
  \footnotesize
  \begin{tabularx}{\linewidth}{@{} l c l @{\hspace{8.7em}} X @{}}
    \toprule
    Ensemble & $\beta$ & Symmetry & Construction \\
    \midrule
    GOE & 1 & real symmetric &
    $G_{ij}\overset{iid}{\sim} \mathcal N(0,1), \qquad A=\dfrac{G+G^{\mathsf T}}{\sqrt{2n}}$ \\
    GUE & 2 & complex Hermitian &
    $G_{ij}\overset{iid}{\sim} \mathcal{CN}(0,1), \qquad A=\dfrac{G+G^{\dagger}}{\sqrt{2n}}$ \\
    \bottomrule
  \end{tabularx}
  \caption{Wigner $\beta$-ensembles.}
  \label{tab:wigner-ensembles}
\end{table}

We also look at Wigner-$\beta$ ensembles as given in \cref{tab:wigner-ensembles}, and the densities are plotted in \cref{fig:gue-gse} using a large number of matrices.
\begin{theorem}\label{thm:lui-rui-wigner}
Let $A_n$ be GOE/GUE with the $1/\sqrt n$ scaling as declared in \cref{tab:wigner-ensembles}.
Then the empirical singular-value measures of $A_n^{\mathrm L}$ and $A_n^{\mathrm R}$ converge almost surely to the quarter-circle law on $[0,2]$.
\end{theorem}
\begin{proof}
See Appendix~\ref{proof-of-lem1}.
\end{proof}
\begin{theorem}\label{thm:ui-wigner}
Let $A_n$ be GOE/GUE with the $1/\sqrt n$ scaling. Then \(\|A_n^{\mathrm B}\|_{\op}=2e^{-c_\beta}\sqrt n\,\big(1+o(1)\big)\ \text{almost surely}.\) Equivalently, the empirical singular-value measures of $T_n\coloneq n^{-1/2}A_n^{\mathrm B}$ converge almost surely to the quarter-circle law on $[0,\,2e^{-c_\beta}]$.
\end{theorem}
\begin{proof}
See Appendix~\ref{proof-of-lem2}.
\end{proof}
\begin{figure}[h!]
  \centering
  \includegraphics[width=\textwidth,height=5cm]{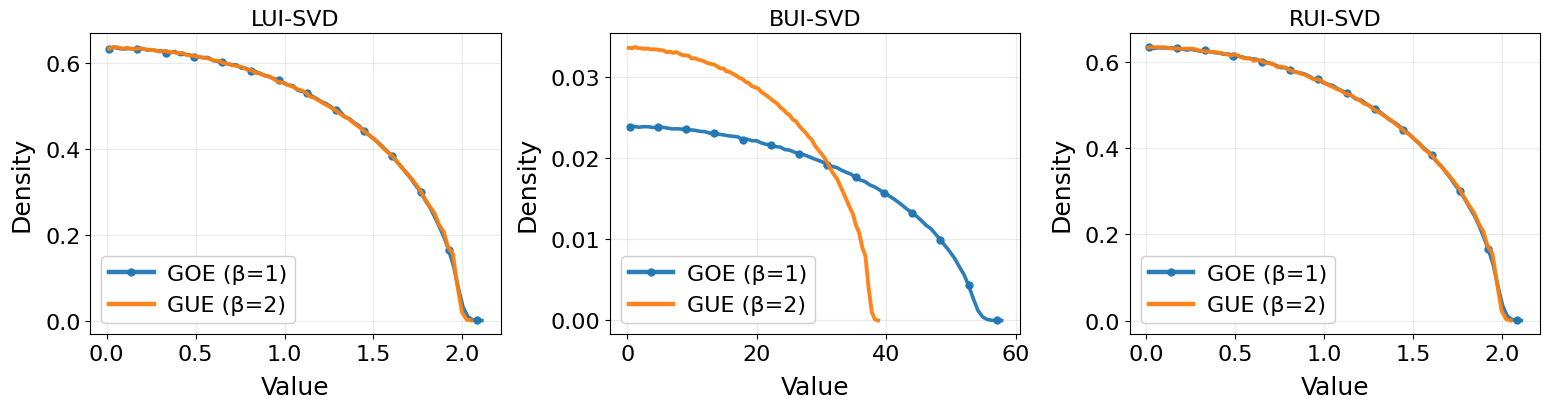}
  \caption{Empirical densities for the invariants $\{\sigma_k^\mathrm{L}\}$, $\{\sigma_k^\mathrm{R}\}$, and $\{\sigma_k^{\mathrm{B}}\}$ with $n=200$. LUI and RUI follow the quarter-circle law of singular value distribution. BUI has dimension-dependent stretched supports $\approx$[$0,2.67\sqrt{n}$] for GUE, and $\approx$[$0,3.77\sqrt{n}$] for GOE.}
  \label{fig:gue-gse}
\end{figure}

\section{Unit-Invariant entanglement entropies}\label{sec:entanglement-measures}
In this section we introduce various entanglement entropies based on unit-invariant singular values defined above. They extend von Neumann, pseudo- or SVD entropies considered earlier in the literature, whose definitions and properties we briefly summarise first.
\subsection{Review of entanglement entropy}
We start with the usual setup of a bipartite Hilbert space
\begin{equation}
\mathcal{H} = \mathcal{H}_\mathbb{A} \otimes \mathcal{H}_\mathbb{B},
\end{equation}
with subsystems of arbitrary finite dimensions $d_{\mathbb{A}}$ and $d_{\mathbb{B}}$, respectively. We consider a bipartite pure state $\ket{\psi}$ and the corresponding density matrix, expressed in some fixed orthonormal basis, as 
\begin{equation}
 \ket{\psi} \;=\; \sum_{i,k} A_{ik}\, \ket{i}_\mathbb{A} \ket{k}_\mathbb{B},\qquad \rho=\ket{\psi}\bra{\psi}=\sum_{i,j,k,l}A_{ik}A^*_{jl}\ket{i}_\mathbb{A}\otimes \ket{k}_\mathbb{B}\bra{j}_\mathbb{A}\otimes \bra{l}_\mathbb{B},
 \label{eq:psi-original}
\end{equation}
with coefficient matrix \(A\) that does not have any all-zero rows or all-zero columns. The reduced density matrix of $\mathbb{A}$ becomes
\begin{equation}
\rho_\mathbb{A}=\Tr_\mathbb{B}(\rho)=\sum_{i,j}(AA^\dagger)_{ij}\ket{i}_\mathbb{A}\bra{j}_\mathbb{A}.
\end{equation}
To compute entropies we will consider a normalised density matrix with $\Tr(\rho_\mathbb{A})=1$, which we can write as the operator\footnote{It is usually more natural to start this construction with normalising $\ket{\psi}$ so that $\Tr(A^\dagger A)=1$, but we keep this discussion more general for later purposes.}
\be
\rho_\mathbb{A}=\frac{AA^\dagger}{\Tr(AA^\dagger)}.
\ee
Then, in the standard computation of von Neumann entropy, we perform the singular value decomposition of $A$
\be
A=U\Sigma V^\dagger,\label{SVDdecA}
\ee
with unitaries $U$ and $V^\dagger$ and diagonal matrix $\Sigma$ with singular values $\sigma_k$
\be
\Sigma=\text{diag}(\sigma_1,\sigma_2,\cdots),
\ee
that leads to
\be
\rho_\mathbb{A}=\frac{U(\Sigma\Sigma^\dagger)U^\dagger}{\Tr(U(\Sigma\Sigma^\dagger)U^\dagger)}=U\frac{(\Sigma\Sigma^\dagger)}{\Tr(\Sigma\Sigma^\dagger)}U^{-1}\,.
\ee
This way, the normalised eigenvalues $\lambda_k$ of $\rho_\mathbb{A}$ are related to the squares of the singular values $\sigma_k$ of $A$ as $\lambda_k=\sigma^2_k/(\sum_l\sigma^2_l)$. Recall that this procedure is equivalent to introducing a {\it particular basis} in the Hilbert space $\mathcal{H}_\mathbb{A} \otimes \mathcal{H}_\mathbb{B}$ through the \textit{Schmidt decomposition} of $\ket{\psi}$ (often understood as a purification of $\rho_\mathbb{A}$). Indeed, we can write the state $\ket{\psi}$ using the SVD decomposition given in \cref{SVDdecA}, as
\be
\ket{\psi}=\sum_{i,j}A_{ij}\ket{i}_\mathbb{A} \ket{j}_\mathbb{B}=\sum_{i,k,j}U_{ik}\Sigma_{kk}V^\dagger_{kj}\ket{i}_\mathbb{A} \ket{j}_\mathbb{B}\equiv \sum_k\sigma_k\ket{k}_\mathbb{A} \ket{k}_\mathbb{B},
\ee
where we defined the basis vectors 
\be
\ket{k}_\mathbb{A}=\sum_{i}U_{ik}\ket{i}_\mathbb{A},\qquad \ket{k}_\mathbb{B}=\sum_jV^\dagger_{kj}\ket{j}_\mathbb{B},
\ee
that span the basis in $\mathbb{A}$ (and its complement $\mathbb{B}$), and the singular values are also called the \textit{Schmidt coefficients}.

Finally, from the singular values (Schmidt coefficients), we can evaluate a family of R\'enyi entropies indexed by a positive integer $n$ as
\begin{equation}
S^{(n)}_\mathbb{A}=\frac{1}{1-n}\log(\Tr(\rho^n_\mathbb{A}))=\frac{1}{1-n}\log\left(\sum_k\lambda^n_k\right)=\frac{1}{1-n}\log\left(\sum_k\frac{(\sigma^{2}_k)^n}{(\sum_l\sigma^2_l)^n}\right).
\end{equation}
The entanglement (von Neumann) entropy, which can also be obtained as the limit of $n\to 1$ of the R\'enyi family, is defined as
\be
S^{\mathrm E}_\mathbb{A}\coloneq -\Tr(\rho_\mathbb{A}\log(\rho_\mathbb{A}))=-\sum_k\lambda_k\log(\lambda_k)=-\sum_k\frac{\sigma^{2}_k}{(\sum_l\sigma^2_l)}\log\left(\frac{\sigma^{2}_k}{(\sum_l\sigma^2_l)}\right).
\ee
In what follows we generalise this construction to other decompositions of matrix $A$ discussed in the previous sections.
\subsection{Unit-Invariant entanglement entropies} \label{ssec-entropies}
To start with, we consider scale-invariant generalisations of von Neumann entropy. We describe three distinct prescriptions by which a \emph{balanced state} associated to \cref{eq:psi-original} is defined. Next, we evaluate reduced density matrices $\rho_\mathbb{A}$ and entanglement entropies expressed by appropriate normalised singular values. 
\paragraph{Left-UI entanglement entropy.} We first define the balanced state which is associated to \cref{eq:psi-original} by acting with some diagonal, local\footnote{Note that we admit different rescalings for each basis vector.} rescaling operator $L^\psi$ on basis vectors in $\mathcal{H}_\mathbb{A}$
\begin{equation}
 \ket{\psi^\mathrm{L}} \;\coloneq \; \Bigl[L^{\psi}_\mathbb{A}\otimes \mathbb I_\mathbb{B}\Bigr]\,\ket{\psi}, \quad L^\psi_\mathbb{A} = \frac{D^\mathrm{L}}{\sqrt{d_\mathbb{A}}}.
 \label{eq:balanced-left}
\end{equation}
More explicitly, we have
\be
\ket{\psi^{\mathrm L}}=\frac{1}{\sqrt{d_\mathbb{A}}}\sum_{ij}(A)_{ij}D^{\mathrm L}_{ii}\ket{i}_\mathbb{A}\otimes \ket{j}_\mathbb{B}= \frac{1}{\sqrt{d_\mathbb{A}}}\sum_{ij}(D^{\mathrm L}A)_{ij}\ket{i}_\mathbb{A}\otimes \ket{j}_\mathbb{B}.
\ee
Using the definition of the balanced matrix in \cref{balancedA}, we can write this as
\be
\ket{\psi^{\mathrm L}}=\frac{1}{\sqrt{d_\mathbb{A}}}\sum_{ij}(A^{\mathrm L})_{ij}\ket{i}_\mathbb{A}\otimes \ket{j}_\mathbb{B}= \frac{1}{\sqrt{d_\mathbb{A}}}\sum_{k}\sigma^{\mathrm L}_k\ket{k}_\mathbb{A}\otimes \ket{k}_\mathbb{B},
\ee
where, in the second step, we defined the Schmidt vectors (suppressing superscripts of $\mathrm L$ to ease the notation)
\be
\ket{k}_\mathbb{A}=\sum_i (U_0)_{ik}\ket{i}_\mathbb{A},\qquad \ket{k}_\mathbb{B}=\sum_j(V^\dagger_0)_{kj}\ket{j}_\mathbb{B}\,.
\ee
Since the normalised left singular values are simply $\hat{\sigma}^{\mathrm L}_k=\sigma^{\mathrm L}_k/\sqrt{d_\mathbb{A}}$ from \cref{eq-normalised-luirui-values}, we have the Schmidt decomposition 
\be
\ket{\psi^{\mathrm L}}=\sum_{k}\hat{\sigma}^{\mathrm L}_k\ket{k}_\mathbb{A}\otimes \ket{k}_\mathbb{B}, \label{SchDLeft}
\ee
that allows us to define the reduced density matrix of $\mathbb{A}$ evaluated from $\ket{\psi^{\mathrm L}}$ as
\begin{equation}
 \rho^{\mathrm L}_\mathbb{A} \;=\; \mathrm{Tr}_\mathbb{B}\Bigl(\ket{\psi^\mathrm{L}}\bra{\psi^\mathrm{L}}\Bigr)
 \;=\; \sum_{k} \Bigl(\hat{\sigma}_k^\mathrm{L}\Bigr)^2\, \ket{k}_\mathbb{A}\bra{k}_\mathbb{A}.
 \label{eq:rhoA-left}
\end{equation}
Finally, we define \textit{Left-Unit-Invariant entanglement entropy} (LUI entanglement entropy) \(S^{\mathrm{L}}_\mathbb{A}\) of $\rho_\mathbb{A}$ as the \textit{von Neumann entropy} of $\rho^{\mathrm L}_\mathbb{A}$ 
\begin{equation}
 S^{\mathrm{L}}_\mathbb{A}(\rho_\mathbb{A})\coloneq -\Tr(\rho^{\mathrm L}_\mathbb{A}\log(\rho^{\mathrm L}_\mathbb{A}))= -\sum_{k} (\hat{\sigma}_k^\mathrm{L})^2 \log\Bigl[(\hat{\sigma}_k^\mathrm{L})^2\Bigr].
 \label{eq:entropy-left}
\end{equation}
From the Schmidt decomposition in \cref{SchDLeft}, we can easily see that $S^{\mathrm{L}}_\mathbb{A}=S^{\mathrm{L}}_\mathbb{B}$ where $\rho^\mathrm{L}_\mathbb{B}=\Tr_{\mathbb{A}}(\ket{\psi^\mathrm{L}}\bra{\psi^\mathrm{L}})$. Moreover, by construction, LUI entanglement entropy is invariant under the action on the state $\ket{\psi}$ with any \textit{local, diagonal scaling operator} \(D_\mathbb{A}\) on $\mathcal{H}_\mathbb{A}$, and any \textit{local unitary operator} \(U_\mathbb{B} \) on $\mathcal{H}_\mathbb{B}$, i.e., $S^\mathrm{L}$ is the same for $\ket{\psi}$ and $\tilde{\ket{\psi}}=D_\mathbb{A} \otimes U_\mathbb{B} \ket{\psi}$.
\paragraph{Right-UI entanglement entropy.} In complete analogy to LUI above, we define 
\begin{equation}
 \ket{\psi^\mathrm{R}} \;\coloneq \; \Bigl[\mathbb I_\mathbb{A}\otimes R^\psi_\mathbb{B}\Bigr]\,\ket{\psi}, \quad R^\psi_\mathbb{B} 
 = \frac{{{D^\mathrm{R}}}}{\sqrt{d_\mathbb{B}}},
 \label{eq:balanced-right}
\end{equation}
or, explicitly,
\be
\ket{\psi^\mathrm{R}}=\frac{1}{\sqrt{d_\mathbb{B}}}\sum_{ij}(A^{\mathrm R})_{ij}\ket{i}_\mathbb{A}\otimes \ket{j}_\mathbb{B}= \frac{1}{\sqrt{d_\mathbb{B}}}\sum_{k}\sigma^{\mathrm R}_k\ket{k}_\mathbb{A}\otimes \ket{k}_\mathbb{B},\label{SchDRight}
\ee
where we again defined the Schmidt vectors $\ket{k}$ from the SVD decomposition of $A^{\mathrm R}$ in \cref{balancedAR}. The reduced density matrix of $\mathbb{A}$ computed from this state becomes
\begin{equation}
 \rho^\mathrm{R}_\mathbb{A} \;=\; \mathrm{Tr}_\mathbb{B}\Bigl(\ket{\psi^\mathrm{R}}\bra{\psi^\mathrm{R}}\Bigr)
 \;=\; \sum_{k} \Bigl(\hat{\sigma}_k^\mathrm{R}\Bigr)^2\, \ket{k}_\mathbb{A}\bra{k}_\mathbb{A},
 \label{eq:rhoA-right}
\end{equation}
and we define the \textit{Right-Unit-Invariant entanglement entropy} (RUI entanglement entropy) $S^\mathrm{R}_\mathbb{A}$ of $\rho_\mathbb{A}$ as the von Neumann entropy of $\rho^\mathrm{R}_\mathbb{A}$
\begin{equation}
 S^\mathrm{R}_\mathbb{A}(\rho_\mathbb{A})\coloneq -\Tr(\rho^\mathrm{R}_\mathbb{A}\log(\rho^\mathrm{R}_\mathbb{A}))= -\sum_{k} (\hat{\sigma}_k^\mathrm{R})^2 \log\Bigl[(\hat{\sigma}_k^\mathrm{R})^2\Bigr].
 \label{eq:entropy-right}
\end{equation}
Again, the Schmidt decomposition of $\ket{\psi^\mathrm{R}}$ in \cref{SchDRight} implies $S^\mathrm{R}_\mathbb{A}=S^\mathrm{R}_\mathbb{B}$ for the entropy of $\mathbb{B}$ computed for $\rho^\mathrm{R}_\mathbb{B}=\Tr_\mathbb{A}(\ket{\psi^\mathrm{R}}\bra{\psi^\mathrm{R}})$. Also, by construction, RUI entanglement entropy is invariant under the action of $U_\mathbb{A} \otimes D_\mathbb{B}$ on state $\ket{\psi}$, i.e., $S^\mathrm{R}$ is the same for a family of states $\tilde{\ket{\psi}}=U_\mathbb{A} \otimes D_\mathbb{B}\ket{\psi}$ with an arbitrary unitary on $\mathbb{A}$ and a diagonal on $\mathbb{B}$.
\paragraph{Bi-UI entanglement entropy.} Ultimately, we define Bi-Unit-Invariant (BUI) entanglement entropy, taking advantage of the simultaneous left- and right-scaling invariants introduced in \cite{uhlmann-main2}. To this end we define a balanced state 
\begin{equation}
 \ket{\psi^\mathrm{B}} \;=\; \frac{M^\psi_\mathbb{A}\otimes N^\psi_\mathbb{B}}{||A^\mathrm{B}||_\mathbb{F}}\ket{\psi}, \quad M^\psi_\mathbb{A} = D^{\mathrm{B_L}}, \quad N^\psi_\mathbb{B} 
 = D^{\mathrm{B_R}},
 \label{eq:balanced-full}
\end{equation}
that can be written explicitly as
\begin{equation}
\ket{\psi^\mathrm{B}}=\frac{1}{||A^\mathrm{B}||_\mathbb{F}}\sum_{ij}({D}^{\mathrm{B_L}}\,A\,{D}^{\mathrm{B_R}})_{ij}\ket{i}_\mathbb{A}\otimes \ket{j}_\mathbb{B}=\sum_{k}\hat{\sigma}^\mathrm{B}_k\ket{k}_\mathbb{A}\otimes \ket{k}_\mathbb{B},
\end{equation}
where we used the definition in \cref{eq:ui-bal-def} and the SVD decomposition of $A^\mathrm{B}$ to construct the Schmidt decomposition with its normalised singular values given by \cref{NormsingvalAU}.

As before, we can evaluate the reduced density matrices $\rho^\mathrm{B}_{\mathbb{A}}$ and $\rho^\mathrm{B}_{\mathbb{B}}$ by performing partial trace over $\rho^\mathrm{B}=\ket{\psi^\mathrm{B}}\bra{\psi^\mathrm{B}}$. Finally, we define the \textit{Bi-Unit-Invariant entanglement entropy} (BUI entanglement entropy) $S^{\mathrm{B}}_\mathbb{A}$ of $\rho_\mathbb{A}$ as the von Neumann entropy of $\rho^{\mathrm B}_{\mathbb{A}}$
\begin{equation}
 S^{\mathrm{B}}_\mathbb{A}(\rho_\mathbb{A})\coloneq -\Tr(\rho^\mathrm{B}_{\mathbb{A}}\log(\rho^{\mathrm{B}}_{\mathbb{A}})) = -\sum_{k} (\hat{\sigma}_k^\mathrm{B})^2 \log\Bigl[(\hat{\sigma}_k^\mathrm{B})^2\Bigr]. 
 \label{eq:entropy-full}
\end{equation}
Of course we again have the property $S^{\mathrm{B}}_\mathbb{A}=S^{\mathrm{B}}_\mathbb{B}$, and this entropy remains invariant under that action on $\ket{\psi}$ with arbitrary diagonal operators of the form \(D_\mathbb{A}\otimes D'_\mathbb{B}\).
\section{Unit-Invariant SVD entropies for transition matrices} \label{UISVDprepost}
In turn, motivated by \cite{Nakata:2020luh,Parzygnat:2023avh}, we present a construction involving pre- or post-selected states and scaling invariance of a reduced transition matrix. 
\subsection{Review of pseudo and SVD entropies}
In this case, we generally consider two independent, normalised pure states of the form \cref{eq:psi-original} in \(\mathcal{H}=\mathcal{H}_\mathbb{A}\otimes \mathcal{H}_\mathbb{B}\): \(\vert \psi_1 \rangle\) and \(\vert \psi_2 \rangle\), with non-zero overlap $\langle \psi_2 \vert \psi_1 \rangle$. We define the \emph{transition matrix} by
\begin{equation}\label{eq:transition-matrix}
\tau^{1|2} = \frac{\vert \psi_1 \rangle \langle \psi_2 \vert}{\langle \psi_2 \vert \psi_1 \rangle},
\end{equation}
which, in general, is a non-Hermitian operator on \(\mathcal{H}\). From this matrix, we compute a partial trace over subsystem \(\mathbb{B}\) to obtain the reduced transition matrix operator on \(\mathbb{A}\)
\begin{equation}\label{eq:reducedtransmatrix}
\tau_\mathbb{A}^{1|2} = \operatorname{Tr}_\mathbb{B}\bigl(\tau^{1|2}\bigr).
\end{equation}
Pseudo-entropy \cite{Nakata:2020luh} is then defined as a generalisation of the von Neumann entropy for this transition matrix 
\be
S^\mathrm P_\mathbb{A}(\tau_\mathbb{A}^{1|2})=-\Tr(\tau_\mathbb{A}^{1|2}\log(\tau_\mathbb{A}^{1|2})).\label{eq:PseudoEntropy}
\ee
Since the eigenvalues of $\tau_\mathbb{A}^{1|2}$ are in general complex, pseudo-entropy will also have real and imaginary parts. These features have already found interesting applications in dS/CFT \cite{Doi:2022iyj}, and in holographic CFTs, its gravity dual was found to be the minimal area in Euclidean time-dependent AdS spacetimes \cite{Nakata:2020luh}. For some specific choices of the pre- and post-selected states, real and imaginary parts of pseudo-entropy were found to satisfy the Kramers-Kronig relations \cite{Caputa:2024gve}. The imaginary part was also shown to be related to chirality of links in Chern-Simons gauge theory \cite{Caputa:2024qkk}. Nevertheless, the imaginary part of pseudo-entropy makes it hard to interpret it from the quantum-information-theoretic point of view and its operational meaning remains mysterious. For this reason its ``improvement" was proposed in \cite{Parzygnat:2023avh} that defined \textit{SVD entropy} as von Neumann entropy computed from the singular values of $\tau_\mathbb{A}^{1|2}$, which are real. More precisely, we first construct a density matrix
\be
\rho_\mathbb{A}^{1|2}=\frac{\sqrt{(\tau_\mathbb{A}^{1|2})^\dagger \tau_\mathbb{A}^{1|2}}}{\Tr\left[\sqrt{(\tau_\mathbb{A}^{1|2})^\dagger \tau_\mathbb{A}^{1|2}}\right]},
\ee
and define the SVD entropy of $\tau_\mathbb{A}^{1|2}$ as the von Neumann entropy of $\rho_\mathbb{A}^{1|2}$
\be\label{eq:svd-eentropy}
S^{\text{SVD}}_\mathbb{A}(\tau_\mathbb{A}^{1|2})=-\Tr(\rho_\mathbb{A}^{1|2}\log(\rho_\mathbb{A}^{1|2})).
\ee
By construction, the singular values $\sigma_k$ of $\tau_\mathbb{A}^{1|2}=U\Sigma V^\dagger$, in $\Sigma=(\sigma_1,\sigma_2,\ldots,\sigma_{d_\mathbb{A}})$, are related to the normalised eigenvalues $\hat{\lambda}_k$ of $\rho_\mathbb{A}^{1|2}$ by
\be
\hat{\lambda}_k=\frac{\sigma_k}{\sum_l\sigma_l},
\ee
so that
\be
S^{\text{SVD}}_\mathbb{A}(\tau_\mathbb{A}^{1|2})=-\sum_k \hat{\lambda}_k\log(\hat{\lambda}_k)\,.\label{SVDEntropy}
\ee
By definition, SVD entropy is always non-negative and bounded, $0\le S^{SVD}_\mathbb{A}\le \log(d_\mathbb{A})$. Moreover, SVD entropy was given an operational meaning as the average number of Bell pairs distillable from intermediate states that appear between $\ket{\psi_1}$ and $\ket{\psi_2}$ \cite{Parzygnat:2023avh}. In the following, we will employ the UISVD values to provide one more variation of the pseudo- and SVD entropies.
\subsection{Unit-Invariant SVD entropies}
Given the above and the UI singular values discussed in \cref{sec-scale}, it is natural to generalise the SVD entropy from \cref{SVDEntropy} by replacing the SVD eigenvalues of $\tau^{1|2}_\mathbb{A}$ by its left-, right-, or bi-unit-invariant singular values in \cref{LSVA,RSVA,sigmaU}. To proceed systematically, let us first rewrite the transition matrix
\be
\tau^{1|2}_\mathbb{A}=\sum_{ij}(\tau^{1|2}_\mathbb{A})_{ij}\ket{i}_\mathbb{A}\bra{j}_\mathbb{A},
\ee
as a vector using the Choi-Jamiołkowski isomorphism
\be
\ket{\tau^{1|2}_\mathbb{A}}=\sum_{ij}(\tau^{1|2}_\mathbb{A})_{ij}\ket{i}_\mathbb{A}\ket{j}^{\star}_\mathbb{A},\label{ChoiState}
\ee
where $\ket{j}^{\star}_\mathbb{A}$ is a CPT conjugate of $\ket{j}_\mathbb{A}$. 
Then we define generalised left-, right-, and bi- UISVD entropies of transition matrix $\tau^{1|2}_\mathbb{A}$ as von Neumann entropies computed from singular values of 
\be
(\tau^{1|2}_\mathbb{A})^{\mathrm L}=D^{\mathrm L}\tau^{1|2}_\mathbb{A},\qquad (\tau^{1|2}_\mathbb{A})^{\mathrm R}=\tau^{1|2}_\mathbb{A}D^{\mathrm R},\qquad (\tau^{1|2}_\mathbb{A})^\mathrm{B}=D^\mathrm{B_L}\tau^{1|2}_\mathbb{A}D^\mathrm{B_R},
\ee
where the balanced matrices $D^{\mathrm L}, D^{\mathrm R}$, $D^{\mathrm B_{\mathrm L}}$, and $D^{\mathrm B_{\mathrm R}}$ are constructed out of the transition matrix $\tau^{1|2}_\mathbb{A}$ in the same way as in \cref{sec-scale}. Note also that, analogously to the UI entropies for states, we can now think about the diagonal scaling operators above as applied to the Choi state in \cref{ChoiState}.

The SVD values are again computed from
\be
(\rho_\mathbb{A}^{1|2})^\mathrm{I}=\frac{\sqrt{((\tau_\mathbb{A}^{1|2})^\mathrm{I})^\dagger (\tau_\mathbb{A}^{1|2})^\mathrm{I}}}{\Tr\left[\sqrt{((\tau_\mathbb{A}^{1|2})^\mathrm{I})^\dagger (\tau_\mathbb{A}^{1|2})^\mathrm{I}}\right]},\qquad \mathrm{I}=\mathrm{L},\mathrm{R},\mathrm{B}\,.
\ee
We denote the UISVD singular values of $(\tau^{1|2}_\mathbb{A})^\mathrm{I}$ by $\sigma^\mathrm{I}_k$ and the eigenvalues of $(\rho_\mathbb{A}^{1|2})^\mathrm{I}$ by $\hat{\lambda}^\mathrm{I}_k$. They are related by 
\be
\hat{\lambda}^\mathrm{I}_k=\frac{\sigma^\mathrm{I}_k}{\sum_i\sigma^\mathrm{I}_i},\qquad \mathrm{I}=\mathrm{L},\mathrm{R},\mathrm{B}\,.
\ee
Finally, we define the left-, right- or bi-Unit-Invariant SVD entropies of $\tau_\mathbb{A}^{1|2}$ as
\be
S^{\mathrm{I}}_\mathbb{A}(\tau_\mathbb{A}^{1|2})=-\sum_k \hat{\lambda}^\mathrm{I}_k\log(\hat{\lambda}^\mathrm{I}_k)\,.\label{UISVDEntropy}
\ee
An analogous construction should be performed for $\tau_\mathbb{B}^{1|2}$ and, similarly to SVD entropy, in general
\be
S^{\mathrm{I}}_\mathbb{A}(\tau_\mathbb{A}^{1|2})\neq S^{\mathrm{I}}_\mathbb{B}(\tau_\mathbb{B}^{1|2})\,.
\ee
We will study some explicit examples of these quantities in the following sections.

Before we proceed, a few comments are in order. First, let us pick one of the states, say the pre-selected state $\ket{\psi_1}$, and act upon its left or right component with an arbitrary diagonal operator, denoted respectively by $D_\mathbb{A}$ and $D_\mathbb{B}$ 
\begin{equation}\label{eq:local-D-red-trans}
 \ket{\psi_{1_\mathbb{A}}} = D_\mathbb{A}\otimes \mathbb I\ket{\psi_1} \quad \text{and} \quad  \ket{\psi_{1_\mathbb{B}}} = \mathbb I\otimes D_\mathbb{B}\ket{\psi_1}.
\end{equation}
The corresponding transition matrices for such rescaled states are then also rescaled by diagonal factors, and we denote them by
\begin{equation}
\tau^{1_\mathbb{A}|2}_\mathbb{A}= \frac{D_\mathbb{A}}{\Tr{(D_\mathbb{A}\tau_\mathbb{A}^{1|2})}} \tau^{1|2}_\mathbb{A} \quad \text{and} \quad \tau^{1_\mathbb{B}|2}_\mathbb{B}= \frac{D_\mathbb{B}}{\Tr{(D_\mathbb{B}\tau_\mathbb{B}^{1|2})}} \tau^{1|2}_\mathbb{B}. 
\end{equation}
Thus, for pre-selected states, the action of local scaling operators on one of the subsystems and tracing out the complement is equivalent to a scaling of the original transition matrix from the left. Consequently, left- and bi- UISVD entropies will be invariant under such transformations (but R-UISVD will detect them).

Secondly, for the case of post-selected states, any local scaling action results in a scaling on the right 
\begin{equation}
\tau^{1|2_\mathbb{A}}_\mathbb{A}= \tau^{1|2}_\mathbb{A}{{\frac{D_\mathbb{A}^\dagger}{\Tr{(D_\mathbb{A}^\dagger\tau_\mathbb{A}^{1|2})}}}} \quad \text{and} \quad \tau^{1|2_\mathbb{B}}_\mathbb{B}= \tau^{1|2}_\mathbb{B} {{\frac{D^\dagger_\mathbb{B}}{\Tr{(D_\mathbb{B}^\dagger\tau_\mathbb{B}^{1|2})}}}}.
\end{equation}
Therefore RUISVD and BUISVD entanglement measures would remain invariant, while LUISVD entropy would detect it. Further, acting on the same subsystem of both pre- and post-selected states leads to the two-sided scaling of the original reduced transition matrix, upon which only BUISVD entanglement entropy remains invariant. 

In the next section, as a proof of concept, we apply these definitions of unit-invariant entropies to various physical setups.
\section{Applications}\label{sec:Applications}
In this final section we present several natural applications of our UISVD entropies, ranging from random states to Chern-Simons theory and Biorthogonal Quantum Mechanics (BQM).
\subsection{Unit-Invariant entropies for the Haar ensemble} 
As the first illustration of the scale-invariant entanglement entropies introduced above, we consider their behaviour for the \emph{Haar ensemble} of bipartite pure states. Let $\mathcal{H}\cong\mathbb{C}^d$ be a finite-dimensional complex Hilbert space with the standard inner product. A random unit vector $|\psi\rangle\in\mathcal{H}$ is called \emph{Haar random} \cite{Page:1993AverageEntropy,Vershynin2018HDP,Mezzadri:2007Haar,Zyczkowski:2001Induced} if its law is the unique $U(d)$-invariant probability measure $\mu_{\mathrm{Haar}}$ on the unit sphere
\begin{equation}
\mathbb{S}(\mathcal{H})=\{\psi\in\mathcal{H}:\langle\psi|\psi\rangle=1\}.
\end{equation}
Now specialise to a bipartite Hilbert space
\be
\mathcal{H}=\mathcal{H}_{\mathbb{A}}\otimes \mathcal{H}_{\mathbb{B}}
\cong \mathbb{C}^{d_{\mathbb{A}}}\otimes \mathbb{C}^{d_{\mathbb{B}}}
\cong \mathbb{C}^{d_{\mathbb{A}}d_{\mathbb{B}}}.
\ee
Fix orthonormal product bases $\{|i\rangle\}_{i=1}^{d_{\mathbb{A}}}$ for $\mathcal{H}_{\mathbb{A}}$ and $\{|j\rangle\}_{j=1}^{d_{\mathbb{B}}}$ for $\mathcal{H}_{\mathbb{B}}$. Any pure state $|\psi\rangle\in\mathcal{H}$ can be written as
\begin{equation}
|\psi\rangle=\sum_{i=1}^{d_{\mathbb{A}}}\sum_{j=1}^{d_{\mathbb{B}}} A_{ij}\,|i\rangle\otimes|j\rangle,
\end{equation}
where $A\in\mathbb{C}^{d_{\mathbb{A}}\times d_{\mathbb{B}}}$ is the \emph{coefficient matrix} of $|\psi\rangle$ in the chosen product basis. In these coordinates, sampling a Haar-random bipartite pure state is equivalent to sampling a random matrix
\begin{equation}
G\in\mathbb{C}^{d_{\mathbb{A}}\times d_{\mathbb{B}}},\qquad
G_{ij}\stackrel{\mathrm{i.i.d.}}{\sim}\mathcal{N}(0,1)+ i\,\mathcal{N}(0,1),
\end{equation}
and normalising in Hilbert-Schmidt (Frobenius) norm,
\begin{equation}
A=\frac{G}{\sqrt{\mathrm{Tr}(G G^\dagger)}},
\end{equation}
so that $|\psi\rangle=\mathrm{vec}(A)$ is Haar distributed on $\mathbb{S}(\mathcal{H})$. 
\begin{figure}[h!]
    \centering
    \begin{minipage}{0.48\textwidth}
        \centering
       \includegraphics[width=\textwidth]{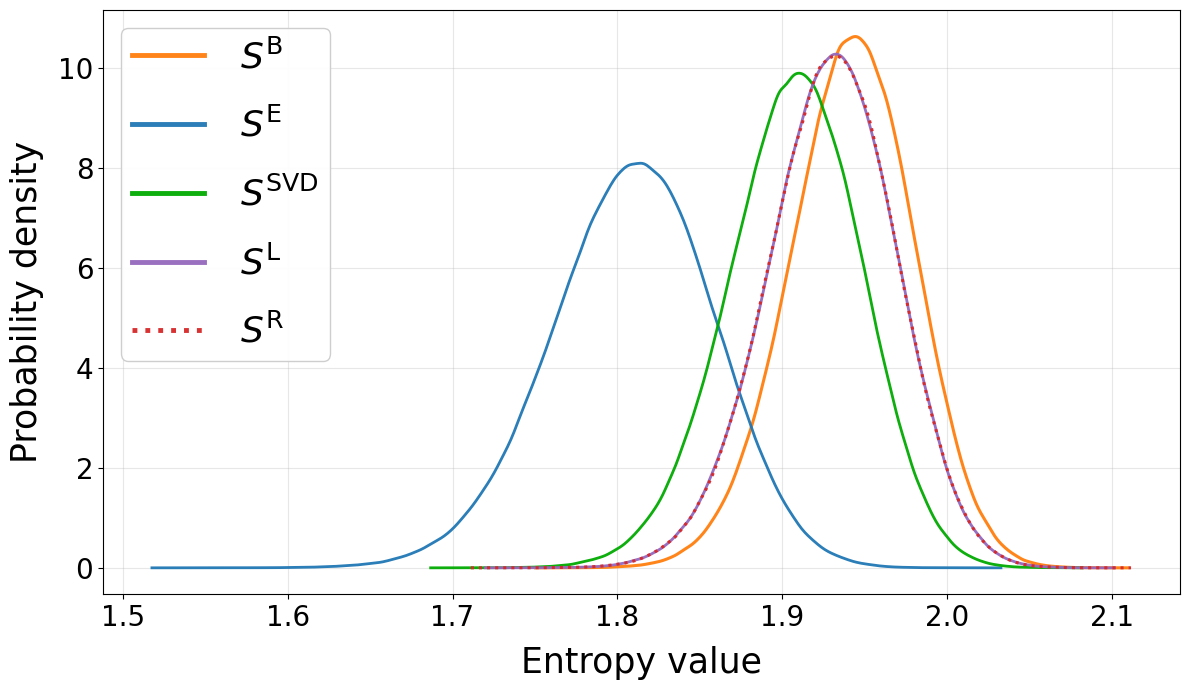}
        \subcaption{Subsystem dimensions $d_\mathbb{A},d_\mathbb{B}=10$}
    \end{minipage}
    \hfill
    \begin{minipage}{0.48\textwidth}
        \centering
        \includegraphics[width=\textwidth]{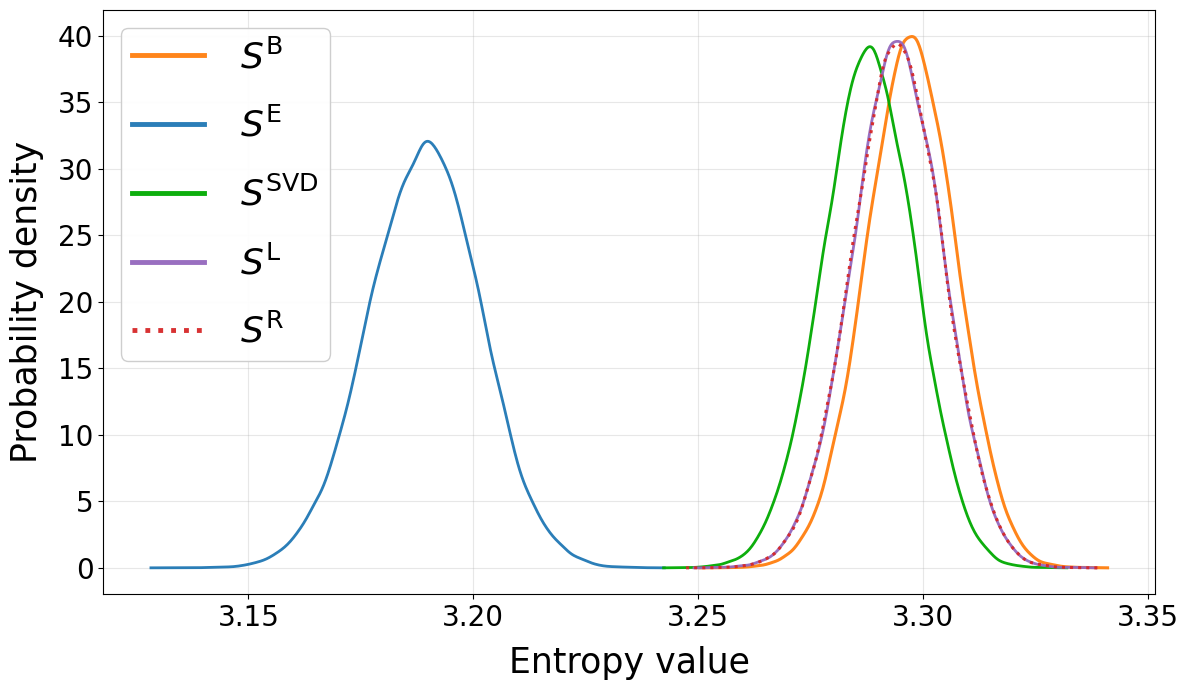}
        \subcaption{Subsystem dimensions $d_\mathbb{A},d_\mathbb{B}=40$}
    \end{minipage}
    \caption{Probability distributions for left-, right-, and bi-unit-invariant entanglement entropies $S^\mathrm{L}, S^\mathrm{R}$ and $S^\mathrm{B}$ for the Haar ensemble. For the sake of completeness we also reproduce the plots of densities of $S^\mathrm{E}$ and ${S}^\mathrm{SVD}$ from \cite{Parzygnat:2023avh}. The mean values of all four SVD-based entropies are systematically shifted to larger values than $S^\mathrm E$ and lie closer to the maximal value $\log d$ as the dimension increases.}
    \label{figure:haar}
\end{figure}

We now generate a large ensemble of Haar-random pure states, construct reduced transition matrices \(\tau_\mathbb{A}^{1|2}\) as in \cref{UISVDprepost} using random pairs and, for each, compute left-, right-, and bi-Unit-Invariant SVD entropies $S^\mathrm{L}, S^\mathrm{R}$, and $S^\mathrm{B}$. For comparison, we also compute the usual bipartite entanglement entropy $S^\mathrm{E}$ from each single Haar-random pure state \(\vert \psi \rangle \in \mathcal{H}\) in the ensemble, as well as SVD entropy $S^\mathrm{SVD}$, following similar analysis in \cite{Parzygnat:2023avh}. Probability distributions for all these measures are shown in \cref{figure:haar}. We conclude that the scale-invariant distributions have a qualitatively similar character to the standard ones; i.e., they have Gaussian-like shapes, while their \emph{averages} are shifted towards higher values with
\begin{equation}
S^\mathrm E < {S}^\mathrm{SVD} < {S}^\mathrm{L} \approx {S}^\mathrm{R} < {S}^\mathrm{B}.
\end{equation}
\subsection{Unit-Invariant entropies in Chern-Simons theory} \label{sec:Chern-simons}
As another interesting and illustrative setup, we apply the entanglement measures introduced in \cref{sec:entanglement-measures} to Chern-Simons theory, following the line of research that applies and tests ideas from quantum information theory in the realm of Chern-Simons and related systems \cite{Balasubramanian_2017,Balasubramanian:2018por,SO3,Caputa:2024qkk}. In particular, in \cite{Caputa:2024qkk}, we observed that several entanglement entropy measures defined from transition matrices exhibit intricate connections with link complement states in Chern-Simons gauge theory. Specifically, the imaginary part of pseudo entanglement entropy $S^\mathrm{P}$, computed from the complex eigenvalues of the reduced transition matrices in \cref{eq:reducedtransmatrix}, provides a new tool to detect the chirality of the underlying links. This holds for any compact gauge group. For the special case of gauge group $U(1)$, it was further shown that the entanglement measure based on the singular values of the reduced transition matrix, called the \textit{excess SVD entanglement entropy} $\Delta S^\mathrm{SVD}$ (constructed using \cref{eq:svd-eentropy}), acts as a \textit{pseudo-metric} on the space of two-component links. These observations naturally lead us to ask: what is the physical interpretation (or a useful application) of the entanglement entropies derived from the remaining three invariants discussed in this work when applied to link complement states? 

To this end, consider a two-component link $\mathscr{L} = \mathcal{K}_1 \bigcup \mathcal{K}_2$ whose link complement state is given by \cite{Balasubramanian_2017,Balasubramanian:2018por}
\begin{equation}
|\mathscr{L}\rangle = \sum_{i,j} V^{\mathscr{L}}_{ij}\, |i\rangle \otimes |j\rangle,
\end{equation}
where the coefficients \( V^{\mathscr{L}}_{ij} \) are \textit{coloured polynomial invariants} (for any gauge group) associated with the link in representations \( i \) and \( j \). In the matrix language, we represent these coefficients as the entries of the matrix
\begin{equation}
A = (V^{\mathscr{L}}_{ij}).
\end{equation}
In this setup, we show that the operation of taking the connected sum and a framing change are respectively examples of diagonal and unitary operations, upon which appropriate measures of our interest are invariant (analogously as in \cref{rescale}).
\begin{figure}[h!]
    \centering
    \includegraphics[width=\linewidth]{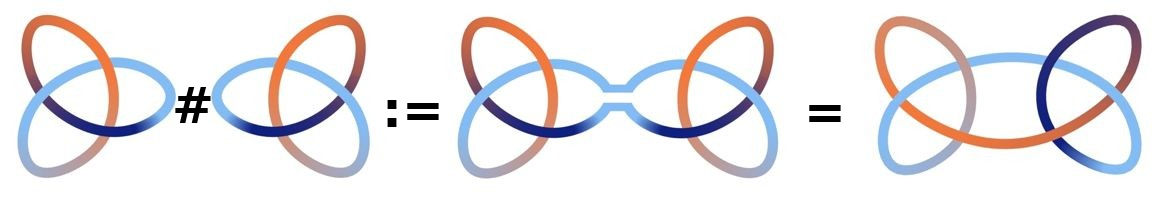}
    \caption{Connected sum of left and right trefoil knots gives a ``granny knot''.}
    \label{fig:granny}
\end{figure}

We begin with the operation of connected sum. Recall that the \textit{connected sum} $\mathcal K \# \mathcal K^\prime$ of two knots $\mathcal K, \mathcal K^\prime$ is defined by first cutting a tiny piece out of each knot so that each one becomes an open string with two loose ends. Then we connect the two knots by joining an end from the first to an end from the second, and joining the remaining two ends, using short connections that do not tangle with anything else. This results in a single closed loop denoted by $\mathcal K \# \mathcal K^\prime$. An example is shown in \cref{fig:granny} using left and right trefoil knots. Recently, this operation has also found a direct, experimentally realised counterpart in real materials \cite{Hall:2025vortexknots}.

Suppose we are given a link $\mathscr{L} = \mathcal{K}_1 \bigcup \mathcal{K}_2$, we perform a connected sum on the first component with a knot \(\mathcal{K}\) and on the second component with a knot \(\mathcal{K}'\), so that we obtain $\mathscr{L}'\coloneq (\mathcal{K}\#\mathcal{K}_1) \bigcup (\mathcal{K}_2\#\mathcal{K}')$. The coloured polynomials (normalised by the unknot) of these knots in representation \( i \) and \( j \) are given by \( V^{\mathcal{K}}_i \) and \( V^{\mathcal{K}'}_j \) respectively. We define the corresponding diagonal operators
\begin{equation}\label{eq:diag-op-CS}
D = \operatorname{diag}\bigl(V^{\mathcal{K}}_0,\, V^{\mathcal{K}}_1,\, V^{\mathcal{K}}_2,\,\dots\bigr)
\quad\text{and}\quad
{D^\prime}^\dagger = \operatorname{diag}\bigl(V^{\mathcal{K}'}_0,\, V^{\mathcal{K}'}_1,\, V^{\mathcal{K}'}_2,\,\dots\bigr).
\end{equation}
Equivalently, in \cref{eq:diag-op-CS} the matrix ${D^\prime}$ is constructed out of the coloured polynomial invariants of the \textit{mirror image} of $\mathcal{K}^\prime$ \cite{Caputa:2024qkk}. The action of taking a connected sum on each component, from the perspective of the coloured polynomial invariants of the link, is equivalent to multiplying the coloured polynomial invariants (normalised by the unknot) of the respective knots on the left and right, while matching the choice of representation. Thus, the new quantum state is in the matrix formulation $A'=DAD'^\dagger$
\begin{equation}
|\mathscr{L}'\rangle = \sum_{i,j} V^{\mathscr{L}'}_{ij} \, |i\rangle \otimes |j\rangle  = \sum_{i,j} V^{\mathcal{K}}_i\, V^{\mathscr{L}}_{ij}\, V^{\mathcal{K}'}_j \, |i\rangle \otimes |j\rangle = \sum_{i,j} A'_{ij}\, |i\rangle \otimes |j\rangle =(D \otimes D')\, |\mathscr{L}\rangle.
\end{equation}
This shows that connected sum with the knots \(\mathcal{K}\) and \(\mathcal{K}'\) on the respective components is equivalent to acting on the initial state with the local operator \( D \otimes D' \). It follows that unit-invariant measures of our interest, which are invariant under diagonal rescalings of this form, are invariant under taking the connected sum.

In turn, we discuss local unitary operations. In the current setup, an interesting class of such operations are changes of \textit{framing} of the component knots. A framing of a knot can be interpreted as the linking number of two boundaries of a ribbon created by thickening that knot trajectory. Formally, a framing is a continuous, nowhere-vanishing normal vector field along the knot, effectively turning it into such a twisted ribbon. The \textit{framing number} is the integer that records how many times this field rotates as one traverses the knot once (with the sign determined by the ambient orientation). In our link-complement states, increasing the framing by $+1$ is implemented by a positive \textit{Dehn twist} on the corresponding boundary torus. If we begin with a two-component link with state $\ket{\mathscr{L}}$, and then vary the framing of the first component knot by $p$ and of the second knot by $q$, then the new link ${\mathscr{L}^\prime}$ encodes the state \cite{Balasubramanian_2017,Leigh:2021trp}
\begin{equation}
 \ket{\mathscr{L}^\prime} = \mathcal{T}^p \otimes \mathcal{T}^q \ket{\mathscr{L}} ,
\end{equation}
where the Dehn twist is effected by the matrices $\mathcal{T}$ (e.g., of the form \cref{eq:modular_ST_matrices} in the case $SU(2)$ Chern-Simons theory that we discuss below). These are exactly the relevant local unitary operators that we have been after. In fact, when we are dealing with compact gauge groups, these $\mathcal T$ matrices are not just unitary but also diagonal, and a change of framing of each component knot leaves all entropies $S^{\mathrm B},S^{\mathrm L},S^{\mathrm R},S^{\mathrm E}$ invariant. 
\begin{figure}[h!]
    \centering
    \includegraphics[width=\linewidth]{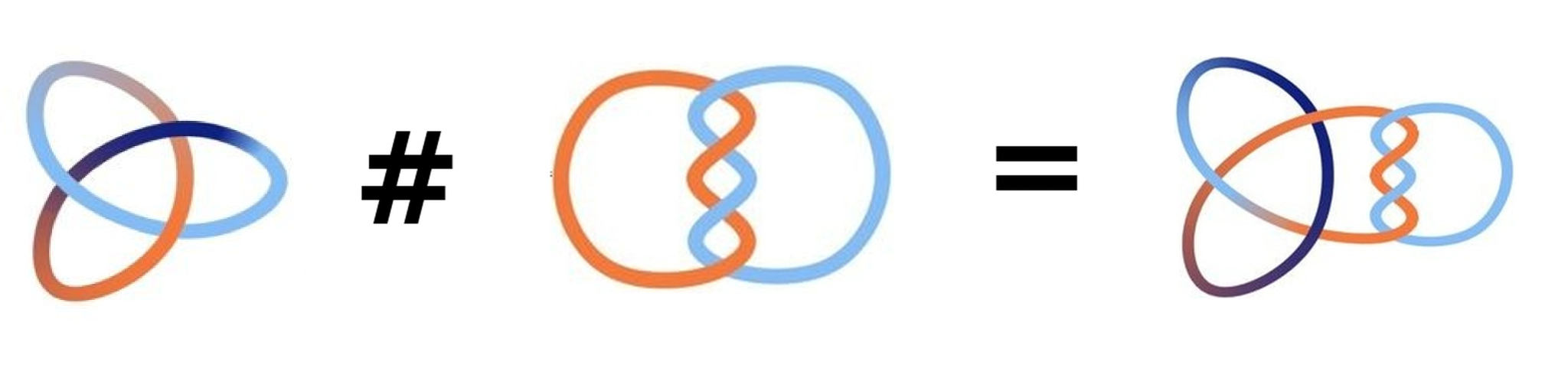}
    \caption{Connected sum of the trefoil knot $3_1$ with $2\N^2_1$} link with $\N=2$, on left side/component resulting in $3_1\#4^2_1$.
    \label{fig:con31-421}
\end{figure}
\paragraph{Example.} As a concrete example, consider $SU(2)$ Chern-Simons theory at \textit{level} $k$, in which case the coloured polynomial invariants are the \textit{coloured Jones polynomials}. We consider $\mathscr{L} = 4^2_1$ link, see \cref{fig:con31-421}, which is the simplest illustrative example in a family of links denoted $2\N^2_1$, whose coloured Jones polynomials and the associated link complement states take the form \cite{Balasubramanian_2017}
\begin{equation}
 |2\N_1^2\rangle = \sum_{i,j}\sum_{l}\left(\mathcal{S}\mathcal{T}^{\N}\mathcal{S}\right)_{0l}\frac{\mathcal{S}_{il}\mathcal{S}_{jl}}{\mathcal{S}_{0l}}|i\rangle \otimes |j\rangle, \quad i,j\in {0,1,\dots,k},
\end{equation}
with \textit{modular matrices}
\begin{equation}\label{eq:modular_ST_matrices}
{\;\mathcal{T}_{i,j}
= e^{\,\pi \sqrt{-1}\,\dfrac{2\,i(i+2)-k}{4\,(k+2)}}\,\delta_{i,j}\;}
, \quad \mathcal{S}_{i,j} = \sqrt{\frac{2}{k+2}}\sin\left(\frac{\pi(i+1)(j+1)}{k+2}\right), 
\end{equation}
which satisfy $\mathcal{S}^2=(\mathcal{ST})^3=\mathbb{I}$. For detailed analysis of this family and large $k$ asymptotics of its entanglement entropies see \cite{Caputa:2024qkk}. Now, to either component of this link we connect the trefoil knot $3_1$ (again see \cref{fig:con31-421}), whose coloured Jones polynomials are given by \cite{Habiro}
\begin{equation}
 V^{3_1}_i=\sum_{l=0}^{i}(-1)^l q^{\frac{l(l+3)}{2}}(q^{\frac{1}{2}} - q^{-\frac{1}{2}})^{2l}\frac{[i+l+1]!}{[i - l]!},\quad [x]=\frac{q^{x/2}-q^{-x/2}}{q^{1/2}-q^{-1/2}},\quad q=e^{\frac{2\pi \sqrt{-1}}{k+2}}. 
\end{equation}
\begin{figure}[h!]
    \centering
    \begin{minipage}{0.49\textwidth}
        \centering
       \includegraphics[width=\textwidth]{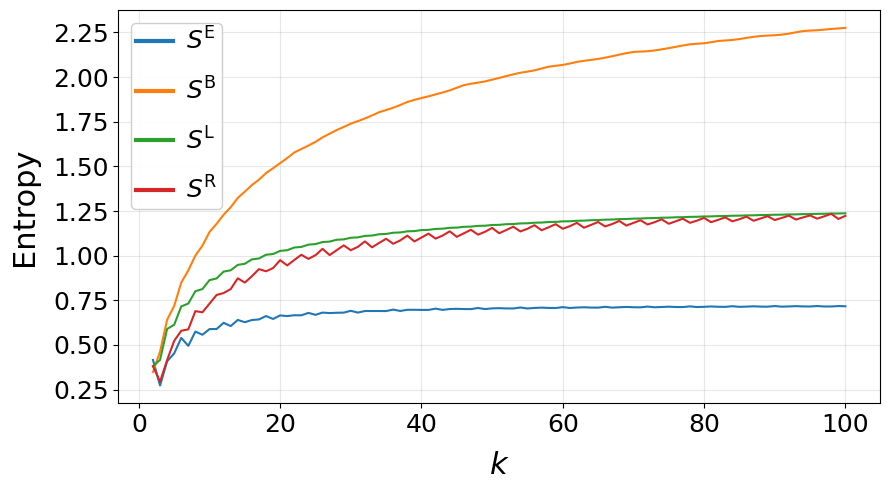}
        \subcaption{$3_1\#4^2_1$}
        \label{31421}
    \end{minipage}
    \hfill
    \begin{minipage}{0.49\textwidth}
        \centering
        \includegraphics[width=\textwidth]{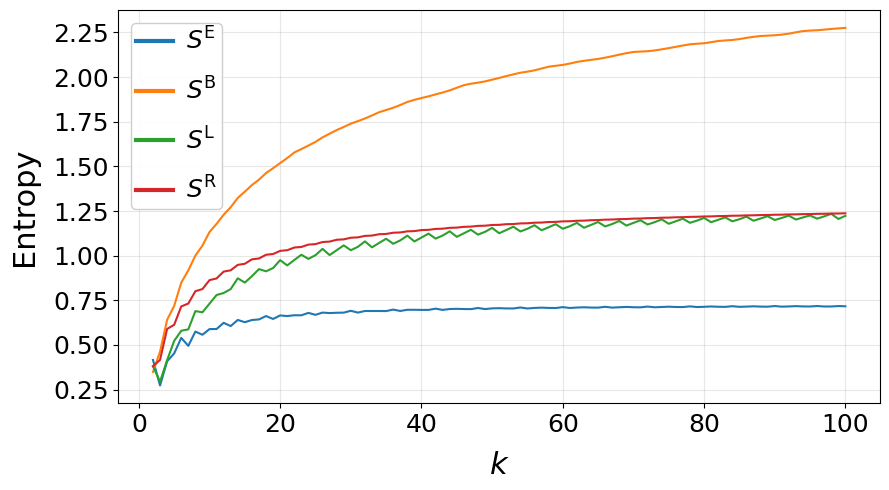}
        \subcaption{$4^2_1\#3_1$}
        \label{42131}
    \end{minipage}
    \caption{Entanglement entropies $S^\mathrm E,S^\mathrm{B},S^\mathrm{L},S^\mathrm{R}$ versus level $k$ of Chern-Simons.}
    \label{fig:5}
\end{figure}

For this setup, in \cref{fig:5} we plot various entropy measures of our interest as defined in \cref{sec:entanglement-measures}. We observe that the $S^\mathrm{B}$ curves in both pictures are invariant; they are equal to the $S^\mathrm{B}$ of the $4^2_1$ link complement state, showing invariance under an arbitrary connected sum on either component of the link. Similarly, the $S^\mathrm{L}$ curve of \cref{31421} and the $S^\mathrm{R}$ curve of \cref{42131} are unaffected by the connected sums on their ``own" sides, and is in fact equal to the $S^\mathrm{L}$/$S^\mathrm{R}$ of the original $4^2_1$ (which are equal to each other in that case). However, we see that the $S^\mathrm{R}$ of \cref{31421} and the $S^\mathrm{L}$ of \cref{42131} are sensitive to, and do indeed detect, the connected sum on the ``other" side. It was shown in \cite{Balasubramanian_2017} for example that usual entanglement entropy $S^\mathrm{E}$ of $\mathcal{K}\#2^2_1$ will not distinguish to which component the knot is connected, but here we see that $S^\mathrm{L}$ and $S^\mathrm{R}$ do in fact detect the ``side" of connected sum in general.

The results in this section hold for any choice of compact gauge group with appropriate choice of coloured polynomials and modular matrices. For example for $SU(N)$ we would use the \textit{coloured HOMFLY-PT polynomials} to construct the link complement states. 
\subsection{UISVD entanglement entropy for Biorthogonal Quantum Mechanics} \label{sec-BQM}
Finally, we study the bi-unit-invariant entanglement entropies in another natural setting: Biorthogonal Quantum Mechanics (BQM). BQM is a generalisation of standard quantum mechanics in which scale transformations act as physical symmetries \cite{Brody:BQM2014}. It is therefore desirable to employ entanglement measures and other observables that remain invariant under such transformations. The unit-invariant measures developed above are ideally suited for this purpose.

We begin by briefly reviewing the essentials of BQM and, in particular, previously studied entanglement measures for biorthogonal states \cite{Brody:BQM2014}, including the standard right-left (RL) construction \cite{Chang:2019nhn,Yang:2024hzt,Herviou:2019yfb}. Here we note that the notation R,L found in the literature has already been used in previous sections to denote the RUI/LUI cases, so in order to avoid a clash, we instead use $\mathcal{R},\mathcal{L}$ to denote the right-left notation of BQM. We then adapt the formalism introduced in the earlier sections to define the BUISVD entanglement entropy in the biorthogonal setting. The key message of this section is that, while the conventional $\mathcal{RL}$ entanglement entropy is generally complex-valued, may take negative values, and can be unbounded, the BUISVD-based alternative is always real, positive, bounded, and fully compatible with the scaling symmetry of BQM. 
\paragraph{Biorthogonal bipartite states.} As the name indicates, the defining feature of BQM is a biorthogonal character of its states \cite{Brody:BQM2014}. We thus consider $\mathcal H = \mathcal H_\mathbb{A} \otimes \mathcal H_\mathbb{B}$ and fix biorthogonal bases on each factor
\begin{equation}
\{|\mathcal{R}^\mathbb{A}_i\rangle\},\ \{\langle \mathcal{L}^\mathbb{A}_i|\} \subset \mathcal H_\mathbb{A},
\qquad
\{|\mathcal{R}^\mathbb{B}_j\rangle\},\ \{\langle \mathcal{L}^\mathbb{B}_j|\} \subset \mathcal H_\mathbb{B},
\end{equation}
satisfying
\begin{equation}
\langle \mathcal{L}^\mathbb{A}_i | \mathcal{R}^\mathbb{A}_{i'}\rangle = \delta_{ii'},
\qquad
\langle \mathcal{L}^\mathbb{B}_j | \mathcal{R}^\mathbb{B}_{j'}\rangle = \delta_{jj'}.
\end{equation}
The corresponding product biorthogonal bases on $\mathcal H$ are
\begin{equation}
|\mathcal{R}^\mathbb{A}_i \mathcal{R}^\mathbb{B}_j\rangle \coloneq |\mathcal{R}^\mathbb{A}_i\rangle \otimes |\mathcal{R}^\mathbb{B}_j\rangle,
\qquad
\langle \mathcal{L}^\mathbb{A}_i \mathcal{L}^\mathbb{B}_j| \coloneq \langle \mathcal{L}^\mathbb{A}_i| \otimes \langle \mathcal{L}^\mathbb{B}_j|.
\end{equation}
A biorthogonal state is a pair $(|\Psi_{\mathcal R}\rangle,\langle\Psi_{\mathcal L}|)$ with expansions
\begin{equation}
|\Psi_{\mathcal R}\rangle
= \sum_{i,j} \Psi_{ij}\,|\mathcal{R}^\mathbb{A}_i \mathcal{R}^\mathbb{B}_j\rangle,
\qquad
\langle\Psi_{\mathcal L}|
= \sum_{k,\ell} \Lambda_{k\ell}\,\langle \mathcal{L}^\mathbb{A}_k \mathcal{L}^\mathbb{B}_\ell|,
\end{equation}
subject to the normalisation condition
\begin{equation}
\langle\Psi_{\mathcal L}|\Psi_{\mathcal R}\rangle
= \sum_{i,j} \Lambda_{ij}\,\Psi_{ij} = 1.
\end{equation}
We collect the right and left coefficients into matrices
$\Psi = (\Psi_{ij})$ and $\Lambda = (\Lambda_{ij})$. The admissible BQM scale transformations are rescalings of the biorthogonal bases \cite{Edvardsson:2022day}
\begin{equation}
|\mathcal{R}^\mathbb{A}_i\rangle \mapsto a_i\,|\mathcal{R}^\mathbb{A}_i\rangle,
\quad
\langle \mathcal{L}^\mathbb{A}_i| \mapsto a_i^{-1}\,\langle \mathcal{L}^\mathbb{A}_i|,
\qquad
|\mathcal{R}^\mathbb{B}_j\rangle \mapsto b_j\,|\mathcal{R}^\mathbb{B}_j\rangle,
\quad
\langle \mathcal{L}^\mathbb{B}_j| \mapsto b_j^{-1}\,\langle \mathcal{L}^\mathbb{B}_j|,
\end{equation}
with all $a_i,b_j \neq 0$. Under this scaling symmetry the coefficient matrices transform as
\begin{equation}
\Psi_{ij} \mapsto \Psi'_{ij} = a_i^{-1} b_j^{-1}\,\Psi_{ij},
\qquad
\Lambda_{ij} \mapsto \Lambda'_{ij} = a_i b_j\,\Lambda_{ij},
\end{equation}
i.e.
\begin{equation}
\Psi \mapsto \Psi' = D_\mathbb{A}^{-1}\Psi D_\mathbb{B}^{-1},
\qquad
\Lambda \mapsto \Lambda' = D_\mathbb{A} \Lambda D_\mathbb{B},
\end{equation}
with $D_\mathbb{A} = \mathrm{diag}(a_i)$ and $D_\mathbb{B} = \mathrm{diag}(b_j)$.
\subsubsection{Biorthogonal entanglement entropies}
For a biorthogonal pure state $(|\Psi_{\mathcal R}\rangle,\langle\Psi_{\mathcal L}|)$ the analogue of a pure-state density operator is \cite{Chang:2019nhn,Yang:2024hzt,Herviou:2019yfb}
\begin{equation}
\rho^{\mathcal{R|L}} \coloneq |\Psi_{\mathcal R}\rangle\langle\Psi_{\mathcal L}|,
\end{equation}
which reproduces BQM expectation values via $\langle O\rangle = \mathrm{Tr}(\rho^{\mathcal{R|L}} O)/\mathrm{Tr}\,\rho^{\mathcal{R|L}}$ for operator $O$. From the perspective of our work, we will simply think about it as a transition matrix (i.e., weak measurement) and extract its unit-invariant singular values the same way as we defined for transition matrices in \cref{UISVDprepost}. The biorthogonal reduced operator on $\mathbb A$ is defined by the partial trace over $\mathbb B$. In the product biorthogonal basis 
\be
\operatorname{Tr}_\mathbb B(O)
=\sum_{j}\Bigl(I_\mathbb A\otimes\langle \mathcal L_\mathbb B^{\,j}|\Bigr)\,O\,\Bigl(I_\mathbb A\otimes|\mathcal R_\mathbb B^{\,j}\rangle\Bigr).
\ee 
Equivalently, for matrix elements \(\bigl(\operatorname{Tr}_\mathbb B O\bigr)_{ik}
=\sum_{j}\,\langle \mathcal L_\mathbb A^{\,i} \mathcal L_\mathbb B^{\,j}\,|\,O\,|\,\mathcal R_\mathbb A^{\,k} \mathcal R_\mathbb B^{\,j}\rangle\). Thus
\begin{equation}
\rho_\mathbb{A}^{\mathcal{R|L}} \coloneq \mathrm{Tr}_\mathbb{B}\,\rho^{\mathcal{R|L}}=\sum_{i,j,k}\psi_{ij}\Lambda_{kj}|\mathcal{R}^\mathbb{A}_{i}\rangle \langle \mathcal{L}^\mathbb{A}_{k}|, \label{rhoRLA}
\end{equation}
or in the matrix form $\rho_\mathbb{A}^{\mathcal{R|L}} = \Psi\,\Lambda^{\mathsf T}$. Under a BQM scale transformation the reduced operator transforms by diagonal similarity
\begin{equation}\label{eq:diag-similarity}
\rho_\mathbb{A}^{\mathcal{R|L}} \mapsto
\rho_\mathbb{A}^{\mathcal{R|L}\,'}
= D_\mathbb{A}^{-1}\,\rho_\mathbb{A}^{\mathcal{R|L}}\,D_\mathbb{A},
\end{equation}
so its eigenvalues are invariant under the admissible rescalings. 

Let $\lambda_k$ denote the eigenvalues of $\rho_\mathbb{A}^{\mathcal{R|L}}$. With the normalisation $\langle\Psi_{\mathcal L}|\Psi_{\mathcal R}\rangle=1$ one has
$\mathrm{Tr}\,\rho_\mathbb{A}^{\mathcal{R|L}} = \sum_k \lambda_k = 1$ but, the same as for transition matrices $\tau^{1|2}_\mathbb{A}$, the $\lambda_k$ can be complex because $\rho_\mathbb{A}^{\mathrm{\mathcal{R|L}}}$ is non-Hermitian. The standard ``biorthogonal entanglement entropy'' used in the BQM literature is the \textit{RL entropy} \cite{Chang:2019nhn,Yang:2024hzt,Herviou:2019yfb} 
\begin{equation}
S^\mathrm{E}_{\mathcal{RL}}(\Psi)
\coloneq - \sum_k \lambda_k \log\lambda_k=S^\mathrm{P}_\mathbb{A}(\rho_\mathbb{A}^{\mathcal{R|L}}),
\end{equation}
which is simply the pseudo entropy \cref{eq:PseudoEntropy} of $\rho_\mathbb{A}^{\mathcal{R|L}}$. This quantity is scale-invariant (since it depends only on the spectrum $\{\lambda_k\}$) but generically complex and not bounded by $\log\dim\mathcal H_\mathbb{A}$. A variant which we denote by $S^\mathrm{E}_{\mathrm{TTC}}(\Psi)$ \cite{Tu:2021tf,Yang:2024ebm}, uses the same spectrum $\{\lambda_k\}$ but replaces $\log\lambda_k$ by $\log|\lambda_k|$ in order to remove branch-cut ambiguities while retaining sensitivity to non-Hermiticity. This TTC entropy is again scale-invariant but, like $S^\mathrm{E}_{\mathcal{RL}}$, is not a non-negative, bounded entropy. In view of these deficiencies, we now see that the UISVD constructions introduced in this paper may prove to be useful.
\paragraph{UISVD entropy in BQM.}\,
We now evaluate the BUISVD entanglement entropy from \cref{UISVDEntropy} directly for the reduced operator $\rho_\mathbb{A}^{\mathcal{R|L}}$ as
\begin{equation}\label{eq:bqm-uisvd-entropy}
S^{\mathrm B}_\mathbb{A}(\rho_\mathbb{A}^{\mathcal{R|L}})
= -\sum_k \hat{\sigma}^\mathrm{B}_k \log \hat{\sigma}^\mathrm{B}_k, \qquad \hat{\sigma}^\mathrm{B}_k = \frac{\sigma^\mathrm{B}_k(\rho_\mathbb{A}^{\mathcal{R|L}})}{\sum_i \sigma^\mathrm{B}_i(\rho_\mathbb{A}^{\mathcal{R|L}})}, 
\end{equation}
where $\hat{\sigma}^\mathrm{B}_k$ are bi-unit-invariant singular values, as in \cref{sigmaU}, of $\rho_\mathbb{A}^{\mathcal{R|L}}$. Since $\sigma_k$ are non-negative and not all zero, the $\hat{\sigma}_k$ form a bona fide probability vector, so
\begin{equation}
0 \le S^{\mathrm B}_\mathbb{A}(\rho_\mathbb{A}^{\mathcal{R|L}})
\le \log r,
\end{equation}
where $r=\mathrm{rank}\,\rho_\mathbb{A}^{\mathcal{R|L}}
\le \text{min}(\dim\mathcal H_\mathbb{A},\dim\mathcal H_\mathbb{B})$. 
Moreover, $S^{\mathrm B}_\mathbb{A}(\rho_\mathbb{A}^{\mathcal{R|L}})$ is invariant under all admissible BQM rescalings of the biorthogonal bases, in consequence of the invariance of the BUI singular values of $\rho_\mathbb{A}^{\mathcal{R|L}}$, as seen in \cref{eq:diag-similarity}. This quantity, however, is asymmetric in the indices and depends on the choice of subsystem to be traced out. 
\paragraph{Example.}
To illustrate the difference between $S^\mathrm{E}_{\mathcal{RL}}$, $S^\mathrm{E}_{\mathrm{TTC}}$, and $S^{\mathrm B}$, consider a two-qubit BQM pure state with a simple analytic dependence on a real parameter $t$. Take $\dim\mathcal H_\mathbb{A} = \dim\mathcal H_\mathbb{B} = 2$ with indices $i,j\in\{0,1\}$ and product biorthogonal basis
$\{|\mathcal{R}^\mathbb{A}_i \mathcal{R}^\mathbb{B}_j\rangle\}$ and dual basis $\{\langle \mathcal{L}^\mathbb{A}_i \mathcal{L}^\mathbb{B}_j|\}$. For each real $t$, define the right and left states by
\begin{equation}
\begin{aligned}
|\Psi_{\mathcal R}(t)\rangle
&= |\mathcal{R}^\mathbb{A}_0 \mathcal{R}^\mathbb{B}_0\rangle
 + |\mathcal{R}^\mathbb{A}_1 \mathcal{R}^\mathbb{B}_1\rangle, \\
\langle\Psi_{\mathcal L}(t)|
&= \tfrac{1}{2}\,\langle \mathcal{L}^\mathbb{A}_0 \mathcal{L}^\mathbb{B}_0|
 - \tfrac{t}{2}\,\langle \mathcal{L}^\mathbb{A}_0 \mathcal{L}^\mathbb{B}_1|
 + \bigl(-\tfrac{t}{2} + i t\bigr)\,\langle \mathcal{L}^\mathbb{A}_1 \mathcal{L}^\mathbb{B}_0|
 + \tfrac{1}{2}\,\langle \mathcal{L}^\mathbb{A}_1 \mathcal{L}^\mathbb{B}_1|,
\end{aligned}
\end{equation}
so that $\langle\Psi_{\mathcal L}(t)|\Psi_{\mathcal R}(t)\rangle=1$ for all $t$.
In the product basis the corresponding coefficient matrices are
\begin{equation}
\Psi =
\begin{pmatrix}
1 & 0\\
0 & 1
\end{pmatrix},
\qquad
\Lambda = \frac{1}{2}
\begin{pmatrix}
1 & -t\\[2pt]
-t + 2i t & 1
\end{pmatrix}.
\end{equation}
In fact, any pair of $\Psi,\Lambda$ which satisfies $\Tr(\Lambda^\mathrm T\Psi)=1$ is a valid BQM state. The biorthogonal reduced operator on $\mathbb A$ is
\begin{equation}
\rho_\mathbb{A}^{\mathcal{R|L}}(t)
= \Psi\,\Lambda^{\mathsf T}
=\frac{1}{2}
\begin{pmatrix}
1 & -t + 2i t\\[2pt]
-t & 1
\end{pmatrix},
\end{equation}
whose eigenvalues and BUI singular values can be explicitly and directly computed using \cref{tab:invariants} and are used to generate the plots in \cref{BQM state plot}. These plots explicitly illustrate that BUISVD entropy $S^{\mathrm B}$ is real and bounded, contrary to $S^\mathrm{E}_{\mathcal{RL}}$ and $S^\mathrm{E}_{\mathrm{TTC}}$.
\begin{figure}[h!]
    \centering
    \includegraphics[width=0.7\linewidth]{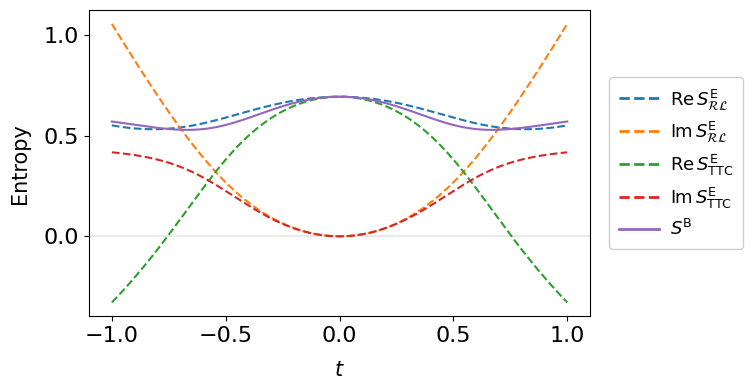}
    \caption{Real and imaginary parts of $S^\mathrm{E}_{\mathcal{RL}}$ and $S^\mathrm{E}_{\mathrm{TTC}}$ together with the BUISVD entropy $S^{\mathrm B}$ as functions of the parameter $t$.}
    \label{BQM state plot}
\end{figure}
\subsubsection{UISVD entropies for transition matrices}
We now turn to BUISVD entanglement measures for pre/post-selected pairs of biorthogonal states and the actual transition matrices in BQM, following recent work \cite{Lu:2025myv,Chen:2025ibe} on entropies for transition matrices in non-Hermitian systems \cite{Tu:2021tf}, and also extending the approach of \cref{UISVDprepost}.
\paragraph{Biorthogonal transition operators.}
On the same bipartite Hilbert space $\mathcal H=\mathcal H_\mathbb{A}\otimes\mathcal H_\mathbb{B}$ with product biorthogonal bases $\{|\mathcal{R}^\mathbb{A}_i \mathcal{R}^\mathbb{B}_j\rangle\}$ and $\{\langle \mathcal{L}^\mathbb{A}_i \mathcal{L}^\mathbb{B}_j|\}$ as before, consider two biorthogonal pure states
\begin{equation}
(|\Psi^{(1)}_{\mathcal{R}}\rangle,\langle\Psi^{(1)}_{\mathcal{L}}|),
\qquad
(|\Psi^{(2)}_{\mathcal{R}}\rangle,\langle\Psi^{(2)}_{\mathcal{L}}|),
\end{equation}
each normalised so that $\langle\Psi^{(k)}_{\mathcal{L}}|\Psi^{(k)}_{\mathcal{R}}\rangle=1$ for $k=1,2$. In the product basis we write
\begin{equation}
|\Psi^{(1)}_{\mathcal{R}}\rangle
= \sum_{i,j}\Psi^{(1)}_{ij}\,|\mathcal{R}^\mathbb{A}_i \mathcal{R}^\mathbb{B}_j\rangle,
\qquad
\langle\Psi^{(2)}_{\mathcal{L}}|
= \sum_{p,q}\Lambda^{(2)}_{pq}\,\langle \mathcal{L}^\mathbb{A}_p \mathcal{L}^\mathbb{B}_q|,
\end{equation}
with coefficient matrices $\Psi^{(1)}=(\Psi^{(1)}_{ij})$ and $\Lambda^{(2)}=(\Lambda^{(2)}_{pq})$. The corresponding right-left transition matrix is
\begin{equation}
\tau^{1|2}
=|\Psi^{(1)}_{\mathcal{R}}\rangle\langle\Psi^{(2)}_{\mathcal{L}}|.
\end{equation}
The reduced transition matrix on $\mathbb A$ is defined by the biorthogonal partial trace over $\mathbb B$,
\begin{equation}\label{BQM-trans-matrix}
\tau^{1|2}_\mathbb{A}
=\mathrm{Tr}_\mathbb{B}\, (\tau^{1|2}),\qquad
(\tau^{1|2}_\mathbb{A})_{ip}
=\langle \mathcal{L}^\mathbb{A}_i|\tau^{1|2}_\mathbb{A}|\mathcal{R}^\mathbb{A}_p\rangle
=\sum_j \Psi^{(1)}_{ij}\,\Lambda^{(2)}_{pj},
\end{equation}
and so as a matrix $\tau^{1|2}_\mathbb{A}=\Psi^{(1)}(\Lambda^{(2)})^{\mathsf T}, $ obtained from the same biorthogonal partial trace rule as in the pure-state case. Under the admissible BQM rescalings the coefficients transform exactly as for a single biorthogonal state. As in the pure-state case, the $\mathbb B$-scale cancels and the surviving $\mathbb A$-scale acts on $\tau_\mathbb{A}$ by diagonal similarity, i.e., $\tau_\mathbb{A} \mapsto \tau^\prime_\mathbb{A}=D_\mathbb{A}^{-1}\tau_\mathbb{A}D_\mathbb{A}$, just like in \cref{eq:diag-similarity}. In particular, the eigenvalues of $\tau_\mathbb{A}$ are invariant under all admissible biorthogonal rescalings, and so are BUI-singular values of $\tau_\mathbb{A}$.
\paragraph{Pseudo entropy.}
To compute the pseudo-entropy \cite{Nakata:2020luh} for the pre/post pair of states $(|\Psi^{(1)}\rangle,|\Psi^{(2)}\rangle)$ along the bipartition $\mathbb A|\mathbb B$ in BQM setting \cite{Chen:2025ibe}, we construct the transition-matrix as in \cref{BQM-trans-matrix}, and denote the set of its normalised eigenvalues by $\{\hat{\lambda}_k\}$. Then the pseudo entropy is simply given by \cref{eq:PseudoEntropy}. Since $\tau^{1|2}_\mathbb{A}$ transforms by diagonal similarity, the spectrum $\{\lambda_k\}$ is invariant under all admissible BQM scales and $S^{\mathrm P}_\mathbb{A}$ is diagonally scale-free. However, the eigenvalues $\lambda_k$ are typically complex and not constrained to lie in $[0,1]$, so $S^{\mathrm P}_\mathbb{A}$ is generically complex and its real part is not bounded by $\log\dim\mathcal H_\mathbb{A}$. 
\paragraph{BUISVD entropy.} To obtain a real, non-negative, and bounded transition entropy that respects the BQM rescaling symmetry, we replace the eigenvalues by the BUI singular values of $\tau_\mathbb{A}$. Let $\sigma^{\mathrm B}_k(\tau^{1|2}_\mathbb{A})$ denote the BUI singular values of $\tau^{1|2}_\mathbb{A}$. These are invariant under arbitrary diagonal left and right scalings. Then we compute the BUISVD entropy \cref{UISVDEntropy} in BQM simply using the BUI singular values of $\tau^{1|2}_\mathbb{A}$ \cref{BQM-trans-matrix}. By construction $S^{\mathrm B}_\mathbb A(\tau^{1|2}_\mathbb{A})$ is invariant under all admissible BQM scales on $\mathbb A$ and $\mathbb B$, and satisfies the bounds
\begin{equation}
0\le S^{\mathrm B}
\le\log r,
\end{equation}
where $r=\mathrm{rank}\,\tau_\mathbb{A} \le\min(\dim\mathcal H_\mathbb{A},\dim\mathcal H_\mathbb{B})$. Moreover $ S^{\mathrm B}=0$ if and only if $\tau^{1|2}_\mathbb{A}$ has rank one. In particular, if the transition operator factorises across the $\mathbb A|\mathbb B$ with a rank-one factor on $\mathbb A$, then $\tau^{1|2}_\mathbb{A}$ has rank one and $ S^{\mathrm B}=0$. In this sense $S^{\mathrm B}$ is a diagonally scale-free analogue of the SVD-based entropies.
\paragraph{Example.}
Consider two qubits with Hilbert space \(\mathcal H = \mathbb C^2\otimes \mathbb C^2\), along with Pauli matrices $\sigma^\alpha$ and define their actions on the two sites by
\begin{equation}
\sigma^\alpha_1 = \sigma^\alpha\otimes \mathbb I_2,\qquad
\sigma^\alpha_2 = \mathbb I_2\otimes\sigma^\alpha.
\end{equation}
An example of a PT-symmetric non-Hermitian XX Hamiltonian with a staggered imaginary field \cite{Korff_2008,Suzuki:2016xx} is
\begin{equation}
H(g,h)
= \sigma^x_1\sigma^x_2
+ i g \bigl(-\sigma^z_1+\sigma^z_2\bigr)
+ h\bigl(\sigma^x_1+\sigma^x_2\bigr),
\qquad g,h\in\mathbb R,\quad |h|\ll1 .
\end{equation}
For $h=0$ this reduces to the unperturbed model with spectrum $E_n(g)\in\{\pm1,\pm\sqrt{1-4g^2}\}$, while for small $h\neq0$ the Hamiltonian remains PT-symmetric but the eigenvalues deform smoothly. For each such $(g,h)$ we biorthogonally diagonalise $H(g,h)$
\begin{equation}
H(g,h)\,|\Psi^{(n)}_{\mathcal{R}}(g,h)\rangle
= E_n(g,h)\,|\Psi^{(n)}_{\mathcal{R}}(g,h)\rangle,
\qquad
\langle\Psi^{(n)}_{\mathcal{L}}(g,h)|\,H(g,h)
= \langle\Psi^{(n)}_{\mathcal{L}}(g,h)|\,E_n(g,h),
\end{equation}
and normalise each eigenpair so that $\langle\Psi^{(n)}_{\mathcal{L}}(g,h)|\Psi^{(n)}_{\mathcal{R}}(g,h)\rangle=1$. Fix an eigenbranch $n$ at a reference point and follow it continuously in $g$ (for fixed $h$); then from any two values $g_1,g_2$ we obtain two eigenpairs $(|\Psi^{(n)}_{\mathcal{R}}(g_1,h)\rangle,\langle\Psi^{(n)}_{\mathcal{L}}(g_1,h)|)$ and $(|\Psi^{(n)}_{\mathcal{R}}(g_2,h)\rangle,\langle\Psi^{(n)}_{\mathcal{L}}(g_2,h)|)$. We regard these as the correlated pre- and post-selected BQM states for each pair of $g$ values. With the bipartition $\mathbb A|\mathbb B$, where $\mathbb A$ is the first qubit and $\mathbb B$ the second, we form the transition operator
\begin{equation}
\tau^{1|2}=|\Psi^{(n)}_{\mathcal{R}}(g_1,h)\rangle\langle\Psi^{(n)}_{\mathcal{L}}(g_2,h)|,
\end{equation}
and its reduced transition matrix $\tau^{1|2}_\mathbb A=\mathrm{Tr}_\mathbb B\, (\tau^{1|2})$. The corresponding pseudo entropy $S^{\mathrm P}$ and BUISVD transition entropy $ S^{\mathrm B}$ are obtained from the eigenvalues and BUI-singular values of $\tau^{1|2}_\mathbb A$ respectively. Since this $\tau^{1|2}_\mathbb A$ is a $2\times2$ matrix, the entropies can be calculated directly using \cref{tab:invariants}, and they are plotted in \cref{fig:BQM-transition} with $h=0.05$. When $h=0$ the chosen branch lives in a two-dimensional sector and $\tau^{1|2}_\mathbb A$ is diagonal, so $S^{\mathrm B}$ saturates and has constant value $\log 2$.
\begin{figure}[h!]
\centering
    \includegraphics[width=\linewidth]{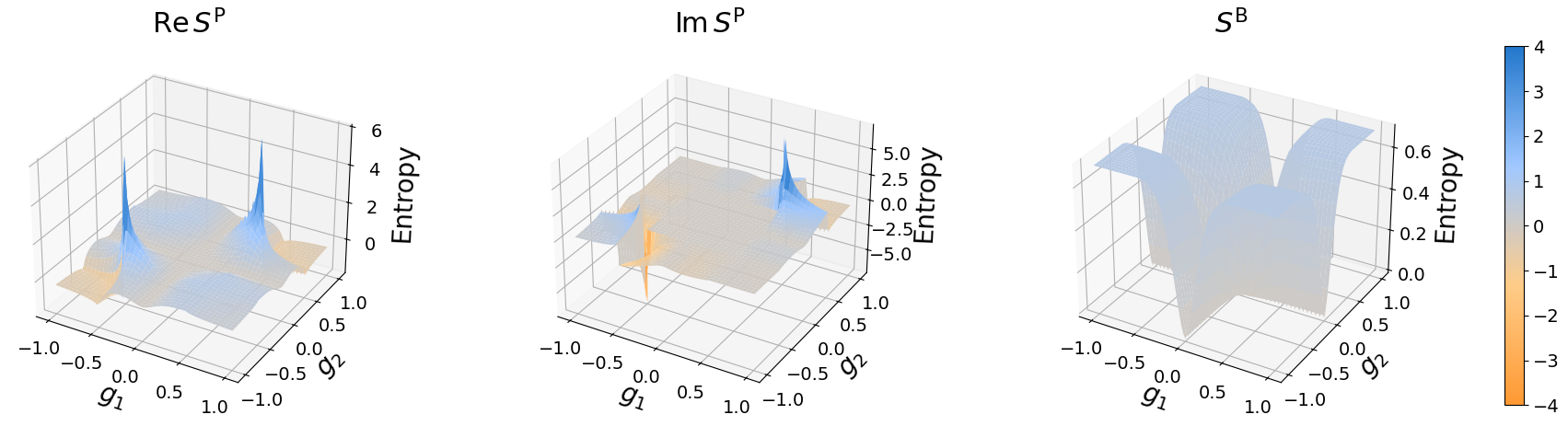}
    \caption{The pre- and post-selected states are taken from a continuous eigenbranch of the two-qubit Hamiltonian, evaluated at $g_1$ and $g_2$ respectively. When $g_1 = g_2$ the states are the same, and thus $S^\mathrm P = S^\mathrm{E}_{\mathcal{RL}}$ and BUISVD entropy $S^{\mathrm B}$ equals the usual BUISVD entanglement entropy of state in \cref{eq:bqm-uisvd-entropy}.}
    \label{fig:BQM-transition}
\end{figure}

\section{Conclusions and outlook}\label{Conclusions}
The central aim of this work was to introduce unit-invariant singular values into quantum information theory and to use them to define generalised von Neumann-type entropies. After reviewing the basic construction and illustrating it with simple examples, we defined left-, right-, and bi-unit-invariant entanglement entropies for both quantum states and transition matrices. We computed these quantities explicitly and demonstrated their utility across a range of physically diverse and relevant examples, including random matrices, Chern-Simons link-complement states, and Biorthogonal Quantum Mechanics. Our selection of examples was deliberately non-exhaustive and intended primarily as a proof of concept.

We expect these mathematically well-defined tools and entanglement constructions to admit many interesting and non-trivial applications in current developments across several areas of physics and mathematics. For instance, there has been significant recent interest in classifying phases and symmetries of general mixed states, notably via SymTFT-based approaches \cite{Schafer-Nameki:2025fiy}. Since these constructions proceed through purifications and the study of their structural properties, our invariant-based entropies may offer a novel perspective on this problem. Moreover, it will be important to clarify the geometric interpretation of UISVD entropies in the context of (A)dS/CFT. As non-Hermitian quantum mechanics is beginning to play an increasingly prominent role in these settings, we anticipate that UISVD entropies will prove particularly useful there as well; see, for example, \cite{Harper:2025lav}. Further, interesting properties of other types of entropies based on rescaled singular values were discussed e.g., in \cite{PhysRevA.87.032308}. On the other hand, various types of singular values underlying our work should be of interest in other contexts in random matrix theory. For example, the normalisations that arise in our diagonal scale functions are similar to those that arise in the Sinkhorn algorithm and related bistochastic matrices \cite{Cappellini_2009,Idel2016MatrixScaling} and deserve further scrutiny. 

In many other settings, diagonal rescalings reflect a genuine freedom of units, normalisation, or residual gauge, rather than a numerical artefact. Such scalings appear across a broad range of problems such as multi-channel and multiterminal transport and hydrodynamics \cite{Lucas2015HydrodynamicTransport,BrandnerSeifert2013Multiterminal,JacquodWhitneyMeairButtiker2012Tenfold,LucasDavisonSachdev2016WeylHydro}, tensor-network and string-net constructions where $F$-data is defined only up to gauge choices \cite{LevinWen2004StringNet,LanWen2013BulkEdgeStringNet,FrancuzLootensVerstraeteDziarmaga2021VariationalMPO}, correlator-based spectroscopy where operator normalisations are conventional \cite{DetmoldEndres2015SignalNoise,BlossierDellaMorteVonHippelMendesSommer2009GEVP}, and calibration problems with per-sensor gain ambiguities \cite{Smirnov2011RIME2}. Since the unit-consistent inverse is designed to respect diagonal changes of units, our diagonal-invariant UISVD entropy diagnostics provide a concise, convention-independent way to summarise ``spectral'' structure in these contexts. We leave analysis of all these ideas for future work.
\paragraph{Note:} All data generated for the various plots in this work and the associated code are publicly available in \href{https://doi.org/10.58132/YUXFDX}{10.58132/YUXFDX} and comply with the FAIR principles.
\acknowledgments
We would like to thank Arindam Bhattacharjee, Nils Carqueville, Giuseppe Di Giulio, Pedram Karimi, Mi{\l}osz Panfil, Arani Paul, Souradeep Purkayastha, Pichai Ramadevi, Marko Stošić, Jeffrey Uhlmann, and Karol {\.Z}yczkowski for useful comments or email conversations. We would also like to thank Anuvab Kayaal for the knot/link graphics. This work was supported by the NCN Sonata Bis 9 grant no. 2019/34/E/ST2/00123 and OPUS grant no. 2022/47/B/ST2/03313 funded by the National Science Centre, Poland. PC is supported by the ERC Consolidator grant (number: 101125449/acronym: QComplexity). Views and opinions expressed are, however, those of the authors only and do not necessarily reflect those of the European Union or the European Research Council. Neither the European Union nor the granting authority can be held responsible for them.
\appendix
\section{Technical details on Unit-Invariant Singular Values}\label{sec:2x2matrix-derivation}
In this appendix we collect some explicit results used in computations in the paper. First, we derive unit-invariant singular values, and for the sake of completeness also eigenvalues and ordinary singular values, for any complex $2\times 2$ matrix 
\begin{equation}
A = \begin{pmatrix}
a_1 & a_2 \\
a_3 & a_4
\end{pmatrix}, \qquad a_i\neq0.
\end{equation}
We will also denote general $2\times2$ diagonal and unitary matrices by
\begin{eqnarray}
D=\left(
\begin{array}{cc}
 \delta_1 & 0 \\
 0 & \delta_2 \\
\end{array}
\right)\,,\qquad U=e^{i\alpha}\left(
\begin{array}{cc}
 e^{i \beta } \cos (\theta ) & e^{i \gamma } \sin (\theta ) \\
 -e^{-i \gamma } \sin (\theta ) & e^{-i \beta } \cos (\theta ) \\
\end{array}
\right)\,,
\end{eqnarray}
where $\delta_1,\delta_2\in\mathbb{C}\backslash\{0\}$ and $\alpha,\beta,\gamma,\theta\in\mathbb{R}$, such that
\begin{eqnarray}
DAU=\left(
\begin{array}{cc}
 \delta_1(a_1 e^{i (\alpha +\beta) } \cos (\theta )-a_2 e^{i (\alpha -\gamma) }
 \sin (\theta)) & \delta_1(a_1 e^{i (\alpha +\gamma) } \sin (\theta )+a_2 e^{i(
 \alpha -\beta) } \cos (\theta ) \\
 \delta_2(a_3 e^{i (\alpha +\beta) } \cos (\theta )-a_4 e^{i (\alpha -\gamma) }
 \sin (\theta ) & \delta_2(a_3 e^{i (\alpha +\gamma) } \sin (\theta )+a_4 e^{i(
 \alpha -\beta) } \cos (\theta )) \\
\end{array}\right),\nonumber\\\label{eq:DAU}
\end{eqnarray}
\begin{eqnarray}
UAD=\left(
\begin{array}{cc}
 \delta_1 \left(a_1 e^{i (\alpha +\beta) } \cos (\theta )+a_3 e^{i (\alpha +\gamma) } \sin
 (\theta )\right) & \delta_2 \left(a_2 e^{i (\alpha +\beta) } \cos (\theta )+a_4 e^{i
 (\alpha +\gamma) } \sin (\theta )\right) \\
 \delta_1 \left(a_3 e^{i (\alpha -\beta) } \cos (\theta )-a_1 e^{i (\alpha -\gamma) } \sin
 (\theta )\right) & \delta_2 \left(a_4 e^{i (\alpha -\beta) } \cos (\theta )-a_2 e^{i
 (\alpha -\gamma) } \sin (\theta )\right) 
\end{array}\right),\nonumber\\ \label{eq:UADp}
\end{eqnarray}
\begin{eqnarray}
DAD'=\left(
\begin{array}{cc}
 \delta_1\delta'_{1}a_1 & \delta_1\delta'_{2}a_2 \\
 \delta_2\delta'_{1}a_3 & \delta_2\delta'_{2}a_4 \\
\end{array}\right)\,.\label{eq:DADp}
\end{eqnarray}
\paragraph{Eigenvalues:} The eigenvalues of \(A\) are found from
\begin{equation}
\det(A-\lambda I)= (a_1-\lambda)(a_4-\lambda)-a_2a_3=0,
\end{equation}
yielding
\begin{equation}\label{eq:eigen2x2}
\lambda_{\pm} = \frac{a_1+a_4\pm\sqrt{(a_1-a_4)^2+4a_2a_3}}{2}.
\end{equation}
\paragraph{Singular Values:} The singular values of  $A$ are the square roots of the eigenvalues of
\begin{equation}
A^\dagger A = \begin{pmatrix}
|a_1|^2 + |a_3|^2 & \overline{a_1}a_2 + \overline{a_3}a_4 \\
a_1\overline{a_2} + a_3\overline{a_4} & |a_2|^2 + |a_4|^2
\end{pmatrix}.
\end{equation}
Since
\begin{equation}
\operatorname{tr}(A^\dagger A) = |a_1|^2 + |a_2|^2 + |a_3|^2 + |a_4|^2, \quad \det(A^\dagger A) = |a_1a_4 - a_2a_3|^2,
\end{equation}
the eigenvalues of \(A^\dagger A\) are obtained using \cref{eq:eigen2x2} as \(\lambda^\prime_{\pm}\), which gives us the singular values of \(A\) as \(\sigma_\pm=\sqrt{\lambda^\prime_{\pm}}\)
\begin{equation}\label{eq:singularvalues2x2}
\sigma_{\pm} = \sqrt{\frac{|a_1|^2 + |a_2|^2 + |a_3|^2 + |a_4|^2 \pm \sqrt{(|a_1|^2 + |a_2|^2 + |a_3|^2 + |a_4|^2)^2 - 4|a_1a_4 - a_2a_3|^2}}{2}}.
\end{equation}
\paragraph{LUI-singular values:}
The left diagonal scale function matrix in \cref{eq:LDSF} is defined using the Euclidean norm for each row
\begin{equation}
D^\mathrm{L}{(A)} = \text{diag} \left( \frac{1}{\sqrt{|a_1|^2 + |a_2|^2}}, \frac{1}{\sqrt{|a_3|^2 + |a_4|^2}} \right).
\end{equation}
The balanced matrix \( A^\mathrm{L} \) is given by
\begin{equation}
A^\mathrm{L} = D^\mathrm{L} \cdot A = 
\begin{bmatrix} 
\frac{a_1}{\sqrt{|a_1|^2 + |a_2|^2}} & \frac{a_2}{\sqrt{|a_1|^2 + |a_2|^2}} \\
\frac{a_3}{\sqrt{|a_3|^2 + |a_4|^2}} & \frac{a_4}{\sqrt{|a_3|^2 + |a_4|^2}}
\end{bmatrix}.
\end{equation}
The singular values of the balanced matrix \( A^\mathrm{L} \), which are the left unit-invariant singular values \( \sigma^\mathrm{L}_{\pm} \) of \( A \), are (using \cref{eq:singularvalues2x2})
\begin{equation}\label{eq:lui-2x2}
\sigma^\mathrm{L}_{\pm} = \sqrt{ 1 \pm \sqrt{1 - \frac{|a_1a_4 - a_2a_3|^2}{(|a_1|^2 + |a_2|^2)(|a_3|^2 + |a_4|^2)}}}.
\end{equation}
We can verify that these values are invariant under the transformation given in \cref{eq:DAU}. In fact, we can prove more generally that for any matrix $A\in\mathbb C^{m\times n}$ with any nonsingular diagonal \(D\in\mathbb C^{m\times m}\) and any unitary \(U\in\mathbb C^{n\times n}\)
\be
\sigma^{\mathrm L}(A)=\sigma^{\mathrm L}(DAU).
\ee
\begin{proof}
Let \(B\coloneq DAU\). Row \(i\) of \(B\) equals
\be
B(i,:)=d_i\,A(i,:)\,U.
\ee
Now take Euclidean norms and use the facts that \(\|c x\|_2=|c|\|x\|_2\) and that \(\|xU\|_2=\|x\|_2\) for unitary \(U\),
\be
\|B(i,:)\|_2=\|d_i\,A(i,:)\,U\|_2
=|d_i|\,\|A(i,:)\,U\|_2
=|d_i|\,\|A(i,:)\|_2
=|d_i|\,r_i.
\ee
So if \(r_i>0\), then \((D^{\mathrm L}(B))_{ii}=1/(|d_i|r_i)\). Using the diagonal expression above,
\be
B^{\mathrm L}=D^{\mathrm L}(B)\,B
=\mathrm{diag}\Big(\tfrac{1}{|d_i|r_i}\Big)\,\mathrm{diag}(d_i)\,A\,U.
\ee
Factor the product of the two diagonal matrices
\be
\mathrm{diag}\Big(\tfrac{1}{|d_i|r_i}\Big)\,\mathrm{diag}(d_i)
=
\underbrace{\mathrm{diag}\Big(\tfrac{d_i}{|d_i|}\Big)}_{\coloneq \Phi\ \text{diagonal unitary}}
\;\underbrace{\mathrm{diag}\Big(\tfrac{1}{r_i}\Big)}_{=D^{\mathrm L}(A)},
\ee
and the \(r_i=0\) case is handled analogously with \((D^{\mathrm L}(B))_{ii}=1=(D^{\mathrm L}(A))_{ii}\) from the definition. Therefore
\be
B^{\mathrm L}=\Phi\,D^{\mathrm L}(A)\,A\,U=\Phi\,A^{\mathrm L}\,U.
\ee
\(\Phi\) is unitary (diagonal with entries of modulus \(1\)), and \(U\) is unitary, hence using property of usual singular values
\be
\sigma(B^{\mathrm L})=\sigma(\Phi A^{\mathrm L}U)=\sigma(A^{\mathrm L}).
\ee
By definition \(\sigma^{\mathrm L}(A)=\sigma(A^{\mathrm L})\) and \(\sigma^{\mathrm L}(B)=\sigma(B^{\mathrm L})\). Thus \(\sigma^{\mathrm L}(DAU)=\sigma^{\mathrm L}(A)\). An entirely analogous argument holds for RUISVD using columns instead of rows.
\end{proof}
\paragraph{RUI-singular values:} The right diagonal scale function matrix \cref{eq:RDSF} is defined using the Euclidean norm for each column
\begin{equation}
D^\mathrm{R}{(A)} = \text{diag} \left( \frac{1}{\sqrt{|a_1|^2 + |a_3|^2}}, \frac{1}{\sqrt{|a_2|^2 + |a_4|^2}} \right).
\end{equation}
The balanced matrix \( A^\mathrm{R} \) is given by
\begin{equation}
A^\mathrm{R} = A \cdot D^\mathrm{R} = 
\begin{bmatrix} 
\frac{a_1}{\sqrt{|a_1|^2 + |a_3|^2}} & \frac{a_2}{\sqrt{|a_2|^2 + |a_4|^2}} \\
\frac{a_3}{\sqrt{|a_1|^2 + |a_3|^2}} & \frac{a_4}{\sqrt{|a_2|^2 + |a_4|^2}}
\end{bmatrix}.
\end{equation}
The singular values of the balanced matrix \( A^\mathrm{R} \), which are the right unit-invariant singular values \( \sigma^\mathrm{R}_{\pm} \) of \( A \), are (using \cref{eq:singularvalues2x2})
\begin{equation}\label{eq:rui-2x2}
\sigma^\mathrm{R}_{\pm} = \sqrt{ 1 \pm \sqrt{1 - \frac{|a_1a_4 - a_2a_3|^2}{(|a_1|^2 + |a_3|^2)(|a_2|^2 + |a_4|^2)}}}.
\end{equation}
We can verify that these values are invariant under transformation given in \cref{eq:UADp}. For the general case, the proof is analogous to the LUI proof as given above.
\paragraph{BUI-singular values:} Define 
\begin{equation}
K = \log(\mathrm{Abs}[A]) = 
\begin{bmatrix}
\ln|a_1| & \ln|a_2|\\
\ln|a_3| & \ln|a_4|
\end{bmatrix},
\quad
J_2
=
\frac12
\begin{bmatrix}
1 & 1\\
1 & 1
\end{bmatrix},
\end{equation}
where $J_p = \big[\frac{1}{p}\big]_{p\times p}$. So for an $m \times n$ matrix we compute
\begin{equation}
Q 
\;=\;
J_m\,K\,J_n
\;-\;
K\,J_n
\;-\;
J_m\,K.
\end{equation}
Here, dimensions $m=n=2$, and we set 
\begin{equation}
P[A]
\;=\;
\exp(Q) \;=\; \begin{bmatrix}
|a_1^{-\tfrac34}a_2^{-\tfrac14}a_3^{-\tfrac14}a_4^{\tfrac14}| \quad \quad
&
|a_1^{-\tfrac14}a_2^{-\tfrac34}a_3^{\tfrac14}a_4^{-\tfrac14}|
\\[7pt]
|a_1^{-\tfrac14}a_2^{\tfrac14}a_3^{-\tfrac34}a_4^{-\tfrac14}| \quad \quad
&
|a_1^{\tfrac14}a_2^{-\tfrac14}a_3^{-\tfrac14}a_4^{-\tfrac34}|
\end{bmatrix}.
\end{equation}
Writing $P[A]$ in the form $\mathbf{x}\,\mathbf{y}^\top$ we
pick the top-left entry (index $(1,1)$) as the reference scale. Namely,
\begin{equation}
x_1 = 1,\quad
x_2 = \frac{P[A]_{2,1}}{P[A]_{1,1}},\qquad
y_1 = P[A]_{1,1},\quad
y_2 = P[A]_{1,2}.
\end{equation}
Then
\begin{equation}
D^{\mathrm{B_L}} = \mathrm{diag}\bigl(x_1,\;x_2\bigr), \quad
{D^{\mathrm{B_R}}} = \mathrm{diag}\bigl(y_1,\;y_2\bigr).
\end{equation}
A direct exponent check yields
\begin{equation}
D^{\mathrm{B_L}}
=
\begin{bmatrix}
1 & 0\\
0 & \sqrt{\bigg|\dfrac{a_1\,a_2}{a_3\,a_4}}\bigg|
\end{bmatrix},
\quad
{D^{\mathrm{B_R}}}
=
\begin{bmatrix}
|a_1^{-\tfrac34}a_2^{-\tfrac14}a_3^{-\tfrac14}a_4^{\tfrac14}| & 0\\[3pt]
0 & |a_1^{-\tfrac14}a_2^{-\tfrac34}a_3^{\tfrac14}a_4^{-\tfrac14}|
\end{bmatrix}.
\end{equation}
The balanced matrix $A^\mathrm{B}=D^{\mathrm{B_L}} A D^{\mathrm{B_R}}$ is
\begin{equation}
\begingroup
\setlength{\arraycolsep}{1.8em}
A^{\mathrm B}=
\begin{pmatrix}
a_1\;|a_1|^{-3/4}\;\bigl|a_4/(a_2 a_3)\bigr|^{1/4}
&
a_2\;|a_2|^{-3/4}\;\bigl|a_3/(a_1 a_4)\bigr|^{1/4}
\\[8pt]
a_3\;|a_3|^{-3/4}\;\bigl|a_2/(a_1 a_4)\bigr|^{1/4}
&
a_4\;|a_4|^{-3/4}\;\bigl|a_1/(a_2 a_3)\bigr|^{1/4}
\end{pmatrix}
\endgroup
\end{equation}
The BUI-singular values of $A$ are the singular values of $A^\mathrm{B}$ (using \cref{eq:singularvalues2x2})
\begin{equation}\label{eq:ui-2x2-balanced-sv}
\sigma^{\mathrm B}_{\pm}
=
\sqrt{
\left(
\sqrt{\bigg|\frac{a_1 a_4}{a_2 a_3}\bigg|}
+
\sqrt{\bigg|\frac{a_2 a_3}{a_1 a_4}\bigg|}
\right)
\ \pm\
\sqrt{
\left(
\sqrt{\bigg|\frac{a_1 a_4}{a_2 a_3}\bigg|}
+
\sqrt{\bigg|\frac{a_2 a_3}{a_1 a_4}\bigg|}
\right)^{\!2}
-
\frac{\big|a_1 a_4 - a_2 a_3\big|^{2}}{\big|a_1 a_2 a_3 a_4\big|}
}
}\,,
\end{equation}
which further simplifies to
\be
\sigma^{\mathrm B}_{\pm}
=
\left|
\sqrt[4]{\frac{a_1 a_4}{a_2 a_3}}
\ \pm\
\sqrt[4]{\frac{a_2 a_3}{a_1 a_4}}
\right|,
\qquad \text{ when } a_i\in\mathbb{R}_{+}.
\ee
We can verify that these values are invariant under transformation given in \cref{eq:DADp}. More generally, for any given matrix $A\in\mathbb{C}^{m\times n}$, the routine below computes positive diagonal scalings $D^{\mathrm{B_L}}=\mathrm{diag}(d^\ell)$ and $D^{\mathrm{B_R}}=\mathrm{diag}(d^r)$ such that \(A^{\mathrm B} \;=\; D^{\mathrm {B_L}}\,A\,D^{\mathrm {B_R}}\) has approximately unit geometric mean per row and per column in magnitude. Then as per the definition the \emph{BUI singular values} of $A$ are the singular values of this balanced matrix $A^{\mathrm B}$. This is the implementation used to generate the BUISVD curves in our numerical plots. It is a Python adaptation of the MATLAB code provided in \cite[Appendix C]{uhlmann-main2}.

\begin{tcolorbox}[
  enhanced,
  breakable,
  colback=black!1,
  colframe=black!35,
  boxrule=0.6pt,
  arc=2mm,
  left=1.5mm,right=1.5mm,top=1.0mm,bottom=1.0mm,
  title={Python code for BUISVD},
  fonttitle=\bfseries
]
\lstset{
  basicstyle=\ttfamily\small,
  keywordstyle=\bfseries,
  showstringspaces=false,
  breaklines=true,
  columns=fullflexible,
  frame=none
}
\begin{lstlisting}[language=Python, literate={~}{{\textasciitilde}}1]
import numpy as np

def dscale(A, tol=1e-15, max_iters=None, return_iters=False):
    A = np.asarray(A)
    if A.ndim != 2:
        raise ValueError("A must be a 2D array.")

    m, n = A.shape
    it = 0

    if np.iscomplexobj(A):
        A = A.astype(np.complex128, copy=False)
        A_abs = np.abs(A)
        A_phase = np.zeros((m, n), dtype=np.complex128)
        np.divide(A, A_abs, out=A_phase, where=(A_abs != 0))
    else:
        A = A.astype(np.float64, copy=False)
        A_abs = np.abs(A)
        A_phase = np.sign(A)
        A_phase[A == 0] = 0.0

    L = np.zeros((m, n), dtype=np.float64)
    M = np.ones((m, n), dtype=np.float64)

    nz = (A_abs > 0.0)
    L[nz] = np.log(A_abs[nz])
    M[~nz] = 0.0

    r = np.sum(M, axis=1)
    c = np.sum(M, axis=0)

    u = np.zeros((m, 1), dtype=np.float64)
    v = np.zeros((1, n), dtype=np.float64)

    dx = 2.0 * tol
    while dx > tol and (max_iters is None or it < max_iters):
        it += 1

        idx = (c > 0)
        p = np.sum(L[:, idx], axis=0) / c[idx]
        L[:, idx] = L[:, idx] - (p[None, :] * M[:, idx])
        v[:, idx] = v[:, idx] - p
        dx = np.mean(np.abs(p))

        idx = (r > 0)
        p = np.sum(L[idx, :], axis=1) / r[idx]
        L[idx, :] = L[idx, :] - (p[:, None] * M[idx, :])
        u[idx, 0] = u[idx, 0] - p
        dx = dx + np.mean(np.abs(p))

    d_l = np.exp(u)
    d_r = np.exp(v)
    A_B = A_phase * np.exp(L)

    if return_iters:
        return A_B, d_l, d_r, it
    return A_B, d_l, d_r

def bui_singular_values(A, tol=1e-15, max_iters=None):
    A_B, _, _ = dscale(A, tol=tol, max_iters=max_iters)
    return np.linalg.svd(A_B, compute_uv=False)

def buisvd(A, tol=1e-15, max_iters=None):
    return bui_singular_values(A, tol=tol, max_iters=max_iters)
\end{lstlisting}
\end{tcolorbox}

\section{Results from Random Matrix Theory}\label{sec:tools}
Here we collect known results from Random Matrix Theory and further prove certain Lemmas, which are required for proving \cref{thm:lui-rui-support,thm:uisvd-support,thm:lui-rui-wigner,thm:ui-wigner}.
\paragraph{Notation and conventions.} We write $A^\dagger$ for the conjugate transpose of a complex matrix $A$, and define its operator (spectral) norm by \( \|A\|_{\op} \coloneq \sqrt{\lambda_{\max}(A^\dagger A)}\), and equivalently, $\|A\|_{\op}=\sigma_{\max}(A)$. We write $\Prob(\cdot)$ for probability. We use $a.s.$ to denote \textit{almost surely} which describes an event that occurs with probability 1. We make use of the Borel-Cantelli lemma, which states that ``if the sum of the probabilities of the events is finite, then the probability that infinitely many of them occur is 0". We use Landau notation where $x_n=o(1)$ means $x_n\to0$ as $n\to\infty$, and $x_n=O(p_n)$ means there exist constants $C>0$ and $n_0$ and a non-negative reference sequence $(p_n)_{n\ge1}$ (e.g., $\sqrt{\log n/n}$) such that $|x_n|\le C p_n$ for all $n\ge n_0$. Throughout, we use notation $c,c_i,\hat c,C,C_i,\hat C$ to denote positive constants (independent of $n$ unless explicitly indicated) whose values may change from line to line. For an $n\times n$ Hermitian matrix $M$ with eigenvalues $\lambda_1(M),\dots,\lambda_n(M)\in\mathbb R$, its \emph{empirical spectral distribution (ESD)} is
\be
\mathrm{ESD}(M)\coloneq \frac1n\sum_{i=1}^n \delta_{\lambda_i(M)}.
\ee
We write $\mu_n \Rightarrow \mu$ for weak convergence of probability measures on $\mathbb R$. For a probability measure $\mu$ on $\mathbb R$, its \emph{Stieltjes transform} is
\be
m_\mu(z)\coloneq \int_{\mathbb R}\frac{1}{x-z}\,d\mu(x),\qquad z\in\mathbb C_+\coloneq \{z\in\mathbb C:\mathrm{Im}(z)>0\}.
\ee
For a Hermitian $M$, we write $m_M(z)\coloneq \frac1n\mathrm{Tr}(M-zI)^{-1}$; then $m_M$ is the Stieltjes transform of $\mathrm{ESD}(M)$. We use the Orlicz norms
\be
\|X\|_{\psi_1}\coloneq \inf\{t>0:\mathbb E e^{|X|/t}\le 2\},\qquad
\|X\|_{\psi_2}\coloneq \inf\{t>0:\mathbb E e^{|X|^2/t^2}\le 2\}.
\ee
A (real) random variable is \emph{sub-Gaussian} if $\|X\|_{\psi_2}<\infty$ and \emph{sub-exponential} if $\|X\|_{\psi_1}<\infty$. For complex $X$, we interpret sub-Gaussian/sub-exponential in the same way with $|X|$ (equivalently, it suffices that the real and imaginary parts, i.e., $\mathrm{Re}(X)$ and $\mathrm{Im}(X)$ are sub-Gaussian). Throughout, for complex entries we use the convention $\mathbb E|X_{ij}|^2=1$ when we say ``variance $1$".  For $\beta\in\{1,2\}$ we use the shorthand $g_\beta$ for a standard real/complex Gaussian, i.e., $g_1\sim\mathcal N(0,1)$ and $g_2\sim\mathcal{CN}(0,1)$ (circular, with $\mathbb E|g_2|^2=1$). We denote $X\stackrel{d}{=}Y$ to mean equality in distribution, i.e.\ $\Prob(X\in B)=\Prob(Y\in B)$ for every Borel set $B$. 

Finally, in this appendix we prove the results for LUISVD and BUISVD. The proofs for RUISVD are entirely analogous to the ones for LUISVD, by simply choosing the column norm $\|A_n(:,j)\|_2$ instead of the row norm $r_i\coloneq \|A_n(i,:)\|_2$ and following the steps of the proofs as done below.
\begin{theorem}[{\cite{Vershynin2018HDP}}]\label{thm:bernstein-subexp}
If $X_1,\dots,X_n$ are independent, mean-zero, sub-exponential random variables with $c_0\coloneq \max_i\|X_i\|_{\psi_1}$, then for all $u\ge0$,
\begin{equation}
\Prob\!\left(\frac{1}{n}\Big|\sum_{i=1}^n X_i\Big|\ge u\right)\le 2\exp\!\left[-nc\,\min\!\left(\frac{u^2}{c_0^2},\,\frac{u}{c_0}\right)\right].
\end{equation}
\end{theorem}
\begin{theorem}[{\cite{BaiSilverstein2010}}]\label{thm:MP-law}
If $X_n$ is $n\times n$ with i.i.d.\ entries of mean $0$, variance 1, then the empirical spectral distribution of $\frac1n X_nX_n^\dagger$ converges almost surely to Marchenko-Pastur $\mathrm{MP}(1,1)$.
\end{theorem}
\begin{theorem}[{\cite{BaiSilverstein2010}}]\label{thm:MP-edge}
Under the assumptions of \cref{thm:MP-law} and finite fourth moment,
\begin{equation}
\lambda_{\max}\!\Big(\frac1n X_nX_n^\dagger\Big)\xrightarrow{\mathrm{a.s.}}4,\qquad
\lambda_{\min}\xrightarrow{\mathrm{a.s.}}0.
\end{equation}
\end{theorem}
\begin{theorem}[{\cite{Tao2012TRMT}}]\label{thm:wigner-law}
Let $H_n$ be GOE/GUE with the $1/\sqrt n$ scaling. Then the eigenvalue ESD of $H_n$ converges almost surely to the semicircle law on $[-2,2]$, and $\|H_n\|_{\op}\to 2$ almost surely.
\end{theorem}
\begin{lemma}[{\cite{Billingsley:Convergence,Tao2012TRMT}}]\label{thm:pushforward}
Let $T_n$ be random $n\times n$ matrices and $H_n\coloneq T_nT_n^\dagger$. If $\mathrm{ESD}(H_n)\Rightarrow F$ \emph{almost surely} with $\mathrm{supp}(F)\subset[0,b]$, then the empirical measure of singular values of $T_n$ converges almost surely to the pushforward $F\circ f^{-1}$ under $f(x)=\sqrt{x}$, supported on $[0,\sqrt b]$.
\end{lemma}
\begin{lemma}[{\cite{BordenaveRMTNotes2019}}]\label{lem:esd-continuity}
Let $M,N\in\mathbb{C}^{n\times n}$ be Hermitian and set $E\coloneq N-M$. For $z=x+\mathrm{i}y$ with $y>0$, define the resolvents $R_M(z)\coloneq (M-zI)^{-1}$, $R_N(z)\coloneq (N-zI)^{-1}$ and their normalised traces $m_M(z)\coloneq \frac1n\mathrm{Tr} R_M(z)$, $m_N(z)\coloneq \frac1n\mathrm{Tr} R_N(z)$. Then
\begin{equation}\label{eq:resolvent-bound}
|m_N(z)-m_M(z)|\ \le\ \frac{\|E\|_{\op}}{y^2}.
\end{equation}
Consequently, if $\|E_n\|_{\op}\to 0$ and $m_{M_n}(z)\to m(z)$ for every fixed $z\in\mathbb{C}_+$, then $m_{N_n}(z)\to m(z)$ for every fixed $z\in\mathbb{C}_+$ as well. In particular, $M_n$ and $N_n$ have the same limiting ESD (by uniqueness of Stieltjes transforms).
\end{lemma}
\begin{lemma}\label{lem:log-constants}
Let $g_1\sim \mathcal N(0,1)$ and let $g_2\sim \mathcal{CN}(0,1)$ where $\mathcal{CN}(0,1)$ denotes the standard circular complex normal, i.e.\ $g_2=X+\mathrm{i}Y$ with $X,Y\overset{\mathrm{i.i.d.}}{\sim}\mathcal N(0,\tfrac12)$ (equivalently, $\mathbb E|g_2|^2=1$). Then
\begin{equation}
c_{\beta=1}\coloneq \mathbb E\log|g_1|=-\tfrac12(\gamma+\log 2),
\qquad
c_{\beta=2}\coloneq \mathbb E\log|g_2|=-\tfrac12\gamma,
\end{equation}
where $\gamma$ is the Euler-Mascheroni constant. Numerically, $c_{\beta=1}\approx -0.6351814227$ and $c_{\beta=2}\approx -0.2886078324$.
\end{lemma}
\begin{proof}
This is a known result, but to make it explicit we give a proof. Let $\Gamma$ be the gamma function and let $\psi(z)\coloneq \Gamma'(z)/\Gamma(z)$ denote the digamma function.
Let $G\sim\mathrm{Gamma}(\alpha,\theta)$ with \textit{shape} $\alpha>0$ and \emph{rate} $\theta>0$, i.e.\ with density
\be
f_G(t)=\frac{\theta^\alpha}{\Gamma(\alpha)}\,t^{\alpha-1}e^{-\theta t}\mathbf 1_{t>0}.
\ee
A standard identity for the Gamma law in this parametrisation is
\be\label{eq:elog-gamma}
\mathbb E[\log G]=\psi(\alpha)-\log\theta,
\ee
see, e.g., \cite[Ex.~1.25]{JorgensenLabouriau2012}. For $g_1\sim\mathcal N(0,1)$, the change of variables $V\coloneq g_1^2$ gives the density
\be
f_V(v)=\frac{1}{\sqrt{2\pi v}}e^{-v/2}\mathbf 1_{v>0},
\ee
which is $\mathrm{Gamma}(\tfrac12,\tfrac12)$ (shape $\tfrac12$, rate $\tfrac12$). Hence
\be
\mathbb E\log|g_1|=\tfrac12\,\mathbb E\log(g_1^2)=\tfrac12\big(\psi(\tfrac12)+\log 2\big).
\ee
For $g_2=X+\mathrm{i}Y$ with $X,Y\overset{\mathrm{i.i.d.}}{\sim}\mathcal N(0,\tfrac12)$, write $X=\frac{X_0}{\sqrt2}$ and $Y=\frac{Y_0}{\sqrt2}$ with $X_0,Y_0\overset{\mathrm{i.i.d.}}{\sim}\mathcal N(0,1)$. Then
\be
|g_2|^2=X^2+Y^2=\frac{X_0^2+Y_0^2}{2}.
\ee
Since $S\coloneq X_0^2+Y_0^2$ has density $f_S(s)=\frac12 e^{-s/2}\mathbf 1_{s>0}$ (i.e.\ $\chi^2_2=\mathrm{Gamma}(1,\tfrac12)$), it follows that $|g_2|^2$ has density $e^{-u}\mathbf 1_{u>0}$, i.e.\ $|g_2|^2\sim\mathrm{Gamma}(1,1)$. Therefore
\be
\mathbb E\log|g_2|=\tfrac12\,\mathbb E\log(|g_2|^2)=\tfrac12\,\psi(1).
\ee
Using $\psi(1)=-\gamma$ and $\psi(\tfrac12)=-\gamma-2\log 2$ \cite[\S5.4(ii), Eqns.~(5.4.12)-(5.4.13)]{DLMF} gives $c_{\beta=1}=-\tfrac12(\gamma+\log 2)$ and $c_{\beta=2}=-\tfrac12\gamma$. The numerical values follow from $\gamma\approx 0.5772156649$ \cite[\S5.2(ii), Eq.~(5.2.3)]{DLMF}.
\end{proof}
\subsection{Derived Lemmas}
\begin{lemma}\label{lem:row-col-conc}
Let $A_n=(1/\sqrt n)\,X_n$ where the entries $\{X_{ij}\}$ of $X_n$ are i.i.d., mean $0$, variance $1$, and sub-Gaussian. Define the row Euclidean ($\ell^2$) norm
\begin{equation}
r_i\coloneq \|A_n(i,:)\|_2.
\end{equation}
There exists a constant $C>0$ (independent of $n$) such that, for all sufficiently large $n$,
\begin{equation}\label{eq:row-col-conc}
\Prob\!\left(\max_{1\le i\le n}\big|r_i-1\big|\geq C\sqrt{\frac{\log n}{n}}\right)< n^{-10}.
\end{equation}
\end{lemma}
\begin{proof} Fix $i$. Since $A_n=(1/\sqrt n)X_n$, with $\E(X_{ij})=0$ and $\E|X_{ij}|^2-|\E X_{ij}|^2=\text{Var}(X_{ij})=1$, we get
\begin{equation}
r_i^2=\|A_n(i,:)\|_2^2=\frac1n\sum_{j=1}^n |X_{ij}|^2
\quad\text{and}\quad
\E [r_i^2] = \frac1n\sum_{j=1}^n \E|X_{ij}|^2 = 1.
\end{equation}
Define centered variables $Z_{ij}\coloneq |X_{ij}|^2-1$ so that
\begin{equation}
r_i^2-1=\frac1n\sum_{j=1}^n Z_{ij}.
\end{equation}
Because $X_{ij}$ is sub-Gaussian, $Z_{ij}$ is \emph{sub-exponential} \cite{Vershynin2018HDP}. Since the $Z_{ij}$ are i.i.d.\ in $j$, we have $c_0\coloneq \|Z_{ij}\|_{\psi_1}$ for all $j$, and hence $\max_{1\le j\le n}\|Z_{ij}\|_{\psi_1}= c_0$. So \cref{thm:bernstein-subexp} yields
\begin{equation}\label{eq:bernstein-applied}
\Prob\!\left(\big|r_i^2-1\big|\geq u\right)
=\Prob\!\left(\left|\frac1n\sum_{j=1}^n Z_{ij}\right|\geq u\right)
\le 2\exp\!\left[-c\,n\,\min\!\left(\frac{u^2}{c_0^2},\,\frac{u}{c_0}\right)\right].
\end{equation}
We now {choose} the deviation level
\begin{equation}
u\coloneq c_1\sqrt{\frac{\log n}{n}},
\end{equation}
with a constant $c_1>0$ to be specified. For $n$ large enough, $u/c_0\le 1$, so
\begin{equation}
\min\!\left(\frac{u^2}{c_0^2},\frac{u}{c_0}\right)=\frac{u^2}{c_0^2}.
\end{equation}
Substitute this $u$ into \cref{eq:bernstein-applied}:
\begin{equation}
\Prob\!\left(\big|r_i^2-1\big|\geq c_1\sqrt{\tfrac{\log n}{n}}\right)
\ \le\ 2\exp\!\left[-\,c\,n\cdot\frac{u^2}{c_0^2}\right]
\ =\ 2\exp\!\left[-\,\frac{cc_1^2}{c_0^2}\,\log n\right].
\end{equation}
Because constants are free to pick in a tail bound, we choose $c_1$ large enough such that\footnote{For notational convenience we choose 12, this can be any sufficiently large value.} 
\begin{equation}\label{eq:b.22ref}
\frac{cc_1^2}{c_0^2}\ge 12 
\quad\Longrightarrow\quad
\Prob\!\left(\big|r_i^2-1\big|\geq c_1\sqrt{\tfrac{\log n}{n}}\right)
\ \le\ 2\exp(-12\log n)
\ =\ 2n^{-12}.
\end{equation}
Define events $E_i\coloneq \big\{|r_i^2-1|\geq c_1\sqrt{\frac{\log n}{n}}\big\}$. Then
\begin{equation}
\Big\{\max_{1\le i\le n}|r_i^2-1|\geq c_1\sqrt{\tfrac{\log n}{n}}\Big\}
\ =\ \bigcup_{i=1}^n E_i.
\end{equation}
By the union bound $\Prob\big(\cup_{i=1}^n E_i\big)\le\sum_{i=1}^n \Prob(E_i)$ and the previous estimate,
\begin{equation}\label{eq:probability-bound-n10}
\Prob\!\left(\max_{1\le i\le n}|r_i^2-1|\geq c_1\sqrt{\tfrac{\log n}{n}}\right)
\ \le\ \sum_{i=1}^n 2n^{-12}
\ =\ 2n^{-11} < n^{-10}
\end{equation}
for all sufficiently large $n$. Now, 
\begin{equation}
|r_i-1|=\frac{|r_i^2-1|}{\,r_i+1\,}.
\end{equation}
Therefore, {always} (since $r_i\ge 0$),
\begin{equation}
|r_i-1|\ \le\ |r_i^2-1|.
\label{eq:basic-sqrt-ineq}
\end{equation}
So we have the implication
\begin{equation}
\Big\{\max_{1\le i\le n}|r_i-1|\geq c_1\sqrt{\tfrac{\log n}{n}}\Big\}
\ \subseteq\
\Big\{\max_{1\le i\le n}|r_i^2-1|\geq c_1\sqrt{\tfrac{\log n}{n}}\Big\}.
\end{equation}
Therefore, using \cref{eq:probability-bound-n10}
\begin{equation}\label{eq:b.28ref}
\Prob\!\left(\max_{1\le i\le n}|r_i-1|\geq C\sqrt{\tfrac{\log n}{n}}\right)
\ <\ n^{-10}.
\end{equation}
\end{proof}
\begin{corollary}\label{cor:DL-DR-expansion}
For all sufficiently large $n$, with probability at least $1-Cn^{-10}$ there exist diagonal $\Delta^\mathrm{L}_{n}$ related to the left diagonal scale function matrix defined in \cref{eq:LDSF} such that
\begin{equation}
D^\mathrm{L}_{n}=I+\Delta^\mathrm{L}_{n},\qquad
\|\Delta^\mathrm{L}_{n}\|_{\op}\ \le\ C\sqrt{\tfrac{\log n}{n}}.
\end{equation}
\end{corollary}
\begin{proof}
By Lemma~\ref{lem:row-col-conc}, with probability at least $1-Cn^{-10}$ we have
\be
\max_i |r_i-1|\le c_1\sqrt{\tfrac{\log n}{n}}.
\ee
On this event, for $n$ sufficiently large enough such that $c_1\sqrt{\frac{\log n}{n}}\le \tfrac12$, we have $\min_i r_i\ge \tfrac12$. Hence
\be
\big|(D^\mathrm{L}_{n})_{ii}-1\big|=\left|\frac{1}{r_i}-1\right|=\frac{|1-r_i|}{r_i}\le 2|r_i-1|.
\ee
Define $\Delta^\mathrm{L}_{n}\coloneq D^\mathrm{L}_{n}-I$ which are all diagonal. By rewriting the constant factors ($2c_1=C$), we get
\be
\|\Delta^\mathrm{L}_{n}\|_{\op}=\max_i\big|(D^\mathrm{L}_{n})_{ii}-1\big|\le C\sqrt{\tfrac{\log n}{n}}.
\ee
\end{proof}
\begin{lemma}\label{lem:cov-op}
If $D=I+\Delta$ is diagonal and $D=D^\dagger$ (e.g.\ real diagonal), then for any $A$,
\begin{equation}
\|DAA^\dagger D^\dagger-AA^\dagger\|_{\op}\ \le\ (2\|\Delta\|_{\op}+\|\Delta\|_{\op}^2)\,\|A\|_{\op}^2.
\end{equation}
\end{lemma}
\begin{proof}
Let $M\coloneq AA^\dagger$. Then
\be
DMD-M=(I+\Delta)M(I+\Delta)-M=\Delta M + M\Delta + \Delta M\Delta.
\ee
Hence by the triangle inequality and submultiplicativity of $\|\cdot\|_{\op}$,
\begin{align}
\|DMD-M\|_{\op}
&\le \|\Delta M\|_{\op} + \|M\Delta\|_{\op} + \|\Delta M\Delta\|_{\op} \\
&\le \|\Delta\|_{\op}\|M\|_{\op} + \|M\|_{\op}\|\Delta\|_{\op} + \|\Delta\|_{\op}\|M\|_{\op}\|\Delta\|_{\op} \\
&= (2\|\Delta\|_{\op}+\|\Delta\|_{\op}^2)\,\|M\|_{\op}.
\end{align}
Finally, $\|M\|_{\op}=\|AA^\dagger\|_{\op}\le \|A\|_{\op}\|A^\dagger\|_{\op}=\|A\|_{\op}^2$, which proves the claim.
\end{proof}
\begin{remark}
We now have all the ingredients required to prove \cref{thm:lui-rui-support}, which is done in Appendix~\ref{proof-of-thm1}.
\end{remark}
\begin{proposition}\label{prop:gauge-fully}
For any square matrix $A = (a_{ij})$ with non-zero entries, we obtain $A^\mathrm{B}$ and set the row, column, and grand log-averages \begin{equation}\label{eq:row-col-gra-mean-rep}
m_i\coloneq \frac1n\sum_{j=1}^n\log|a_{ij}|,
\qquad
n_j\coloneq \frac1n\sum_{i=1}^n\log|a_{ij}|,
\qquad
\hat{m}\coloneq \frac1{n^2}\sum_{i,j=1}^n\log|a_{ij}|.
\end{equation}
If we write $A^\mathrm{B}_{ij}=d^\ell_i a_{ij} {d^r_j}$ (entrywise), then the unknowns $(\log d^\ell_i)_{i=1}^n$ and $(\log {d^r_j})_{j=1}^n$ with $d^\ell_i,{d^r_j}\in\mathbb R_{+}$ (equivalently, $\log d^\ell_i,\log {d^r_j}\in\mathbb R$) that satisfy \cref{eq:UI-balance} are precisely
\begin{equation}
\log d^\ell_i=-m_i-\alpha,\qquad \log {d^r_j}=-n_j-\beta,\qquad \alpha+\beta=-\hat{m} 
\end{equation}
for some real constants $\alpha,\beta$. The pair $(\alpha,\beta)$ is a one-parameter \emph{scale family}. The \emph{symmetric scale} is $\alpha=\beta=-\hat{m}/2$, equivalently
\begin{equation}
\frac1n\sum_{i=1}^n \log d^\ell_i=\frac1n\sum_{j=1}^n \log {d^r_j}.
\end{equation}
\end{proposition}
\begin{proof}
Write $A^\mathrm{B}_{ij}=d^\ell_i a_{ij} {d^r_j}$ (entry-wise), i.e., $D^\mathrm{B_L}=\diag(d^\ell_i),D^\mathrm{B_R}=\diag({d^r_j})$. Then
\begin{equation}
\log|A^\mathrm{B}_{ij}|=\log d^\ell_i+\log|a_{ij}|+\log {d^r_j}.
\end{equation}
Fix $i$ and average over $j$:
\begin{equation}
\frac1n\sum_{j=1}^n \log|A^\mathrm{B}_{ij}|
= \log d^\ell_i+\frac1n\sum_{j=1}^n\log|a_{ij}|+\frac1n\sum_{j=1}^n\log {d^r_j}
= \log d^\ell_i+m_i+\alpha,
\end{equation}
where we set $\alpha\coloneq \frac1n\sum_{j=1}^n \log {d^r_j}$. The row constraint \cref{eq:UI-balance} gives
\begin{equation}\label{eq:row-eq}
\log d^\ell_i=-m_i-\alpha\qquad(1\le i\le n).
\end{equation}
Fix $j$ and average over $i$:
\begin{equation}
\frac1n\sum_{i=1}^n \log|A^\mathrm{B}_{ij}|
= \frac1n\sum_{i=1}^n\log d^\ell_i+\frac1n\sum_{i=1}^n\log|a_{ij}|+\log {d^r_j}
= \beta+n_j+\log {d^r_j},
\end{equation}
where $\beta\coloneq \frac1n\sum_{i=1}^n \log d^\ell_i$. The column constraint \cref{eq:UI-balance} gives
\begin{equation}\label{eq:col-eq}
\log {d^r_j}=-n_j-\beta\qquad(1\le j\le n).
\end{equation}
Average \cref{eq:row-eq} over $i$ or \cref{eq:col-eq} over $j$ to get the single constraint
\begin{equation}
\alpha+\beta=-\hat{m}.
\end{equation}
This is exactly the displayed solution family. The \emph{symmetric scale} chooses $\alpha=\beta=-\hat{m}/2$, which forces
\begin{equation}
\frac1n\sum_i \log d^\ell_i=\beta=-\hat{m}/2=\alpha=\frac1n\sum_j \log {d^r_j}.\qedhere
\end{equation}
\end{proof}
\begin{remark}
For continuous entry laws (e.g.\ Gaussian Ginibre/Wigner), $\Prob(a_{ij}=0)=0$, so the logarithms are well-defined almost surely and our BUISVD scalings are a.s.\ well-posed. 
\end{remark}
\begin{lemma}\label{lem:log-mean-conc}
Let $a_{ij}=X_{ij}/\sqrt n$ where $X_{ij}$ are i.i.d.\ real or complex random variables with \emph{mean} $0$, \emph{variance} $1$, \emph{sub-Gaussian} tails, $\Prob(X_{ij}=0)=0$, and a small-ball bound near $0$, namely: there exist $\alpha>0$ and $c_2<\infty$ such that $\Prob(|X_{ij}|\le t)\le c_2\,t^{\alpha}$ for all $t\in(0,1)$. Define the row, column, and grand log-averages $m_i,n_j,\hat{m}$ (as in \cref{eq:row-col-gra-mean-rep}). Set $c_\star\coloneq \E\log|X_{11}|$ (finite under the assumptions above). Then there exists $C>0$ such that, for all sufficiently large $n$, with probability at least $1-Cn^{-10}$,
\begin{equation}
\max_i\Big|m_i-\big(c_\star-\tfrac12\log n\big)\Big|,\quad
\max_j\Big|n_j-\big(c_\star-\tfrac12\log n\big)\Big|,\quad
\Big|\hat{m}-\big(c_\star-\tfrac12\log n\big)\Big|
\ \le\ C\sqrt{\tfrac{\log n}{n}}.
\end{equation}
\end{lemma}
\begin{proof}
We do rows; columns are identical; $\hat{m}$ is a single average over $n^2$ terms. First, under the stated assumptions $c_\star=\E\log|X_{11}|$ is finite. Also, by increasing $c_2$ if necessary, we may assume the small-ball bound $\Prob(|X_{11}|\le t)\le c_2 t^\alpha$ holds for all $t>0$ (it is trivial for $t\ge1$ since $\Prob(|X_{11}|\le t)\le 1$).  Fix $i$. Define
\begin{equation}\label{def:of-y}
m_i-\Big(c_\star-\tfrac12\log n\Big)=\frac1n\sum_{j=1}^n Y_{ij},\qquad
Y_{ij}\coloneq \log|a_{ij}|-\E\log|a_{ij}|=\log|X_{ij}|-\E\log|X_{ij}|.
\end{equation}
Thus $\{Y_{ij}\}_{j=1}^n$ are i.i.d.\ mean-zero. We now claim $Y\coloneq \log|X|-\E\log|X|$ is \emph{sub-exponential}, i.e.\ $\|Y\|_{\psi_1}\le c_3$ for some $c_3<\infty$ depending only on the entry law.

To prove the claim note that sub-Gaussianity of $X$ implies the tail bound $\Prob(|X|\geq u)\leq 2 e^{-c_4u^2}$, and by definition \cref{def:of-y}
\be
Y\geq t \iff \log|X|-c_\star \geq t \iff \log|X|\geq t+c_\star \iff |X| \geq e^{t+c_\star}, 
\ee
and therefore $\Prob(Y\geq t)=\Prob(|X|\geq e^{t+c_\star})$. So now applying the sub-Gaussian tail bound with $u=e^{t+c_\star}$, we have for the \emph{right tail} and $t\ge0$,
\begin{equation}
\Prob(Y\ge t)=\Prob\big(|X|\ge e^{t+c_\star}\big)\ \le\ 2\exp(-c_4(e^{t+c_\star})^2) = 2\exp\!\big(-c_5\,e^{2t}\big),
\end{equation}
which is super-exponential in $t$, which is stronger than exponential decay, and thus automatically supports the sub-exponential behaviour. For the \emph{left tail}, we can obtain from the definition \cref{def:of-y} again that $\Prob(Y\leq -t)=\Prob(|X|\leq e^{c_\star-t})$, and the small-ball bound gives, for $t\ge0$,
\begin{equation}
\Prob(Y\le -t)=\Prob\big(|X|\le e^{-t+c_\star}\big)\ \le\ c_2\,e^{\alpha c_\star}\,e^{-\alpha t} \ =:\ c_6e^{-\alpha t}.
\end{equation}
Hence $\Prob(|Y|\ge t)\le A e^{-B t}$ for some $A,B>0$, which implies $Y$ is sub-exponential (see, e.g., \cite{Vershynin2018HDP}). This proves the claim. Applying \cref{thm:bernstein-subexp}, for all $u>0$,
\begin{equation}
\Prob\!\left(\left|\frac1n\sum_{j=1}^n Y_{ij}\right|\geq u\right)
\ \le\ 2\exp\!\left(-c n\,\min\!\left(\frac{u^2}{c_3^2},\frac{u}{c_3}\right)\right).
\end{equation}
Choose $u=c_7\sqrt{\frac{\log n}{n}}$ with $c_7$ large so that the RHS is $\le n^{-12}$. A union bound over $i=1,\dots,n$ as in steps from \cref{eq:b.22ref} to \cref{eq:b.28ref} yields
\begin{equation}
\Prob\!\left(\max_i\left|m_i-\Big(c_\star-\tfrac12\log n\Big)\right|\geq c_7\sqrt{\tfrac{\log n}{n}}\right)\le n^{-11}.
\end{equation}
The column bound is identical. For the grand mean,
\begin{equation}
\hat{m}-\Big(c_\star-\tfrac12\log n\Big)=\frac1{n^2}\sum_{i,j=1}^n Y_{ij},
\end{equation}
and the same \cref{thm:bernstein-subexp} bound with $n^2$ samples gives a stronger tail; we keep the displayed order for uniformity. Adjusting constants and $n$ large ensures the stated $1-Cn^{-10}$ bound.
\end{proof}
\begin{remark}
If $X_{ij}$ are standard real Gaussian $\,\mathcal N(0,1)$ then $c_\star=\mathbb E\log|X_{11}|=c_{\beta=1}$. If $X_{ij}$ are standard circular complex Gaussian $\,\mathcal{CN}(0,1)$ (with $\mathbb E|X_{11}|^2=1$) then $c_\star=\mathbb E\log|X_{11}|=c_{\beta=2}$, as computed in Lemma~\ref{lem:log-constants}.
\end{remark}
\begin{remark}
We now have all the results required to prove \cref{thm:uisvd-support}, as done in Appendix~\ref{proof-of-thm2}. Thereafter we need one Lemma each to prove \cref{thm:lui-rui-wigner,thm:ui-wigner}, which are given below.
\end{remark}
\begin{lemma}\label{lem:wigner-rowcol-norms}
Let $A_n = (a_{ij})$ be Wigner-$\beta$ (GOE/GUE) with the $1/\sqrt n$ scaling and set
$r_i\coloneq \|A_n(i,:)\|_2$.
Then there exists $C>0$ such that, with probability at least $1-C n^{-10}$,
\begin{equation}\label{eq:wigner-rowcol-norms}
\max_i \bigl|r_i-1\bigr| \;\le\; C\sqrt{\tfrac{\log n}{n}}.
\end{equation}
\end{lemma}
\begin{proof}
We prove the row bound; the column bound is identical by symmetry. Fix $i\in\{1,\dots,n\}$ and write
\begin{equation}
r_i^2=\sum_{j=1}^n |a_{ij}|^2=\sum_{j\neq i}|a_{ij}|^2+|a_{ii}|^2.
\end{equation}
By the GOE/GUE constructions in Table~\ref{tab:wigner-ensembles}, for fixed $i$ the variables $\{a_{ij}:j\neq i\}$ are independent and mean zero.
Moreover, the entries of $G$ are standard Gaussians in the sense that $\E G_{ij}^2=\Var(G_{ij})=1$ in the GOE case, while $\E|G_{ij}|^2=1$ in the GUE case. Also for the GUE case, i.e., $G_{ij}\sim \mathcal{CN}(0,1)$, the real part $\operatorname{Re}(G_{ij})\sim \mathcal N(0,\frac{1}{2})$, and thus $\E(\operatorname{Re}G_{ij})^2=\frac{1}{2}$.

From \cref{tab:wigner-ensembles} we see that if $\beta=1$ (GOE) then $a_{ij}=(G_{ij}+G_{ji})/\sqrt{2n}$ for $j\neq i$ and $a_{ii}=2G_{ii}/\sqrt{2n}$, so we get
\be
\E|a_{ij}|^2=(\E G_{ij}^2+\E G_{ji}^2)/(2n)=1/n, \qquad \E|a_{ii}|^2=4\,\E G_{ii}^2/(2n)=2/n.
\ee
If $\beta=2$ (GUE) then $a_{ij}=(G_{ij}+\overline{G_{ji}})/\sqrt{2n}$ for $j\neq i$ and $a_{ii}=2\operatorname{Re}(G_{ii})/\sqrt{2n}$, hence
\be
\E|a_{ij}|^2=(\E|G_{ij}|^2+\E|G_{ji}|^2)/(2n)=1/n, \qquad
\E|a_{ii}|^2=4\,\E(\operatorname{Re}G_{ii})^2/(2n)=1/n.
\ee
Thus $\E|a_{ii}|^2=\nu_\beta/n$ with $\nu_1=2$ and $\nu_2=1$. Therefore
\begin{equation}\label{eq:vbeta-values}
\E r_i^2=\frac{n-1}{n}+\frac{\nu_\beta}{n}=
\begin{cases}
1+\frac{1}{n}, & \beta=1\text{ (GOE)},\\[2pt]
1, & \beta=2\text{ (GUE)}.
\end{cases}
\end{equation}
Define centered variables $Y_{ij}\coloneq |a_{ij}|^2-\E|a_{ij}|^2$ for $j\neq i$ and $X_{ii}\coloneq |a_{ii}|^2-\E|a_{ii}|^2$. Since $a_{ij}$ is sub-Gaussian with $\|a_{ij}\|_{\psi_2}=O(n^{-1/2})$, as before from claim within proof of Lemma~\ref{lem:log-mean-conc} (see also \cite{Vershynin2018HDP}) we get $\|Y_{ij}\|_{\psi_1}=O(n^{-1})$ uniformly in $j\neq i$ (and in $\beta$), and similarly $\|X_{ii}\|_{\psi_1}=O(n^{-1})$, i.e., sub-exponential. Therefore, applying \cref{thm:bernstein-subexp} to the $n$ independent mean-zero variables $\{Y_{ij}:j\neq i\}\cup\{X_{ii}\}$ yields, for $u\in(0,1)$,
\begin{equation}\label{eq:bernstein-sum}
\Prob\Big(\bigl|r_i^2-\E r_i^2\bigr|\geq u\Big)\ \le\ 2\exp\!\big(-c n u^2\big),
\end{equation}
for some constant $c>0$. Choose $u\coloneq c\sqrt{(\log n)/n}$ with $c$ large enough so that the RHS is $\le n^{-11}$. On the event $\{|r_i^2-\E r_i^2|\le u\}$ we get
\begin{equation}
\bigl|r_i-\sqrt{\E r_i^2}\bigr|\ =\ \frac{|r_i^2-\E r_i^2|}{r_i+\sqrt{\E r_i^2}}\ 
\ \le\ \,|r_i^2-\E r_i^2|\ \le\ u.
\end{equation}
Hence
\begin{equation}
|r_i-1| \ \le\ \bigl|r_i-\sqrt{\E r_i^2}\bigr|+\bigl|\sqrt{\E r_i^2}-1\bigr|
\ \le\ u + \bigl|\E r_i^2-1\bigr|
\ \le\ c\sqrt{\tfrac{\log n}{n}}+\frac{1}{n}.
\end{equation}
We absorb the negligible $1/n$ term into the $\sqrt{(\log n)/n}$ term by increasing $c$ to $\hat c$ if needed, to obtain
\begin{equation}
\Prob\!\left(|r_i-1|\geq \hat c\sqrt{\tfrac{\log n}{n}}\right)\ \le\ n^{-11}.
\end{equation}
A union bound over $i=1,\dots,n$ yields
\begin{equation}
\Prob\!\left(\max_i |r_i-1|\geq C\sqrt{\tfrac{\log n}{n}}\right)\ \le\ C n^{-10}\ 
\end{equation}
for an adjusted constant $C>0$, proving \cref{eq:wigner-rowcol-norms}. The column bound follows identically by applying the same argument.
\end{proof}
\begin{remark}\label{rem:wigner-entry-laws}
Under the GOE/GUE constructions in \cref{tab:wigner-ensembles} with the $1/\sqrt n$ scaling, for each fixed $i$ the collection $\{a_{ij}\}_{j\ne i}\cup\{a_{ii}\}$ is independent, with
\be
a_{ij}\stackrel{d}{=} \frac{g_\beta}{\sqrt n}\quad (j\ne i),
\qquad
a_{ii}\stackrel{d}{=} \frac{\sqrt{\nu_\beta}\,g_1}{\sqrt n},
\qquad
\nu_1=2,\ \nu_2=1.
\ee
\end{remark}
\begin{lemma}\label{lem:log-mean-wigner}
Let $A_n = (a_{ij})$ be Wigner-$\beta$ with the $1/\sqrt n$ scaling, and define the row, column, and grand log-averages $m_i,n_j,\hat{m}$ as in \cref{eq:row-col-gra-mean-rep}. Let $c_\beta$ be as in Lemma~\ref{lem:log-constants}. Then $\max_i\big|m_i-(c_\beta-\tfrac12\log n)\big|=O_{\mathrm{a.s.}}\!\big(\sqrt{(\log n)/n}\big)$, and similarly for $n_j$ and $\hat{m}$.
\end{lemma}
\begin{proof}
We treat rows; columns are identical; $\hat m$ is similar. Set $\mu\coloneq c_\beta-\tfrac12\log n$ and fix $i$. By \cref{rem:wigner-entry-laws}, $\{\log|a_{ij}|\}_{j\ne i}$ are independent with $\E\log|a_{ij}|=\mu$, and we define
\be
Y_{ij}\coloneq \log|a_{ij}|-\mu,\qquad j\ne i.
\ee
By the claim within proof of Lemma~\ref{lem:log-mean-conc}, the variables $Y_{ij}$ are centered sub-exponential with parameters depending only on $\beta$. Hence by \cref{thm:bernstein-subexp}, for all $u\in(0,1)$,
\begin{equation}\label{eq:bernstein-row-log}
\Prob\!\left(\left|\frac1n\sum_{j\ne i}Y_{ij}\right|\geq u\right)\le 2\exp(-c n u^2),
\end{equation}
for some $c>0$ depending only on $\beta$. Write
\begin{equation}\label{eq:b.70}
m_i-\mu=\frac1n\sum_{j\ne i}Y_{ij}+\Delta_i,
\qquad
\Delta_i\coloneq \frac1n\big(\log|a_{ii}|-\mu\big).
\end{equation}
By \cref{rem:wigner-entry-laws}, $a_{ii}\stackrel{d}{=}\sqrt{\nu_\beta}\,g_1/\sqrt n$, so $\log|a_{ii}|-\E\log|a_{ii}|$ is sub-exponential by the same argument as in Lemma~\ref{lem:log-mean-conc} (applied to $g_1\sim\mathcal N(0,1)$). Therefore, using the sub-exponential tail bound, for all $u\in(0,1)$,
\begin{equation}\label{eq:diagonal-small}
\Prob\big(|\Delta_i-\E\Delta_i|\geq u\big)
=
\Prob\Big(\big|\log|a_{ii}|-\E\log|a_{ii}|\big|\geq nu\Big)
\le 2\exp(-c\,\min(n^2u^2,n u)),
\end{equation}
for a (possibly different) $c>0$ depending only on $\beta$. Also, 
\be
|\E\Delta_i|=\frac1n|\E\log|a_{ii}|-\mu|=O(1/n).
\ee
Now take $u_n\coloneq \hat{c}\sqrt{\frac{\log n}{n}}$. So for sufficiently large $n$ we have that $|\E\Delta_i|\le u_n$. Writing $\Delta_i=(\Delta_i-\E\Delta_i)+\E\Delta_i$, we get by \cref{eq:b.70} and the triangle inequality,
\be
|m_i-\mu|
\le \left|\frac1n\sum_{j\ne i}Y_{ij}\right|+|\Delta_i-\E\Delta_i|+|\E\Delta_i|,
\ee
where the events $\{|m_i-\mu|\geq 3u_n\}\subseteq\left\{\left|\frac1n\sum_{j\ne i}Y_{ij}\right|\geq u_n\right\}\cup\{|\Delta_i-\E\Delta_i|\geq u_n\}$. Thus,
\be
\Prob\big(|m_i-\mu|\geq 3u_n\big)
\le
\Prob\!\left(\left|\frac1n\sum_{j\ne i}Y_{ij}\right|\geq u_n\right)
+\Prob\big(|\Delta_i-\E\Delta_i|\geq u_n\big).
\ee
Using \cref{eq:bernstein-row-log} and \cref{eq:diagonal-small} (with $u=u_n$), where we note that $nu_n\to\infty$ and $\min(n^2u_n^2,nu_n)=nu_n\ge n u_n^2$ for all sufficiently large $n$, we obtain
\be
\Prob\big(|m_i-\mu|\geq 3u_n\big)\le C\exp(-c n u_n^2),
\ee
for constants $C,c>0$ depending only on $\beta$. A union bound gives
\be
\Prob\Big(\max_{1\le i\le n}|m_i-\mu|\geq 3u_n\Big)
\le C\,n\,e^{-c n u_n^2}
= C\,n^{\,1-c{\hat{c}}^2},
\ee
which is summable for $\hat{c}$ large enough; Borel-Cantelli yields $\max_i|m_i-\mu|=O_{a.s.}(u_n)$. The proof for $n_j$ is entirely analogous. For $\hat m=\frac1{n^2}\sum_{i,j}\log|a_{ij}|$, symmetry $|a_{ji}|=|a_{ij}|$ implies
\be
\hat m-\mu
=\frac{2}{n^2}\sum_{i<j}\big(\log|a_{ij}|-\mu\big)
+\frac{1}{n^2}\sum_{i=1}^n\big(\log|a_{ii}|-\mu\big).
\ee
It follows similarly that $|\hat m-\mu|=O_{a.s.}(u_n)$.
\end{proof}
\subsection{Proofs of Section~\ref{sec:RMT}}\label{sec:proofs-RMT}
\begin{proof}[\textbf{Proof of \cref{thm:lui-rui-support}}]\label{proof-of-thm1}
We do the LUI case, RUI is analogous. Let $r_i\coloneq \|A_n(i,:)\|_2$. Lemma~\ref{lem:row-col-conc} yields with high-probability
\begin{equation}
\max_i|r_i-1|\ \le\ c_1\sqrt{\tfrac{\log n}{n}}.
\end{equation}
Using Corollary~\ref{cor:DL-DR-expansion},
\begin{equation}\label{eq:diag-expand}
D^\mathrm{L}_{n}=I+\Delta^\mathrm{L}_{n},\qquad
\|\Delta^\mathrm{L}_{n}\|_{\op}\ \le\ c_1\sqrt{\tfrac{\log n}{n}}.
\end{equation}
So we get
\begin{equation}
A^\mathrm{L}_n(A^\mathrm{L}_n)^\dagger
= D^\mathrm{L}_{n}A_nA_n^\dagger D^\mathrm{L}_{n}
= A_nA_n^\dagger + E^\mathrm{L}_n,
\end{equation}
with $E^\mathrm{L}_n\coloneq \Delta^\mathrm{L}_{n}A_nA_n^\dagger + A_nA_n^\dagger\Delta^\mathrm{L}_{n} + \Delta^\mathrm{L}_{n}A_nA_n^\dagger\Delta^\mathrm{L}_{n}$. Lemma~\ref{lem:cov-op} gives
\begin{equation}\label{eq:E-matrix-bound}
\|E^\mathrm{L}_n\|_{\op}\ \le\ (2\|\Delta^\mathrm{L}_{n}\|_{\op}+\|\Delta^\mathrm{L}_{n}\|_{\op}^2)\,\|A_n\|_{\op}^2
\ \le\ c_2\sqrt{\tfrac{\log n}{n}}\cdot \|A_n\|_{\op}^2.
\end{equation}
By \cref{thm:MP-edge} we have
\begin{equation}\label{eq:A-matrix-opnorm}
\lambda_{\max}\!\Big(\frac1n X_nX_n^\dagger\Big)\to 4\quad\Longrightarrow\quad \|A_n\|_{\op}=\sqrt{\lambda_{\max}(A_nA_n^\dagger)}\to 2
\end{equation}
almost surely. Define the high-probability event
\begin{equation}
\mathcal E_n\coloneq \Big\{\max_i|r_i-1|\le c_1\sqrt{\tfrac{\log n}{n}}\Big\}.
\end{equation}
By Lemma~\ref{lem:row-col-conc}, for the complement event $\Prob(\mathcal E_n^c)\le C\,n^{-10}$ and hence $\sum_{n=1}^\infty \Prob(\mathcal E_n^c)<\infty$; by the Borel-Cantelli lemma,
$\mathcal E_n$ holds eventually almost surely. So from \cref{eq:diag-expand,eq:E-matrix-bound,eq:A-matrix-opnorm} we get 
\begin{equation}
\|E^\mathrm{L}_n\|_{\op}\ \longrightarrow\ 0\qquad\text{almost surely}.
\end{equation}
Now set
\begin{equation}
M_n\coloneq A_nA_n^\dagger,\qquad N_n\coloneq A^\mathrm{L}_n(A^\mathrm{L}_n)^\dagger,\qquad E_n\coloneq N_n-M_n=E^\mathrm{L}_n.
\end{equation}
For any fixed $z=x+\mathrm{i}y$ with $y>0$, Lemma~\ref{lem:esd-continuity} gives
\begin{equation}
|m_{N_n}(z)-m_{M_n}(z)|\ \le\ \frac{\|E_n\|_{\op}}{y^2}\ \xrightarrow[n\to\infty]{a.s.}\ 0.
\end{equation}
By \cref{thm:MP-law}, $m_{M_n}(z)\to m_{\mathrm{MP}(1,1)}(z)$ almost surely for every such $z$; hence we also have
$m_{N_n}(z)\to m_{\mathrm{MP}(1,1)}(z)$ almost surely. Since $\mathrm{MP}(1,1)$ has density $\frac{1}{2\pi x}\sqrt{x(4-x)}\,\mathbf 1_{(0,4)}(x)$ on $[0,4]$, its pushforward under $x\mapsto \sqrt{x}$ has density
\be
f(s)=2s\cdot \frac{1}{2\pi s^2}\sqrt{s^2(4-s^2)}\,\mathbf 1_{(0,2)}(s)=\frac{1}{\pi}\sqrt{4-s^2}\,\mathbf 1_{[0,2]}(s),
\ee
i.e.\ the quarter-circle law on $[0,2]$.  By Lemma~\ref{lem:esd-continuity} the empirical spectral distribution (ESD) of $N_n$ converges almost surely to the MP$(1,1)$ law supported on $[0,4]$. Finally, by Lemma~\ref{thm:pushforward} the empirical measure of singular values of $A^\mathrm{L}_n$ converges almost surely to the quarter-circle law on $[0,2]$ with density $f(s)=\frac{1}{\pi}\sqrt{4-s^2}\,\mathbf{1}_{[0,2]}(s)$.

The RUI case is identical where we replace $D^\mathrm{L}_{n}$ by $D^\mathrm{R}_{n}$ and use $\|A_n(:,j)\|_2$ instead of $r_i=\|A_n(i,:)\|_2$.
\end{proof}
\begin{proof}[\textbf{Proof of \cref{thm:uisvd-support}}]\label{proof-of-thm2}
Write $a_{ij}\coloneq X_{ij}/\sqrt n$ (so $A_n=(a_{ij})$) and define the row, column, and grand log-averages $m_i,n_j,\hat{m}$ as in \cref{eq:row-col-gra-mean-rep}. By Proposition~\ref{prop:gauge-fully} (symmetric scale),
\begin{equation}\label{eq:gauge-symmetric-log}
\log(D^{\mathrm{B_L}}_{n})_{ii}=-m_i+\tfrac12\hat{m},\qquad
\log(D^{\mathrm{B_R}}_{n})_{jj}=-n_j+\tfrac12\hat{m}.
\end{equation}
Since $\Prob(X_{11}=0)=0$, we have $\Prob(a_{ij}\neq 0\ \forall i,j)=1$ for each fixed $n$, so Proposition~\ref{prop:gauge-fully} applies almost surely. Moreover, changing $(\alpha,\beta)$ within the one-parameter family in Proposition~\ref{prop:gauge-fully} rescales $D^{\mathrm{B_L}}$ and $D^{\mathrm{B_R}}$ inversely and leaves
$A_n^{\mathrm B}=D^{\mathrm{B_L}}_nA_nD^{\mathrm{B_R}}_n$ unchanged; the symmetric scale fixes a unique choice of the pair $(D^{\mathrm{B_L}}_n,D^{\mathrm{B_R}}_n)$. Introduce deviations
\begin{equation}\label{eq:83ref}
\delta_i\coloneq m_i-\Big(c_\star-\tfrac12\log n\Big),\quad
\epsilon_j\coloneq n_j-\Big(c_\star-\tfrac12\log n\Big),\quad
\bar\delta\coloneq \hat{m}-\Big(c_\star-\tfrac12\log n\Big),
\end{equation}
so that by Lemma~\ref{lem:log-mean-conc}, $\max_i|\delta_i|,\ \max_j|\epsilon_j|,\ |\bar\delta|\ \le C\sqrt{\tfrac{\log n}{n}}$ on the same high-probability event. We get
\begin{equation}\label{eq:repeat-ui-stretch}
\log(D^{\mathrm{B_L}}_{n})_{ii}
= -\Big(c_\star-\tfrac12\log n+\delta_i\Big)+\tfrac12\Big(c_\star-\tfrac12\log n+\bar\delta\Big)
= -\tfrac{c_\star}{2}+\tfrac14\log n + \Big(-\delta_i+\tfrac12\bar\delta\Big).
\end{equation}
Exponentiating and factoring out $e^{-c_\star/2}n^{1/4}$ gives
\begin{equation}\label{eq:DL-exp}
(D^{\mathrm{B_L}}_{n})_{ii}=e^{-c_\star/2}\,n^{1/4}\,\exp\!\Big(-\delta_i+\tfrac12\bar\delta\Big)
= e^{-c_\star/2}\,n^{1/4}\,\big(1+\delta^{\mathrm{B_L}}_i\big),
\end{equation}
where we define
\begin{equation}
\delta^{\mathrm{B_L}}_i\coloneq \exp\!\Big(-\delta_i+\tfrac12\bar\delta\Big)-1.
\end{equation}
Since $|-\delta_i+\tfrac12\bar\delta|\le C_1\sqrt{\tfrac{\log n}{n}}$, the elementary bound
\begin{equation}
|e^x-1|\le e^{|x|}-1\le |x|e^{|x|}\qquad (x\in\mathbb{R})
\end{equation}
implies $|\delta^{\mathrm{B_L}}_i|\le C_1\sqrt{\tfrac{\log n}{n}}$ for all $i$ (uniformly on the high-probability event). Define $\Delta^{\mathrm{B_L}}_{n}\coloneq \diag(\delta^{\mathrm{B_L}}_{1},\dots,\delta^{\mathrm{B_L}}_{n})$ so that
\begin{equation}\label{eq:DU-L-bnd}
D^{\mathrm{B_L}}_{n}=e^{-c_\star/2}\,n^{1/4}\,(I+\Delta^{\mathrm{B_L}}_{n}),
\qquad \|\Delta^{\mathrm{B_L}}_{n}\|_{\op}=\max_i|\delta^{\mathrm{B_L}}_i|\ \le\ C_1\sqrt{\tfrac{\log n}{n}}.
\end{equation}
An identical computation (with $n_j,\epsilon_j$) yields
\begin{equation}\label{eq:DU-R-bnd}
D^{\mathrm{B_R}}_{n}=e^{-c_\star/2}\,n^{1/4}\,(I+\Delta^{\mathrm{B_R}}_{n}),
\qquad \|\Delta^{\mathrm{B_R}}_{n}\|_{\op}\ \le\ C_2\sqrt{\tfrac{\log n}{n}}.
\end{equation}
Multiplying out with $A_n$ gives the core factorisation. Set, for brevity,
\begin{equation}
\hat \delta_n\coloneq \hat C\sqrt{\tfrac{\log n}{n}}\qquad\text{so that}\qquad
\|I+\Delta^{\mathrm{B_L}}_{n}\|_{\op},\ \|I+\Delta^{\mathrm{B_R}}_{n}\|_{\op}\ \le\ 1+\hat \delta_n.
\end{equation}
Then for the upper bound
\begin{equation}
\|A^{\mathrm{B}}_n\|_{\op}
\le e^{-c_\star}\sqrt n\,\|I+\Delta^{\mathrm{B_L}}_{n}\|_{\op}\,\|A_n\|_{\op}\,\|I+\Delta^{\mathrm{B_R}}_{n}\|_{\op}
\le e^{-c_\star}\sqrt n\,(1+\hat \delta_n)^2\,\|A_n\|_{\op}.
\end{equation}
By \cref{thm:MP-edge}, $\|A_n\|_{\op}\to 2$ almost surely; enlarging the constant in $\hat \delta_n$ if needed yields
\begin{equation}
\|A^{\mathrm{B}}_n\|_{\op}\ \le\ 2e^{-c_\star}\sqrt n\,\big(1+o(1)\big)
\end{equation}
for all large $n$ on the same event. For the lower bound, using
\begin{equation}
(I+\Delta^{\mathrm{B_L}}_{n})A_n(I+\Delta^{\mathrm{B_R}}_{n})=A_n+\Delta^{\mathrm{B_L}}_{n}A_n+A_n\Delta^{\mathrm{B_R}}_{n}+\Delta^{\mathrm{B_L}}_{n}A_n\Delta^{\mathrm{B_R}}_{n},
\end{equation}
we obtain using reverse triangle inequality and sub-multiplicativity
\begin{equation}
\begin{aligned}
\|(I+\Delta^{\mathrm{B_L}}_{n})A_n(I+\Delta^{\mathrm{B_R}}_{n})\|_{\op}
&\ \ge\ \|A_n\|_{\op}-\|A_n\|_{\op}\big(\|\Delta^{\mathrm{B_L}}_{n}\|_{\op}+\|\Delta^{\mathrm{B_R}}_{n}\|_{\op}\big)\\
&\quad-\|A_n\|_{\op}\|\Delta^{\mathrm{B_L}}_{n}\|_{\op}\|\Delta^{\mathrm{B_R}}_{n}\|_{\op}.
\end{aligned}
\end{equation}
On the high-probability event,
\begin{equation}
\|A^{\mathrm{B}}_n\|_{\op}
\ \ge\ e^{-c_\star}\sqrt n\,\|A_n\|_{\op}\,\Big(1-2\hat \delta_n-\hat \delta_n^2\Big)
\ =\ 2e^{-c_\star}\sqrt n\,\big(1-o(1)\big),
\end{equation}
since $\|A_n\|_{\op}\to 2$ and $\hat \delta_n\to 0$. Together with the upper bound,
\begin{equation}
\frac{\|A^{\mathrm{B}}_n\|_{\op}}{\sqrt n}\ \longrightarrow\ 2e^{-c_\star}\qquad\text{almost surely}.
\end{equation}
For the normalised matrices $T_n\coloneq n^{-1/2}A^{\mathrm{B}}_n$, rewrite as
\begin{equation}\label{eq:97ref}
T_n=e^{-c_\star}\,(I+\Delta^{\mathrm{B_L}}_{n})\,A_n\,(I+\Delta^{\mathrm{B_R}}_{n}).
\end{equation}
Set
\begin{equation}
H_n\coloneq T_nT_n^\dagger,\qquad \widehat H_n\coloneq e^{-2c_\star}A_nA_n^\dagger.
\end{equation}
Define
\be
P_n\coloneq A_n(I+\Delta^{\mathrm{B_R}}_{n}),\qquad 
\widetilde H_n\coloneq e^{-2c_\star}P_nP_n^\dagger,
\ee
so that
\be
H_n=e^{-2c_\star}(I+\Delta^{\mathrm{B_L}}_{n})\,P_nP_n^\dagger\,(I+\Delta^{\mathrm{B_L}}_{n}).
\ee
Applying Lemma~\ref{lem:cov-op} with $D=I+\Delta^{\mathrm{B_L}}_{n}$ and $A=P_n$ gives
\be
\|H_n-\widetilde H_n\|_{\op}\le e^{-2c_\star}\bigl(2\|\Delta^{\mathrm{B_L}}_{n}\|_{\op}+\|\Delta^{\mathrm{B_L}}_{n}\|_{\op}^2\bigr)\,\|P_n\|_{\op}^2.
\ee
Next,
\be
\widetilde H_n-\widehat H_n
=e^{-2c_\star}\Big(P_nP_n^\dagger-A_nA_n^\dagger\Big)
=e^{-2c_\star}A_n\Big((I+\Delta^{\mathrm{B_R}}_{n})(I+\Delta^{\mathrm{B_R}}_{n})^\dagger-I\Big)A_n^\dagger,
\ee
hence
\be
\|\widetilde H_n-\widehat H_n\|_{\op}
\le e^{-2c_\star}\bigl(\|\Delta^{\mathrm{B_R}}_{n}+(\Delta^{\mathrm{B_R}}_{n})^\dagger\|_{\op}+\|\Delta^{\mathrm{B_R}}_{n}\|_{\op}^2\bigr)\,\|A_n\|_{\op}^2.
\ee
Finally, $\|P_n\|_{\op}\le \|A_n\|_{\op}\,\|I+\Delta^{\mathrm{B_R}}_{n}\|_{\op}$, so combining the above yields
\begin{equation}
\begin{aligned}
\|H_n-\widehat H_n\|_{\op}
\;&\le\;
e^{-2c_\star}
\Big(2\|\Delta^{\mathrm{B_L}}_{n}\|_{\op}
 +\|\Delta^{\mathrm{B_L}}_{n}\|_{\op}^2\Big)
\,\|A_n\|_{\op}^2
\,\|I+\Delta^{\mathrm{B_R}}_{n}\|_{\op}^2
\\
&\quad
+e^{-2c_\star}
\Big(\|\Delta^{\mathrm{B_R}}_{n}
 +(\Delta^{\mathrm{B_R}}_{n})^\dagger\|_{\op}
 +\|\Delta^{\mathrm{B_R}}_{n}\|_{\op}^2\Big)
\,\|A_n\|_{\op}^2.
\end{aligned}
\end{equation}
In particular, absorbing quadratic terms into the constant $C$ since $\|\Delta\|_{\op}=o(1)$, we have
\be\label{eq:105ref}
\|H_n-\widehat H_n\|_{\op}
\le e^{-2c_\star}\,C\bigl(\|\Delta^{\mathrm{B_L}}_{n}\|_{\op}+\|\Delta^{\mathrm{B_R}}_{n}\|_{\op}\bigr)\,\|A_n\|_{\op}^2
\ee
for all large $n$ on the same event. On the high-probability event, \cref{eq:DU-L-bnd,eq:DU-R-bnd,eq:A-matrix-opnorm} give
\begin{equation}
\|H_n-\widehat H_n\|_{\op}\ =\ O\!\Big(\sqrt{\tfrac{\log n}{n}}\Big)\ \xrightarrow[n\to\infty]{}\ 0.
\end{equation}
By Lemma~\ref{lem:esd-continuity}, $\mathrm{ESD}(H_n)$ and $\mathrm{ESD}(\widehat H_n)$ have the same limiting law. By \cref{thm:MP-law}, $\mathrm{ESD}(A_nA_n^\dagger)\Rightarrow \mathrm{MP}(1,1)$; therefore $\mathrm{ESD}(\widehat H_n)\Rightarrow \mathrm{MP}(1,1)$ scaled by $e^{-2c_\star}$. Finally, by Lemma~\ref{thm:pushforward}, the empirical measures of the singular values of $T_n$ converge almost surely to the quarter-circle law on $[0,2e^{-c_\star}]$. The failure probabilities are $\le C n^{-10}$; hence $\sum_n \Prob(\text{fail at }n)<\infty$. By the Borel-Cantelli lemma, the high-probability event (and therefore $\|\Delta^{\mathrm{B_L}}_{n}\|_{\op},\|\Delta^{\mathrm{B_R}}_{n}\|_{\op}\to 0$) holds eventually almost surely, which promotes the convergence above to almost sure convergence.
\end{proof}
\begin{proof}[\textbf{Proof of Theorem~\ref{thm:lui-rui-wigner}}]\label{proof-of-lem1}
We show for LUI, with the RUI case being analogous, by taking column norms instead of row norms. Write $r_i\coloneq \|A_n(i,:)\|_2$, and further define $D_n^{\mathrm L}\coloneq \mathrm{Diag}(1/r_1,\dots,1/r_n)$ so that $A_n^{\mathrm L}=D_n^{\mathrm L}A_n$. Let $\delta_n^{\mathrm L}\coloneq \max_i|r_i-1|$. Then
\be
\|D_n^{\mathrm L}-I\|_{\op} = \max_i \bigg| \frac{1}{r_i} - 1 \bigg| = \max_i\frac{|1-r_i|}{r_i}
\le \frac{\delta_n^{\mathrm L}}{\min_i r_i}.
\ee
By Lemma~\ref{lem:wigner-rowcol-norms}, for sufficiently large $n$ we have $\Prob(\delta_n^{\mathrm L}>\tfrac12)\le Cn^{-10}$. Since $\sum_n Cn^{-10}<\infty$, the Borel-Cantelli lemma implies that $\delta_n^{\mathrm L}\le \tfrac12$ for all sufficiently large $n$ almost surely; hence
$\min_i r_i\ge 1-\tfrac12=\tfrac12$ and therefore
\be
\|D_n^{\mathrm L}-I\|_{\op}\le 2\,\delta_n^{\mathrm L}
=O\!\Big(\sqrt{\tfrac{\log n}{n}}\Big).
\ee
By Borel-Cantelli this holds almost surely eventually. Set $\Delta_n^{\mathrm L}\coloneq D_n^{\mathrm L}-I$. Consider 
\begin{equation}
M_n\coloneq (A_n^{\mathrm L})(A_n^{\mathrm L})^\dagger
=D_n^{\mathrm L}\,A_nA_n^\dagger\,D_n^{\mathrm L},\qquad
N_n\coloneq A_nA_n^\dagger.
\end{equation}
By Lemma~\ref{lem:cov-op},
\begin{equation}
\|M_n-N_n\|_{\op}\;\le\;\bigl(2\|\Delta_n^{\mathrm L}\|_{\op}+\|\Delta_n^{\mathrm L}\|_{\op}^2\bigr)\,\|N_n\|_{\op}.
\end{equation}
\cref{thm:wigner-law} implies $\|A_n\|_{\op}\to 2$ almost surely, hence $\|N_n\|_{\op}=\|A_n\|_{\op}^2\to 4$.
Therefore $\|M_n-N_n\|_{\op}\to 0$ almost surely. By Lemma~\ref{lem:esd-continuity} $M_n$ and $N_n$ have the same limiting ESD. Since the non-zero eigenvalues of $N_n$ (resp.\ $M_n$) are the squares of the singular values of $A_n$ (resp.\ $A_n^{\mathrm L}$), the singular-value ESDs of $A_n^{\mathrm L}$ and $A_n$ have the same limit. Finally, by \cref{thm:wigner-law}, the eigenvalue ESD of $A_n$ converges to semicircle on $[-2,2]$, hence its singular-value ESD converges to the quarter-circle on $[0,2]$ (since $A_n$ is Hermitian we have $A_nA_n^\dagger=A_n^2$, and then using Lemma~\ref{thm:pushforward}). 

The right-balanced case $A_n^{\mathrm R}$ is identical, replacing $D_n^{\mathrm L}$ by $D_n^{\mathrm R}$ and taking Euclidean column norms instead of row norms.
\end{proof}
\begin{proof}[\textbf{Proof of Theorem~\ref{thm:ui-wigner}}]\label{proof-of-lem2}
Let $a_{ij}=(A_n)_{ij}$ and define the row, column, and grand log-averages $m_i,n_j,\hat m$ as in \cref{eq:row-col-gra-mean-rep}. By Proposition~\ref{prop:gauge-fully} (symmetric scale),
\begin{equation}\label{eq:gauge-symmetric-log-wigner}
\log(D^{\mathrm{B_L}}_{n})_{ii}=-m_i+\tfrac12\hat{m},\qquad
\log(D^{\mathrm{B_R}}_{n})_{jj}=-n_j+\tfrac12\hat{m}.
\end{equation}
By Lemma~\ref{lem:log-mean-wigner},
\begin{equation}\label{eq:log-mean-wigner-bnd}
\max_i\Big|m_i-\Big(c_\beta-\tfrac12\log n\Big)\Big|,\quad
\max_j\Big|n_j-\Big(c_\beta-\tfrac12\log n\Big)\Big|,\quad
\Big|\hat{m}-\Big(c_\beta-\tfrac12\log n\Big)\Big|
\;=\;O_{\mathrm{a.s.}}\!\Big(\sqrt{\tfrac{\log n}{n}}\Big).
\end{equation}
Define deviations $\delta_i,\epsilon_j,\bar\delta$ exactly as in \cref{eq:83ref} with $c_\star$ replaced by $c_\beta$. Then the same computation as in \cref{eq:repeat-ui-stretch} to \cref{eq:DU-R-bnd} yields diagonal matrices $\Delta^{\mathrm{B_L}}_{n},\Delta^{\mathrm{B_R}}_{n}$ such that
\begin{equation}\label{eq:DU-exp-wigner}
D^{\mathrm{B_L}}_{n}=e^{-c_\beta/2}\,n^{1/4}\,(I+\Delta^{\mathrm{B_L}}_{n}),\qquad
D^{\mathrm{B_R}}_{n}=e^{-c_\beta/2}\,n^{1/4}\,(I+\Delta^{\mathrm{B_R}}_{n}),
\end{equation}
with
\begin{equation}\label{eq:DeltaU-op-wigner}
\|\Delta^{\mathrm{B_L}}_{n}\|_{\op}+\|\Delta^{\mathrm{B_R}}_{n}\|_{\op}
\;=\;O_{\mathrm{a.s.}}\!\Big(\sqrt{\tfrac{\log n}{n}}\Big).
\end{equation}
In particular, $\|\Delta^{\mathrm{B_L}}_{n}\|_{\op},\|\Delta^{\mathrm{B_R}}_{n}\|_{\op}\to 0$ almost surely. Thereafter we define normalised matrices $T_n\coloneq n^{-1/2}A^{\mathrm{B}}_n$ and then follow the steps and definitions from \cref{eq:97ref} to \cref{eq:105ref} exactly, again taking care to substitute $c_\star$ with $c_\beta$, which combined with \cref{thm:wigner-law} yields
\begin{equation}\label{eq:Hhat-pert-wigner}
\|H_n-\widehat H_n\|_{\op}\ \xrightarrow[n\to\infty]{\mathrm{a.s.}}\ 0.
\end{equation}
By Lemma~\ref{lem:esd-continuity}, $\mathrm{ESD}(H_n)$ and $\mathrm{ESD}(\widehat H_n)$ have the same limiting law. By \cref{thm:wigner-law}, the eigenvalue ESD of $A_n$ converges almost surely to the semicircle law on $[-2,2]$; therefore the eigenvalue ESD of $A_nA_n^\dagger$ (in this case $A_n^2$ because Hermitian) converges almost surely to the pushforward of semicircle under $x\mapsto x^2$, which is $\mathrm{MP}(1,1)$ on $[0,4]$. Hence $\mathrm{ESD}(\widehat H_n)$ converges almost surely to $\mathrm{MP}(1,1)$ scaled by $e^{-2c_\beta}$. Finally, by Lemma~\ref{thm:pushforward}, the empirical singular-value measures of $T_n$ converge almost surely to the quarter-circle law on $[0,2e^{-c_\beta}]$. Moreover, taking operator norms in the analogous equation of \cref{eq:97ref} and using $\|A_n\|_{\op}\to 2$ and $\|\Delta^{\mathrm{B_L}}_{n}\|_{\op},\|\Delta^{\mathrm{B_R}}_{n}\|_{\op}\to 0$ yields
\be
\|A_n^{\mathrm B}\|_{\op}=2e^{-c_\beta}\sqrt n\,\big(1+o(1)\big)
\qquad\text{almost surely},
\ee
which implies the claimed stretched support.
\end{proof}
\newpage
\bibliography{Refs}
\bibliographystyle{JHEP}
\end{document}